\documentclass[11pt,a4paper]{article}
\usepackage[a4paper,hmargin=2.8cm,vmargin=3cm]{geometry}
\usepackage[utf8x]{inputenc}
\usepackage{amsmath,amsthm,amsfonts,amssymb}
\usepackage{authblk}
\usepackage[bookmarksnumbered=true]{hyperref}

\usepackage{color}
\usepackage{appendix}

\usepackage{graphicx}

\numberwithin{equation}{section}

\newtheorem{theorem}{Theorem}[section]

\newtheorem{lemma}[theorem]{Lemma}
\newtheorem{proposition}[theorem]{Proposition}
\theoremstyle{definition}

\def\wt{\widetilde}
\def\cE{\mathcal{E}}
\def\cU{\mathcal{U}}
\def\cD{\mathcal{D}}
\def\cM{\mathcal{M}}
\def\cB{\mathcal{B}}
\def\bR{\mathbb{R}}
\def\bN{\mathbb{N}}
\def\eps{\varepsilon}

\let\e=\varepsilon

\let\.=\cdot
\let\0=\emptyset

\def\tr{\text{\rm tr}}

\def\square{\hbox{$\sqcap\kern-7pt\sqcup$}}

\def\be{\begin{equation}}
\def\ee{\end{equation}}
\def\bea{\begin{eqnarray}}
\def\eea{\end{eqnarray}}

\def\o{\omega}

\title{From the Hartree dynamics to the Vlasov equation}
\author{Niels Benedikter, Marcello Porta, Chiara Saffirio 
and Benjamin Schlein}

\def\adresse{
\begin{description}
\item[N. Benedikter:] 
Department of Mathematical Sciences,\\ University of Copenhagen,
Universitetsparken 5,
2100 Copenhagen,
Denmark\\
E-mail: \texttt{niels.benedikter@math.ku.dk }

\item[M. Porta:] Institute of Mathematics, \\ University of Z\"urich, Winterthurerstrasse 190, 8057 Z\"urich, Switzerland\\
E-mail: \texttt{marcello.porta@math.uzh.ch}

\item[C. Saffirio:] Institute of Mathematics, \\ University of Z\"urich, Winterthurerstrasse 190, 8057 Z\"urich, Switzerland\\
E-mail: \texttt{chiara.saffirio@math.uzh.ch}

\item[B. Schlein:] Institute of Mathematics, \\ University of Z\"urich, Winterthurerstrasse 190, 8057 Z\"urich, Switzerland\\
E-mail: \texttt{benjamin.schlein@math.uzh.ch}
\end{description}
}

\date{\today}

\begin{document}

\maketitle

\begin{abstract}
We consider the evolution of quasi-free states describing $N$ fermions in the mean field limit, as governed 
by the nonlinear Hartree equation. In the limit of large $N$, we study the convergence towards the classical Vlasov equation. For a class of regular interaction potentials, we establish precise bounds on the rate of convergence.    
\end{abstract}


\section{Introduction and main results}
\setcounter{equation}{0}

This work is motivated by the study of the time-evolution of systems of $N$ fermions in the mean field regime, characterized by a large number of weak collisions. The many body evolution of $N$ fermions is generated by the Hamilton operator 
\begin{equation}\label{eq:ham} H_N = \sum_{j=1}^N -\Delta_{x_j} + \lambda \sum_{i<j}^N V(x_i - x_j)
\end{equation}
acting on 
\[ L^2_a (\bR^{3N}) = \{ \psi \in L^2 (\bR^{3N}) : \psi (x_{\pi 1}, \dots , x_{\pi N}) = \sigma_\pi \psi (x_1, \dots , x_N) \text{ for all } \pi \in S_N \}, \]
the subspace of permutation antisymmetric functions 
in $L^2 (\bR^{3N})$ ($\sigma_\pi$ denotes here the sign of the permutation $\pi$). Due to the antisymmetry, the kinetic energy in (\ref{eq:ham}) is typically (for data occupying a volume of order one) of the order $N^{5/3}$ (for bosons, particles described by permutation symmetric wave functions, it is much smaller, of order $N$). Hence, to obtain a non-trivial competition between kinetic and potential energy, we have to choose $\lambda = N^{-1/3}$. Moreover, the large kinetic energy of the particles implies that we can only follow their time evolution for short times, of the order $N^{-1/3}$ (the kinetic energy per particle is proportional to $N^{2/3}$; the typical velocity of the particles is therefore of the order $N^{1/3}$). After rescaling time, the evolution of the $N$ fermions is governed by the many body Schr\"odinger equation
\begin{equation}\label{eq:schro} i N^{1/3} \partial_t \psi_{N,t} = \left[ \sum_{j=1}^N -\Delta_{x_j} + \frac{1}{N^{1/3}} \sum_{i<j}^N V (x_i -x_j) \right] \psi_{N,t} 
\end{equation}
for $\psi_{N,t} \in L^2_a (\bR^{3N})$. It is convenient to rewrite (\ref{eq:schro}) as follows. We introduce the small parameter \[ \eps = N^{-1/3} \] and we multiply (\ref{eq:schro}) by $\eps^2$. We obtain
\begin{equation}\label{eq:schro2} i\eps \partial_t \psi_{N,t} = \left[ \sum_{j=1}^N - \eps^2 \Delta_{x_j} + \frac{1}{N} \sum_{i<j}^N V(x_i -x_j) \right] \psi_{N,t}  \,. 
\end{equation}
Hence, the mean field scaling for fermionic systems (characterized by the $N^{-1}$ factor in front of the potential energy) is naturally linked with a semiclassical scaling, where $\eps = N^{-1/3}$ plays the role of Planck's constant. Notice that for particles in $d$ dimensions, similar arguments show that we would have to take $\eps = N^{-1/d}$; in fact, our analysis applies to general dimensions (with appropriate changes on the regularity assumptions); to simplify our presentation we will only discuss the case $d=3$. 

{F}rom the point of view of physics, we are interested in understanding the evolution of the fermionic system resulting from a change of the external fields. In other words, we are interested in the solution of (\ref{eq:schro2}) for initial data describing equilibrium states of trapped systems. It is expected (and in certain cases, it is even known) that equilibrium states in the mean-field regime are approximately quasi-free. 

At zero temperature, the relevant quasi-free states are Slater determinants, having the form
\[ \psi_\text{Slater} (x_1, \dots , x_N) =\frac{1}{\sqrt{N!}} \det \, (f_j (x_i))_{1 \leq i,j \leq N} \]
where $\{ f_j\}_{j=1}^N$ is an orthonormal system in $L^2 (\bR^3)$. Slater determinants are completely characterized by their one-particle reduced density $\omega_N$, defined as the non-negative trace class operator over $L^2 (\bR^3)$ with the integral kernel 
\[ \omega_N (x;y) = N \int dx_2 \dots  dx_N \, \psi_\text{Slater} (x, x_2, \dots , x_N) \overline{\psi_\text{Slater} (y, x_2, \dots , x_N)}\;. \]
A simple computation shows that 
\[ \omega_N = \sum_{j=1}^N |f_j \rangle \langle f_j |\;, \]
i.e. $\omega_{N}$ is the orthogonal projection onto the $N$-dimensional space spanned by the $N$ orbitals $f_1, \dots , f_N$ defining $\psi_\text{Slater}$ (we used here the notation $|f \rangle \langle f |$ to indicate the orthogonal projection onto $f \in L^2 (\bR^3)$). In the language of probability theory, the one-particle reduced density corresponds to the one-particle marginal distribution, obtained by integrating out the degrees of freedom of the other $(N-1)$ particles. Slater determinants have the properties that higher order marginals can all be expressed in terms of $\omega_N$ via the Wick rule (this is, in fact, the defining property of quasi-free states). 

The many-body evolution of a Slater determinant, as determined by (\ref{eq:schro2}), is not a Slater determinant. Still, because of the mean-field form of the interaction, we can expect it to remain close, in an appropriate sense, to a Slater determinant. Under this assumption, it is easy to find a self-consistent equation for the dynamics of the Slater determinant. We obtain the nonlinear Hartree-Fock equation 
\begin{equation}\label{eq:HF0}
i\eps \partial_t \omega_{N,t} = \left[ -\eps^2 \Delta + (V* \rho_t) - X_t , \omega_{N,t} \right]\,.
 \end{equation}
Here $\rho_t (x) = N^{-1} \omega_N (x;x)$ is the normalized density of particles at $x \in \bR^3$, the exchange operator $X_t$ has the integral kernel $X_t (x;y) = N^{-1} V(x-y) \omega_{N,t} (x;y)$, and, as before, $\eps = N^{-1/3}$. It is easy to check that, if $\omega_{N,t=0}$ is an orthogonal projection with rank $N$, then the same is true for the solution $\omega_{N,t}$; in other words, the Hartree-Fock evolution of a Slater determinant is again a Slater determinant. 

In \cite{BPS}, it was shown that indeed, for sufficiently regular interaction potentials, the many body Schr\"odinger evolution of initial Slater determinants can be approximated by the Hartree-Fock evolution, in the sense that the one-particle reduced density associated with the solution $\psi_{N,t}$ of (\ref{eq:schro2}) remains close (in the Hilbert-Schmidt and in the trace norm) to the solution $\omega_{N,t}$ of the Hartree-Fock equation (\ref{eq:HF0}). Previous results in this direction have been obtained in \cite{EESY}; convergence towards the Hartree-Fock dynamics in other regimes, which do not involve a semiclassical limit, has been also established in \cite{BGGM,FK,PP,BBPPT}. 

At positive temperature, on the other hand, relevant quasi-free states approximating equilibria of trapped systems are mixed states, described by a one-particle reduced density $\omega_N$ with $\tr\, \omega_N = N$ and $0 \leq \omega_N \leq 1$ (it follows from the Shale-Stinespring condition, see e. g. \cite[Theorem 9.5]{Sol}, that every such $\omega_N$ is the one-particle reduced density of a quasi-free state with $N$ particles;  Slater determinants form a special case, with $\omega_N$ having only the eigenvalues $0$ and $1$).  In the simple case of $N$ fermions with one-particle Hamiltonian $h = -\eps^2 \Delta + V_\text{ext}$ and no interaction, equilibrium at temperature $T > 0$ is described by the Gibbs state with one-particle reduced density 
\begin{equation}\label{eq:gibbs} \omega_N = \frac{1}{1+ e^{\frac{1}{T} (-\eps^2 \Delta + V_\text{ext} - \mu)}} 
\end{equation}
where the chemical potential $\mu \in \bR$ has to be chosen so that $\tr \, \omega_N = N$. If we turn on a mean-field interaction, it is expected that equilibrium states continue to be approximated by quasi-free states with one-particle reduced density of the form (\ref{eq:gibbs}), with the external potential $V_\text{ext}$ appropriately modified to take into account, in a self-consistent manner, the interaction among the particles (for results in this direction see, for example, \cite{NP1,P}). 

In suitable scaling regimes, the state of the system at positive temperature is expected to be well approximated by an appropriate mixed quasi-free state. Similarly as in the case of Slater determinants, mixed quasi-free states are completely characterized by their one-particle reduced density. All higher order correlation functions (i.e. all higher order marginals) can be expressed in terms of $\omega_N$ \footnote{In general quasi-free states are characterized by two operators on $L^2 (\bR^3)$, a one-particle reduced density $\omega_N$ and a pairing density $\alpha$. Here we restrict our attention to states with $\alpha = 0$; this is expected to be a very good approximation for equilibrium states of fermions in the mean field regime considered here.}. For the evolution of mixed quasi-free states, we find the same self-consistent equation (\ref{eq:HF0}) derived for Slater determinants. We observe here that the properties $\tr \, \omega_N = N$ and $0 \leq \omega_N \leq 1$, characterizing the reduced one-particle density of mixed quasi-free states, are preserved by the Hartree-Fock equation (\ref{eq:HF0}). In \cite{BJPSS}, it was shown that, for sufficiently regular potential, the many-body evolution of a mixed quasi-free state can be approximated by the self-consistent Hartree-Fock equation (\ref{eq:HF0}) (also here, the convergence has been established through bounds on the distance between 
reduced densities). 

To summarize, it follows from the analysis of  \cite{BPS,BJPSS} that the many-body evolution of fermionic quasi-free states can be approximated by the Hartree-Fock equation (\ref{eq:HF0}). This holds true for Slater determinants (in this case $\omega_{N,t}$ is an orthogonal projection with rank $N$) as well as for general mixed quasi-free states (satisfying only $\tr\, \omega_{N,t} = N$ and the bounds $0 \leq \omega_{N,t} \leq 1$).

In the mean field regime, the energy contribution associated with the exchange term can be estimated as follows, for bounded potentials $V$:
\be
\Big| \frac{1}{2N} \int dxdy\, V(x-y) |\omega(x;y)|^{2} \Big| \leq \frac{\| V\|_{\infty}}{2N}\| \omega_{N} \|_{\text{HS}}^{2} \leq C\;,
\ee
where the full energy is of order $N$ (here we used that the Hilbert-Schmidt norm\footnote{The Hilbert-Schmidt norm of a compact operator $A$ is defined as $\| A \|_\text{HS}^2 = \tr A^* A$.} of $\omega_{N}$ is bounded by $N^{1/2}$). Because of the smallness of the exchange term, instead of considering the Hartree-Fock equation (\ref{eq:HF0}), we will drop the exchange term and study the fermionic Hartree dynamics, governed by the nonlinear equation 
\begin{equation}\label{eq:hartree} i\eps \partial_t \omega_{N,t} = \left[ -\eps^2 \Delta + (V*\rho_t), \omega_{N,t} \right] 
\end{equation}
with $\rho_t (x) = N^{-1} \omega_{N,t} (x;x)$ (a proof of the fact that the exchange term does not affect the dynamics can be found in Appendix A of \cite{BPS}).

The Hartree equation (\ref{eq:hartree}) still depends on $N$ (recall the choice $\eps = N^{-1/3}$ and the normalization $\tr\,\omega_{N} = N$). It is therefore natural to ask what happens to it in the limit $N \to \infty$. To answer this question, we define the Wigner transform of the one-particle reduced density $\omega_{N,t}$ by setting 
\begin{equation}\label{eq:WT} W_{N,t} (x,v) = \left(\frac{\eps}{2\pi}\right)^3 \int \omega_{N,t} \left( x+ \frac{\eps y}{2}; x - \frac{\eps y}{2} \right) e^{-i v\cdot y} dy\;. \end{equation} 
Hence, $W_{N,t}$ is a function of position and velocity, defined on the phase-space $\bR^3 \times \bR^3$. It is normalized so that \[ \int W_{N,t} (x,v) dx dv = \eps^3 \tr \, \omega_{N,t} = 1\;. \] 
The Wigner transform can be inverted, noticing that
\begin{equation}\label{eq:weyl} \omega_{N,t} (x;y) = N \int dv \, W_{N,t} \Big( \frac{x+y}{2} , v \Big) e^{iv\cdot \frac{x-y}{\eps}}\;. 
\end{equation}
Eq.\ (\ref{eq:weyl}) is known as the Weyl quantization of the function $W_{N,t}$. Notice that \ $\|\omega_{N,t}\|_{\rm HS}=\sqrt{N}\|W_{N,t}\|_2$.

The Wigner transform $W_{N,t}$ can be used to compute expectations in the quasi-free state described by $\omega_{N,t}$ of observables depending only on the position $x$ or on the momentum $-i\eps\nabla$ of the particles. In fact, for a large class of functions $f$ on $\bR^3$, 
\[ \tr \, f(x) \, \omega_{N,t} = \int dx f(x) \omega_{N,t} (x;x) = N \int dv dx f(x) W_{N,t} (x,v) \]
and 
\[ \tr \, f(i\eps \nabla) \, \omega_{N,t} = N \int dx dv \, f(v) W_{N,t} (x,v)\,.  \]
In other words, $\int dv \, W_{N,t} (x,v)$ is the density of fermions in position space at point $x \in \bR^3$, while $\int dx\, W_{N,t} (x,v)$ is the density of particles with velocity $v \in \bR^3$. Notice, however, that $W_{N,t}$ is not a probability density on the phase-space, because in general it is not positive (this observation is related with the Heisenberg principle; position and momentum of the particles cannot be measured simultaneously with arbitrary precision). 

{F}rom (\ref{eq:hartree}), we find an evolution equation for the Wigner transform $W_{N,t}$:
\[ \begin{split} 
i\eps \partial_t W_{N,t} &(x,v) =\;  \frac{1}{(2\pi)^3} \int dy \, i\eps\partial_t \omega_{N,t} \left( x+ \frac{\eps y}{2} ; x - \frac{\eps y}{2} \right) e^{-iv \cdot y} \\ = \; &\frac{\eps^2}{(2\pi)^3} \int dy \, (-\Delta_{x+\eps y/2} + \Delta_{x-\eps y/2})\, \omega_{N,t} \left( x+ \frac{\eps y}{2} ; x - \frac{\eps y}{2} \right) e^{-iv \cdot y} \\ &+\frac{1}{(2\pi)^3} \int dy \, \left( (V*\rho_t) (x+\eps y/2) - (V*\rho_t) (x-\eps y/2) \right) \omega_{N,t} \left( x+ \frac{\eps y}{2} ; x - \frac{\eps y}{2} \right) e^{-iv \cdot y}\;. \end{split} \]
Using $-\Delta_{x+\eps y/2} + \Delta_{x-\eps y/2} = -2/\eps \nabla_x \cdot \nabla_y$ and expanding
\[ (V*\rho_t) (x+\eps y/2) - (V*\rho_t) (x-\eps y/2)  \simeq \eps y \cdot \nabla (V*\rho_t) + O(\eps^2) \]
we conclude, formally, that  
\[ \begin{split} 
i\eps \partial_t W_{N,t} (x,v) = \; &-2\eps \frac{1}{(2\pi)^3} \nabla_x \cdot \int dy \, \nabla_y  \omega_{N,t} \left( x+ \frac{\eps y}{2} ; x - \frac{\eps y}{2} \right) e^{-iv \cdot y} dy\\ &+\eps \nabla (V*\rho_t) (x) \frac{1}{(2\pi)^3} \int dy \, y \, \omega_{N,t} \left( x+ \frac{\eps y}{2} ; x - \frac{\eps y}{2} \right) e^{-iv \cdot y} dy + O (\eps^2)\\ \; &= -2i\eps v \cdot \nabla_x W_{N,t} (x,v) +i\eps \nabla (V*\rho_t) (x) \cdot \nabla_v W_{N,t} (x,v) + O(\eps^2)\;. 
\end{split} \]
As a consequence, we expect that, in the limit $N \to \infty$ (and hence $\eps \to 0$; recall that $\eps = N^{-1/3}$), $W_{N,t}$ approaches a solution $W_{t}$ of the classical Vlasov equation 
\begin{equation}\label{eq:vlasov} \partial_t W_{t} + 2 v \cdot \nabla_x W_{t} = \nabla (V*\varrho_{t}) \cdot \nabla_v W_{t} 
\end{equation}
with the density $\varrho_{t}  (x) = \int W_{t} (x,v) dv$ (in contrast with $W_{N,t}$, the limit $W_{t}$ is a probability density, if this is true at time $t=0$). The goal of this paper is to study the convergence of the Hartree dynamics towards the Vlasov equation (\ref{eq:vlasov}), in the limit $N \to\infty$. 

This work is not the first one devoted to the derivation of the Vlasov equation (\ref{eq:vlasov}) from quantum evolution equations. In \cite{NS,Sp}, the Vlasov equation is obtained directly from many body quantum dynamics, starting from the fundamental $N$-fermion Schr\"odinger equation (the Vlasov equation also emerges in the $N$-boson case, if the mean field limit is combined with a semiclassical limit; see \cite{GMP}, where the dynamics of factored WKB states is analyzed). In \cite{LP,MM}, the authors take the Hartree equation (\ref{eq:hartree}) as starting point of their analysis, and they prove convergence (in a weak sense) towards the solution of the Vlasov equation (\ref{eq:vlasov}). Note that the analysis of \cite{LP,MM} also applies to singular interactions, including a Coulomb potential (the analysis was extended to the Hartree-Fock equation in \cite{GIMS}).

In \cite{NS,Sp,LP,MM,GIMS}, the convergence towards the classical Vlasov dynamics is established in an abstract sense, without control on its rate. The problem of determining bounds on the rate of convergence is not only of academic interest. When considering applications to real physical systems, the number of particles $N$ is large but, of course, finite. Bounds on the rate of convergence are therefore important to decide whether $N$ is large enough for the Vlasov equation to be a good approximation of the Hartree and of the full many body Schr\"odinger dynamics. 


Bounds on the rate of convergence of the Hartree evolution towards the Vlasov equation have been first obtained in \cite{APPP}. In this paper the authors obtain the convergence in the Hilbert-Schmidt with a relative rate $\eps^{2/7} = N^{-2/21}$ for sufficiently regular initial data and potentials (they require $V \in H^1 (\bR^3)$ and that $\widehat{V} \in L^1 (\bR^3, (1+|p|^4) dp)$). For smooth potentials, an expansion of the solution of the Hartree equation (\ref{eq:hartree}) in powers of $\eps$ has been shown in \cite{PP09} (with no control on the remainder) and in \cite{AKN,AKN2}. 

Our approach here is similar to the one of \cite{APPP}; we consider the solution of the Hartree equation (\ref{eq:hartree}) for initial data $\omega_N$ with sufficiently smooth Wigner transform $W_N$, and we compare it with the Weyl quantization of the solution of the Vlasov equation (\ref{eq:vlasov}), with initial data $W_N$.
We consider regular interaction potentials. In Theorem \ref{thm:main} and in Theorem \ref{cor:relaxed-main} we establish bounds on the norm-distance of the solution of the Hartree equation $\omega_{N,t}$ with initial data $\omega_N$ and the Weyl quantization $\wt{\omega}_{N,t}$ of the solution of the Vlasov equation with initial data $W_{N}$. For every fixed $t \in \bR$, the relative error is of the order $\eps = N^{-1/3}$ in the limit of large $N$. The dependence on $N$ of these bounds is expected to be optimal. This expectation is confirmed by the expansion of \cite{AKN}, where the next order corrections are constructed (in fact, if we assumed initial data with smooth Wigner transform $W_N \in W^{\infty,\infty} (\bR^3 \times \bR^3)$ and smooth interaction potential $V \in W^{\infty,\infty} (\bR^3)$, the result of Theorem \ref{thm:main} would follow from Theorem 1.2 in \cite{AKN}). 

In Theorem \ref{thm:main}, we get convergence in the trace-norm, for very regular initial data. In Theorem \ref{cor:relaxed-main}, we bound the Hilbert-Schmidt norm, under weaker assumptions on the regularity of $W_N$. The strategy to show Theorem \ref{cor:relaxed-main} is similar to the one of \cite{APPP}; we regularize the initial data, we compare the solutions of the regularized Hartree and Vlasov equations and then we establish stability of both equations with respect to  the regularization. We can improve the bounds of \cite{APPP} by using the trace norm convergence shown in Theorem \ref{thm:main}
for the solutions with regularized data. The nonlinearity in the Hartree and in the Vlasov equation depends on the convolution of the potential with the density of particles in space. Differences among densities can be easily controlled through the trace-norm of the corresponding fermionic operators (which are bounded in Theorem \ref{thm:main}). Estimating them directly by means of Hilbert-Schmidt norms, as done in \cite{APPP}, leads instead to a deterioration of the rate of convergence. 

Notice that, in Theorems \ref{thm:main} and \ref{cor:relaxed-main}, we consider the solution of the Vlasov equation for initial data which are not probability densities. The well-posedness of the Vlasov equation for such initial data can be obtained adapting the arguments of \cite{Dobr79}; in Appendix \ref{app:dobr} we sketch the proof. 

If we assume additionally that the sequence of initial data $\omega_N$ has a limit, in the sense that their Wigner transform $W_N$ converge towards a probability density $W_0$, then we can also establish the convergence of the Wigner transform $W_{N,t}$ of the solution of the Hartree equation towards the solution of the Vlasov equation $W_t$ with initial data $W_0$ (in this case, the solution of the Vlasov equation is a classical probability density, for all $t \in \bR$). This is the content of Theorem \ref{thm:new}. 

Our bounds on the norms of the distance between the Wigner transform $W_{N,t}$ and the solution of the Vlasov equation $W_t$ (as well as the bounds for the distance between $W_{N,t}$ and the Weyl quantization $\wt{W}_{N,t}$ of the solution of the Vlasov equation with initial data $W_N$) hold for sufficiently regular initial data. In particular, Theorem \ref{cor:relaxed-main} needs $W_N \in H^2 (\bR^3\times \bR^{3})$ (with some additional weights; see Theorem \ref{thm:main2} for the precise assumptions). This condition is justified for initial data describing equilibrium states of confined fermionic system at positive temperatures. At zero temperature, on the other hand, the system at equilibrium relaxes to its ground state, which can be approximated by a Slater determinant. Typically, in this case, the corresponding Wigner transform is not regular. For example, the ground state of a system of $N$ free fermions in a periodic box with volume one is a Slater determinant with Wigner transform
\begin{equation}\label{eq:WN-free} W_N (x,v) = N^{-1} {\bf 1} (|v| \leq c \rho^{1/3}) 
\end{equation}
where $\rho = N$ is the density of the particles (this system is translation invariant; therefore, particles are uniformly distributed in the box). Eq.\ (\ref{eq:WN-free}) corresponds to the idea that to construct the free ground state, we should fill the $N$ one-particle states with the smallest possible energy (by the antisymmetry of fermionic wave functions, there cannot be two particles in the same state). If we switch on an external potential and a mean-field interaction, it is believed that the ground state can still
be approximated by a state with Wigner transform of the form  (\ref{eq:WN-free}); the only difference is that now we have to fill low energy states locally, according to an effective particle density $\rho_\text{TF}$ that can be determined by minimizing the Thomas-Fermi functional 
\[ \cE_\text{TF} (\rho) = \frac{3}{5} c_\text{TF} \int dx \, \rho^{5/3} (x) +\int dx \, V_\text{ext} (x) \rho (x) + \frac{1}{2} \int dx dy \, V(x-y) \rho (x) \rho (y) \]
among all $\rho \in L^1 \cap L^{5/3} (\bR^3)$ with $\| \rho \|_1 = N$.  The resulting sequence of Wigner transforms $W_N (x,v) = N^{-1} {\bf 1} (|v| \leq c \rho_\text{TF}^{1/3} (x))$ is not in $H^2 (\bR^3 \times \bR^3)$. So, while Theorem \ref{thm:main}, Theorem \ref{cor:relaxed-main} and Theorem \ref{thm:new} provide a good description of the fermionic dynamics in the mean field limit at positive temperature, they cannot be applied at zero temperature. 

For such initial data, we do not get norm convergence towards the solution of the Vlasov equation. Nevertheless, in Theorem \ref{thm:main2} and Theorem \ref{thm:new2} we can still prove convergence for the expectation of a class of semiclassical observables. Semiclassical observables are functions of the multiplication operator $x$ and of the momentum operator $-i\eps \nabla$; they detect variations in the spatial distribution of the particles on ``macroscopic'' scales of order one and, at the same time, they are sensitive to variations of order $\eps^{-1}$ in the momentum distribution  (corresponding to the ``microscopic'' length scale $\eps$). 

Let us stress the fact that, to the best of our knowledge, Theorem \ref{thm:main2} and Theorem \ref{thm:new2} are the first rigorous results concerning convergence from the Hartree dynamics towards the Vlasov equation that can be applied to reasonable approximations of ground states. 

In Section \ref{sec:statements}, in the remarks following our main theorems, we provide explicit examples of fermionic states, constructed with the help of coherent states, approximating ground states and positive temperature equilibrium states of fermionic systems in the mean-field regime, to which our theorems can be applied.

\section{Statement of the results}
\setcounter{equation}{0}
\label{sec:statements}

In order to state our results in a precise form, we need to introduce some norms for functions on the phase space $(x,v) \in \bR^3 \times \bR^3$. For $s \in \bN$, we define the Sobolev norm 
\[ \| f \|^2_{H^s} = \sum_{|\beta| \leq s} \int |\nabla^\beta f(x,v)|^2 dx dv \]
where $\beta$ is a multi-index, and $\nabla^\beta$ can act on both position and momentum variables. For $s,a \in \bN$, we introduce also the weighted norms 
\[ \| f \|^2_{H^s_{a}} = \sum_{|\beta| \leq s} \int (1+x^2+v^2)^{a} |\nabla^\beta f (x,v)|^2 dx dv \]

We are now ready to state our main theorems. In the first theorem, we assume strong regularity of the initial data, and we prove bounds in the trace-norm. 

\begin{theorem}\label{thm:main}
Let $V \in W^{2,\infty} (\bR^3)$. Let $\omega_N$ be a sequence of reduced densities on $L^2 (\bR^3)$, with $\tr \, \omega_N = N$, $0 \leq \omega_N \leq 1$ and with Wigner transform $W_N$ satisfying $\| W_N \|_{H^5_4} \leq C$, uniformly in $N$. 

We denote by $\omega_{N,t}$ the solution of the Hartree equation 
\begin{equation}\label{eq:hartree1} i\partial_t \omega_{N,t} = \left[ -\eps^2 \Delta + (V * \rho_t) , \omega_{N,t} \right] 
\end{equation}
with $\rho_t (x) = N^{-1} \omega_{N,t} (x;x)$ and initial data $\omega_N$. 

On the other hand, we denote by $\wt{W}_{N,t}$ the solution of the Vlasov equation 
\begin{equation}\label{eq:vlasov1} \partial_t \wt{W}_{N,t} + 2 v \cdot \nabla_x \wt{W}_{N,t} = \nabla (V * \wt{\rho}_t) \cdot \nabla_v \wt{W}_{N,t} \end{equation} 
with $\wt{\rho}_{t} (x) = \int dv\, \wt{W}_{N,t} (x,v)$ and with initial data $\wt{W}_{N,0} = W_N$. Moreover, let $\wt{\omega}_{N,t}$ be the Weyl quantization of $\wt{W}_{N,t}$, defined as in (\ref{eq:weyl}).

Then there exists a constant $C>0$ (depending on $\| V \|_{W^{2,\infty}}$ and on $\sup_N \, \| W_N \|_{H^2_4}$, but not on the higher Sobolev norms of $W_N$) such that 
\be\label{eq:main}
\tr \, \left| \o_{N,t}-\wt\o_{N,t} \right| \leq C N \e \, \exp (C \exp (C|t|)) \Big[ 1 + \sum_{k=1}^{3} \e^{k} \sup_N \, \|W_N \|_{H^{k+2}_{4}} \Big]\, .
\ee
\end{theorem}
{\it Remarks.} 
\begin{enumerate}
\item[1)] Recall that we use the normalization $\tr \omega_{N,t} = N$. In this sense, (\ref{eq:main}) shows that $\omega_{N,t}$ and $\wt{\omega}_{N,t}$ are close, in the limit of large $N$, since their difference is smaller, by a factor $\eps = N^{-1/3}$, than their trace norms. 
\item[2)] The assumption $\| W_N \|_{H^5_4} \leq C$ on the Wigner transform of the initial data is equivalent to suitable commutator estimates for the initial fermionic reduced density $\omega_{N}$ with the differential operator $\nabla$ and the multiplication operator $x$. We begin by noticing that
\be \label{eq:commgrad2}
\begin{split}
\| \nabla_{x} W_{N} \|^{2}_{2} &= \int dxdv\, |\nabla_{x} W_{N}(x,v)|^{2}\\
&= \int dxdv\, \Big| \frac{\eps^{3}}{(2\pi)^{3}}\int dy\, e^{-iv\cdot y} [\nabla,\, \omega_{N}](x + \eps y/2, x - \eps y/2) \Big|^{2}\\
&= N^{-1} \| [\nabla, \omega_{N}] \|_{\rm HS}^{2}\;.
\end{split}
\ee
Similarly, we find $\| \nabla_{v} W_{N} \|_{2}^{2} = N^{-1}\eps^{-2} \| [ x, \omega_{N}] \|_{\rm HS}^{2}$. As for the weights in the definition of the $H_a^s$-norms of $W_N$, we notice that
\[
\| (1 + x^{2} + v^{2})^{a/2} W_{N} \|_{2}^{2}\leq C N^{-1}\| (1 +  x^{2} - \e^{2}\Delta)^{a/2} \omega_{N} \|_{\rm HS}^{2}\;,
\]
for some $N$-independent constant $C>0$. 

Proceeding analogously, one can show that the estimate $\| W_{N} \|_{H^{5}_{4}}\leq C$ follows from the bounds 
\be\label{eq:manycomm}
N^{-1}\| (1 + x^{2} - \eps^{2}\Delta)^{a/2} [a_{1}, [a_{2}, [a_{3}, [a_{4}, [a_{5}, \omega_{N}]]]]] \|_{\rm HS}^{2}\leq C\;,
\ee
uniformly in $N$ and for all choices of $a_1, \dots,  a_5$ with either $a_{i} = x/\eps$ or $a_i = \nabla$.

Therefore the commutator structure allows to quantify the regularity and decay properties of the quantum state $W_N$. Estimates of commutators $[x,\omega_{N}]$ and $[\e\nabla,\omega_N]$ already played a key role in \cite{BJPSS,BPS}.

\item[3)] The estimate $\sup_N \| W_N \|_{H^5_4} \leq C$ or, equivalently, the bounds (\ref{eq:manycomm}), are expected to hold true for fermionic mixed states, describing systems of $N$ particles in equilibrium at positive temperature, in the mean-field regime, \cite{BJPSS}. A reasonable approximation for the reduced density of such a state is given by the superposition 
\be\label{eq:coherent}
\begin{split}
&\omega_{N}(x;y) = \int dpdr\, M(r,p) f_{pr}(x) \overline{f_{pr}(y)}\;,
\end{split}
\ee
of the coherent states 
\begin{equation}\label{eq:fpr} f_{pr}(x) = \eps^{-3/2} e^{-ip\cdot x/\eps} g(x - r) \end{equation} 
with a probability density $M$ with $0\leq M(r,p) \leq 1$ and \[ \int dp dr\, M(r,p) = 1 \]
In (\ref{eq:fpr}), the function $g$ is assumed to vary on the (possibly $N$-dependent) scale $\delta$ and to be normalized so that $\| g \|_{2} = 1$. For simplicity, we shall make the explicit choice
\be\label{eq:g-Gauss}
g(x) = \frac{1}{(2\pi\delta^2)^{3/4}}e^{-x^{2}/2\delta^2}\;.
\ee
It is simple to check that, with the definition (\ref{eq:coherent}), one indeed finds that $0 \leq \omega_N \leq 1$ and $\tr\, \omega_N = N$. 

The smoothness and decay properties of the Wigner transform $W_{N}$ of (\ref{eq:coherent}) follow
from analogous properties of the phase space density $M(r,p)$, i.e.
\be\label{eq:boundcoherent}
\| W_{N} \|_{H^{5}_{4}} \leq C\| M \|_{H^{5}_{4}}\; .
\ee

In fact, according to the previous remark, to prove (\ref{eq:boundcoherent}) it is enough to show (\ref{eq:manycomm}). To this end, we notice that 
\be
\begin{split}
[ x/\eps, \omega_{N}](x;y) &= \int dpdr\, M(p,r) (-i\nabla_{p}) f_{pr}(x) \overline{f_{pr}(y)} = \int dpdr\, (i\nabla_{p} M(p,r)) f_{pr}(x) \overline{f_{pr}(y)}\\
[\nabla, \omega_{N}](x;y) &= \int dpdr\, M(p,r) \nabla_{r} f_{pr}(x) \overline{f_{pr}(y)} = -\int dpdr\, (\nabla_{r} M(p,r)) f_{pr}(x) \overline{f_{pr}(y)}\;.
\end{split}
\ee
More generally, using integration by parts, all commutators of $\omega_{N}$ with $x/\eps$ and $\nabla$ can be written as superpositions of coherent states, weighted by derivatives of the phase space density. Therefore,  (\ref{eq:boundcoherent}) follows from 
\be\label{eq:gaussian}
|\langle f_{pr}, f_{p'r'}\rangle| = \Big|\int dx\, f_{pr}(x) \overline{f_{p'r'}}(x)\Big| = C N\exp\Big\{ -\frac{(r - r')^{2}}{4\delta^2} - \frac{\delta^2}{4\eps^{2}}(p - p')^{2} \Big\}\;,
\ee
for a constant $C>0$, independent of $N$ and $\delta$, and from the bound
\be
\begin{split}
\| [a_{1}, [a_{2}, \ldots , [a_{j}, \omega_{N}]\ldots ] \|_{\rm HS}^{2} &\leq \int dp dp'drdr'\, |\nabla^{\beta} M(p,r)| |\nabla^{\beta} M(p',r')| |\langle f_{p,r}, f_{p',r'} \rangle|^{2}\\ 
&\leq CN \| \nabla^{\beta} M \|_{2}\|\| \nabla^{\beta} M \|_{2} \leq CN \| M \|^{2}_{H^{j}_{0}}\;,
\end{split}
\ee
for an appropriate multi-index $\beta$ with $|\beta| = j$. The effect of the operators $(1+x^2 - \eps^2 \Delta)$ appearing in (\ref{eq:manycomm}) can be controlled using the decay of (\ref{eq:gaussian}) and of the probability density $M$. 

We conclude that, for any probability density $M \in H^{5}_4 (\bR^3 \times \bR^3)$ with $0 \leq M (r,p) \leq 1$ for all $r,p \in \bR^3$, the sequence of reduced densities (\ref{eq:coherent}) is an example of initial data satisfying the assumption of Theorem \ref{thm:main}. 
\end{enumerate}

In our second theorem, we relax partly the regularity assumption on the initial data. To reach this goal, we start from (\ref{eq:main}) and we apply an approximation argument. In contrast with Theorem \ref{thm:main}, here we only get bounds for the difference $\omega_{N,t} - \wt{\omega}_{N,t}$ in the Hilbert-Schmidt norm (the Hilbert-Schmidt norm of a reduced density is directly related with the $L^2$ norm of its Wigner transform; there is no such simple relation between the trace norm of a reduced density and the $L^1$-norm of its Wigner transform). 

\begin{theorem}\label{cor:relaxed-main}\label{thm:HS}
Let $V\in L^1 (\bR^3)$ be such that
\be\label{eq:pot-hyp}
\int |\widehat{V}(p)| (1+|p|^2) \, dp < \infty\,.
\ee
Let $\omega_N$ be a sequence of reduced densities on $L^2 (\bR^3)$, with $\tr \, \omega_N = N$, $0 \leq \omega_N \leq 1$ and with Wigner transform $W_N$ satisfying $\| W_N \|_{H^2_4} \leq C$, uniformly in $N$. 

As in Theorem \ref{thm:main}, we denote by $\omega_{N,t}$ the solution of the Hartree equation (\ref{eq:hartree1}) with initial data $\omega_N$ and by 
$\wt{\omega}_{N,t}$ the Weyl quantization of the solution $\wt{W}_{N,t}$ of the Vlasov equation (\ref{eq:vlasov1}) with initial data $\wt{W}_{N,0} = W_N$. Then, there exists a constant $C>0$ depending only on $\sup_N \| W_N \|_{H^2_4}$ and on the integral (\ref{eq:pot-hyp}) such that
\be\label{eq:HS-est}
\|\o_{N,t}-\wt\o_{N,t}\|_{\rm HS} \leq C \sqrt{N}\e\,  \exp (C \exp (C|t|)) \, .  
\ee
\end{theorem}

Instead of comparing the solution $\omega_{N,t}$ of the Hartree equation with the Weyl quantization $\wt{\omega}_{N,t}$ of the solution of the Vlasov equation $\wt{W}_{N,t}$, we can equivalently compare $\wt{W}_{N,t}$ with the Wigner transform $W_{N,t}$ of $\omega_{N,t}$. Eq. (\ref{eq:HS-est}) implies that 
\begin{equation}\label{eq:WWNt} \| W_{N,t} - \wt{W}_{N,t} \|_2 \leq C \eps \exp (C \exp (C|t|)) \end{equation}

If we assume that the fermionic initial data $\omega_N$ has a Wigner transform $W_N$ (with appropriately bounded $H^2_4$-norm) approaching, in the limit of large $N$, a probability density $W_0$ on the phase space, we can also compare the Wigner transform $W_{N,t}$ of the solution $\omega_{N,t}$ of the Hartree equation with the solution $W_t$ of the Vlasov equation with initial data $W_0$. In the next theorem, we show the $L^2$-convergence of $W_{N,t}$ towards $W_t$. 

\begin{theorem}\label{thm:new}
Let $V \in L^1 (\bR^3)$ be such that (\ref{eq:pot-hyp}) holds true. Let $\omega_N$ be a sequence of reduced densities on $L^2 (\bR^3)$, with $\tr \, \omega_N = N$, $0 \leq \omega_N \leq 1$ and with Wigner transform $W_N$ satisfying $\| W_N \|_{H^2_4} \leq C$, uniformly in $N$. 

Furthermore, let $W_0$ be a probability density on $\bR^3 \times \bR^3$ with $\| W_0 \|_{H^2_4} < \infty$ and such that 
\begin{equation}\label{eq:init-conv} \| W_N - W_0 \|_1 \leq C \kappa_{N,1}, \quad \text{and } \| W_N - W_0 \|_{2} \leq C \kappa_{N,2} 
\end{equation}
for sequences $\kappa_{N,1}, \kappa_{N,2} \geq 0$ with $\kappa_{N,j} \to 0$ as $N \to \infty$ for $j=1,2$. 

Let $\omega_{N,t}$ denote the solution of the Hartree equation (\ref{eq:hartree1}) with initial data $\omega_N$ and let $W_{N,t}$ be its Wigner transform. On the other hand, let $W_t$ denote the solution of the Vlasov equation (\ref{eq:vlasov1}), with initial data $W_0$. Then we have
\begin{equation}\label{eq:state-new} \| W_{N,t} - W_t \|_2 \leq C \eps \exp (C \exp (C |t|)) + C (\kappa_{N,1} + \kappa_{N,2}) \exp (C |t|) \end{equation}
\end{theorem}

{\it Remarks.} 
\begin{enumerate}
\item[1)] Notice that, if $\| W_N - W_0 \|_1 \leq \kappa_{N,1}$ for a sequence $\kappa_{N,1} \to 0$, and if $\| W_N \|_{H^2_4} , \| W_0 \|_{H^2_4} \leq C$ uniformly in $N$, then, automatically, $\| W_N - W_0 \|_2 \leq C \kappa_{N,1}^{1/2}$, i.e. the second condition in (\ref{eq:init-conv}) follows from the first one, if we take $\kappa_{N,2} = \kappa_{N,1}^{1/2}$. However, it is often possible to get a better estimate on $\kappa_{N,2}$, improving the bound (\ref{eq:state-new}) (for instance, in the example discussed in the next remark, we find $\kappa_{N,2} = \kappa_{N,1} = \eps^{1/2}$). 

\item[2)] An interesting example of sequence of initial data satisfying all assumptions of Theorem \ref{thm:new} can be constructed again by means of coherent states. As in (\ref{eq:coherent}), consider the fermionic reduced densities
\[ \omega_N (x;y) = \int dp dr M(r,p) f_{pr} (x) \overline{f}_{pr} (y) \]
with $f_{pr} (x) = \eps^{-3/2} e^{ip \cdot x/\eps} g(x-r)$ and with $M$ a probability density on the phase-space, with $0 \leq M (r,p) \leq 1$ and $\| M \|_1 = 1$ and such that $\| M \|_{H^2_4} < \infty$. For simplicity, we choose $g$ as in (\ref{eq:g-Gauss}) to be a Gaussian function, localized on the length scale $\delta = \delta (N)$, with $\delta (N) \to 0$ as $N \to \infty$. 

The Wigner transform of $\omega_N$, defined as in (\ref{eq:WT}), is given by 
\[ \begin{split} W_N (x,v) &= \frac{\eps^3}{(2\pi)^3} \int dy  \, \omega_N \left( x + \frac{\eps y}{2} ; x- \frac{\eps y}{2} \right) e^{i y \cdot v}  \\
&= \frac{1}{(2\pi)^3 \,  (2\pi \delta^2)^{3/2}} \int dy dr dp \, M(r,p) e^{iy \cdot (v-p)} e^{-\frac{(x-r + \eps y/2)^2}{2\delta^2}} e^{-\frac{(x-r-\eps y/2)^2}{2\delta^2}} \\ &= \frac{2^{3/2}}{(2\pi \eps)^3} \int dr dp \, M(r,p) \, e^{-\frac{(x-r)^2}{2\delta^2}} e^{- \frac{\delta^2 (p-v)^2}{\eps^2}} \end{split} \]
where, in the last step, we evaluated the integral over $y$. We find 
\[ \begin{split} \| W_N - M \|_1 \leq \; &\frac{2^{3/2}}{(2\pi)^{3}} \int dx dv dr dp \, e^{-\frac{r^2}{2}} \, e^{-p^2} \, \left| M \left( x+ \delta r, v + \eps p/ \delta \right) - M(x,v) \right|  \\
\leq \; &\frac{2^{3/2}}{(2\pi)^{3}} \int dx dv dr dp \, e^{-\frac{r^2}{2}} \, e^{-p^2} \\ &\hspace{.5cm} \times  \int_0^1 d\lambda \left[ \delta |r| \left| (\nabla_x M) \left( x+ \lambda \delta r , v + \lambda \eps p / \delta \right) \right|  \right. \\ &\hspace{5cm} \left. + \frac{\eps}{\delta} |p| \left| (\nabla_v M) \left( x+ \lambda \delta r , v + \lambda \eps p / \delta \right)\right| \right] \\ 
\leq \; &C \delta \| \nabla_x M \|_1 + C \frac{\eps}{\delta} \| \nabla_v M \|_1 \\ \leq \; &C \left[ \delta + \frac{\eps}{\delta} \right] \, \| M \|_{H^2_4} 
\end{split}\]
and similarly, 
\[ \| W_N - M \|_2 \leq C \left[ \delta + \frac{\eps}{\delta} \right] \, \| M \|_{H^2_4} \]
To optimize the rate of the convergence $W_N \to M$ (i.e. to make the sequence of initial data as ``classical'' as possible), we choose $\delta = \eps^{1/2}$ (recall that $\eps = N^{-1/3}$). {F}rom Theorem \ref{thm:new}, we conclude then that the distance between the Wigner transform $W_{N,t}$ of the solution of the Hartree equation and the solution $W_t$ of the Vlasov equation with initial data given by the probability density $W_0 = M$ is bounded by 
\[  \| W_{N,t} - W_{t} \|_2 \leq C \eps^{1/2} \exp (C \exp (C |t|)) \]
\end{enumerate}

Although in Theorem \ref{cor:relaxed-main} and in Theorem \ref{thm:new} the assumptions on $W_N$ are weaker than in Theorem~\ref{thm:main}, we still need $W_N \in H^2_{4} (\bR^3 \times \bR^3)$, with a norm bounded uniformly in $N$. As pointed out in the introduction, this assumption is typically satisfied for interesting initial data at positive temperature (like the ones constructed in the remarks after Theorem \ref{thm:new}), but it is not valid for Slater determinants approximating the ground state, which are relevant at zero temperature. 

In the next two theorems, we establish a weaker form of convergence for the solution of the Hartree equation towards 
the solution of the Vlasov equation. We prove convergence after testing against a semiclassical observable (whose kernel varies on the length-scale $\eps$ in the $(x-y)$ direction). The advantage of these two results, as compared with Theorems \ref{thm:main} and \ref{cor:relaxed-main}, is the fact that they require much weaker assumptions on the initial data; in particular, they can be applied to reasonable and physically interesting approximations of the ground state of confined systems (examples of such states are constructed in the remark after Theorem \ref{thm:new2}).

\begin{theorem}\label{thm:main2}
Let $V \in L^1 (\bR^3)$ be so that 
\begin{equation}\label{eq:ass-pot} \int |\widehat{V} (p)| (1+ |p|^3) dp < \infty 
\end{equation}
Let $\omega_N$ be a sequence of reduced densities on $L^2 (\bR^3)$, with $\tr \, \omega_N = N$, $0 \leq \omega_N \leq 1$, such that 
\begin{equation}\label{eq:sc-in} \tr \, |[ x, \omega_N]| \leq C N \eps, \quad \tr \, |[ \eps \nabla, \omega_N]| \leq C N \eps \end{equation}
Denote by $W_N \in L^1 (\bR^3 \times \bR^3)$ the Wigner transform of $\omega_N$. We assume that \[ \| W_N \|_{W^{1,1}} = \sum_{|\beta| \leq 1} \int dx dv |\nabla^\beta W_N (x,v)| \leq C \]
uniformly in $N$.

Let $\omega_{N,t}$ be the solution of the Hartree equation (\ref{eq:hartree1}) with initial data 
$\omega_N$. On the other hand, let $\wt{\omega}_{N,t}$ be the Weyl quantization of the solution $\wt{W}_{N,t}$ of the Vlasov equation (\ref{eq:vlasov1}) with initial data $W_N$. 

Then there exists a constant $C >0$, such that 
\begin{equation} \label{eq:main3} \left| \tr\, e^{ip \cdot x + q \cdot \eps \nabla} \, \left( \omega_{N,t} - \wt{\omega}_{N,t} \right) \right| \leq C N \eps (1+|p|+|q|)^2 e^{C|t|} 
\end{equation}
for all $p,q \in \bR^3$, $t \in \bR$. 

\end{theorem}

Notice that the expectation of the observable appearing in (\ref{eq:main3}) can also be expressed in terms of Wigner transforms. In fact, for any fermionic operator $\omega_N$, we find 
\[ \begin{split} 
\tr \, e^{i p \cdot x + q \cdot \eps \nabla} \omega_N &= \int dx\, e^{i/2 \eps p \cdot q}  e^{ip \cdot x} \omega_N \left( x-\eps q ; x \right) \\ &= N \int dx dv\, W_N (x,v) e^{ip \cdot x} e^{i q \cdot v} = N \widehat{W}_{N} (p,q) \end{split} \] 
Hence (\ref{eq:main3}) can be translated into the bound
\[ \Big| \widehat{W}_{N,t} (p,q) - \widehat{\wt{W}}_{N,t} (p,q) \Big| \leq C \eps  (1+|p|+|q|)^2 e^{C|t|} \]
where we recall that $W_{N,t}$ is the Wigner transform of the solution $\omega_{N,t}$ of the Hartree equation while $\wt{W}_{N,t}$ is the solution of the Vlasov equation with initial data $W_N$. 

If the sequence $W_N$ has a limit $W_0$, a probability density on phase-space, then one can also compare the Fourier transform of $W_{N,t}$ with the solution $W_t$ of the Vlasov equation with initial data $W_0$. 
 
\begin{theorem}\label{thm:new2}
Let $V \in L^1 (\bR^3)$ satisfy (\ref{eq:ass-pot}). Let $\omega_N$ be a sequence of reduced densities on $L^2 (\bR^3)$, with $\tr \, \omega_N = N$, $0 \leq \omega_N \leq 1$ and such that 
\[ \tr \, | [ x,\omega_N ]| \leq C N \eps, \quad \tr \, |[\eps \nabla , \omega_N]| \leq C N\eps \]
Denote by $W_N \in L^1 (\bR^3 \times \bR^3)$ the Wigner transform of $\omega_N$. We assume that $\| W_N \|_{W^{1,1}} \leq C$ uniformly in $N$.  

Furthermore, let $W_0 \in W^{1,1} (\bR^3 \times \bR^3)$ be a probability density, such that 
\[ \| W_N - W_0 \|_1 \leq \kappa_N \]
for a sequence $\kappa_N$ with $\kappa_N \to 0$ as $N \to \infty$. 

Let $\omega_{N,t}$ be the solution of the Hartree equation (\ref{eq:hartree1}) with initial data 
$\omega_N$ and let $W_{N,t}$ be the Wigner transform of $\omega_{N,t}$. On the other hand, let $W_t$ denote the solution of the Vlasov equation with initial data $W_0$. Then we have
\[ \sup_{p,q} \frac{1}{(1+|p|+|q|)^2} \left| \widehat{W}_{N,t} (p, q) - \widehat{W}_t (p,q) \right| \leq C \left(\eps + \kappa_N \right) e^{C|t|}  \]
\end{theorem}

{\it Remark.} A physically interesting example of sequence of initial data satisfying the assumptions of Theorem \ref{thm:new2} can be constructed also here with coherent states. Similarly to (\ref{eq:coherent}), we consider the sequence of fermionic reduced densities
\begin{equation}\label{eq:omegaN} \omega_N (x;y) = \int dr dp \, M(r,p) f_{rp} (x) \bar{f}_{rp} (y) \end{equation}
with a probability density $M \in W^{1,1} (\bR^3 \times \bR^3)$, the coherent states 
\[ f_{r,p} (x) = \eps^{-3/2} e^{-ip \cdot x / \eps} g (x-r) \]
and the Gaussian function $g(x) = (2\pi \delta^2)^{-3/4} e^{-x^2/2\delta^2}$. 
We notice that
\[ \begin{split} [x,\omega_N] (x;y) &= \eps \int dr dp \, (\nabla_p M)(r,p) f_{rp} (x) \bar{f}_{rp} (y) \\
[ \eps \nabla, \omega_N] (x;y) &= \eps \int dr dp \, (\nabla_r M) (r,p) f_{rp} (x) \bar{f}_{rp} (y)  
\end{split} \]
Hence, we obtain
\[ \tr \, |[x,\omega_N]| \leq N \eps \| \nabla_v M \|_1, \quad \tr \, |[\eps \nabla, \omega_N]| \leq N \eps \| \nabla_r M \|_1 \]
Moreover, it is simple to check that the Wigner transform $W_N$ of $\omega_N$ satisfies $\| W_N \|_{W^{1,1}} \leq C$ uniformly in $N$ and (similarly to the remark after Theorem \ref{thm:new}), 
\[ \| W_N -M \|_1 \leq C (\delta+ \eps/ \delta) \| M \|_{W^{1,1}} \]

Choosing $\delta = \eps^{1/2}$, we find $\| W_N - M \|_1 \leq C \eps^{1/2}$. Theorem \ref{thm:new2} implies therefore that the Wigner transform $W_{N,t}$ of the solution of the Hartree equation with the initial data (\ref{eq:omegaN}) is such that  
\[ \sup_{p,q \in \bR^3} \frac{1}{(1+|q| + |p|)^2} \left| \widehat{W}_{N,t} (p,q) - \widehat{W}_t (p,q) \right| \leq C \eps^{1/2} \, e^{C|t|} \]
for all $t \in \bR$. Here $W_t$ denotes the solution of the Vlasov equation with the initial data given by the probability density $W_0 = M$. Notice that the assumption $M \in W^{1,1}$ is also compatible with   
$M$ being an approximate characteristic function; this observation is important at zero temperature, to describe systems at or close to the ground state. 

\bigskip

The rest of the paper is devoted to the proof of our five main theorems, appearing in Sections \ref{sect:proof1}-\ref{sect:proof2}. Appendix \ref{app:prop} contains an important lemma on the propagation of regularity for the solution of the Vlasov equation (\ref{eq:vlasov}), which is used in Sect. \ref{sect:proof1} and Sect. \ref{sec:proof1b}. Appendix \ref{sect:comm}, on the other hand, contains a bound on the propagation of certain semiclassical commutators, which plays a key role in Sect. \ref{sec:proof1b} and in Sect. \ref{sect:proof2}.

\section{Trace norm convergence for regular data}\label{sect:proof1}

Here we prove Theorem \ref{thm:main}. Recall that $\omega_{N,t}$ denotes the solution of the Hartree equation 
\[ i\eps \partial_t \omega_{N,t} = \left[ h_H (t) , \omega_{N,t} \right] \]
with the Hartree Hamiltonian 
\[ h_H (t) = -\eps^2 \Delta + (V * \rho_t) (x) \] and the 
density $\rho_t (x) = N^{-1} \omega_{N,t} (x;x)$. We introduce the two-parameter group of unitary transformations $\cU (t;s)$, generated by $h_H (t)$. In other words, $\cU (t;s)$ solves the equation 
\begin{equation}\label{eq:aux-dyn} i\eps \partial_t \cU (t;s) = h_H (t) \cU (t;s) 
\end{equation}
with $\cU (s;s) = 1$, for all $s \in \bR$. Notice that $\omega_{N,t} = \cU (t;0) \omega_N \cU^* (t;0)$.

On the other hand, $\wt{\omega}_{N,t}$ is the Wigner transform of the solution $\wt{W}_{N,t}$ of the Vlasov equation (\ref{eq:vlasov}). We find that $\wt{\omega}_{N,t}$ satisfies 
\[ i\eps \partial_t \wt{\omega}_{N,t} = \left[ -\eps^2 \Delta , \wt{\omega}_{N,t} \right] + A_t \]
where $A_t$ is the operator with the kernel 
\[ A_t (x;y) = \nabla (V*\wt{\rho}_t) \left( \frac{x+y}{2} \right) \cdot (x-y) \, \wt{\omega}_{N,t} (x;y) \]

We conjugate now the difference $\omega_{N,t} - \wt{\omega}_{N,t}$ with the unitary operator $\cU(t;0)$. Taking the time derivative, we find  
\be
\begin{split}
i\,\e\,\partial_t\,\mathcal{U}^*(t;0)\,&(\o_{N,t}-\wt\o_{N,t})\,\mathcal{U}(t;0) \\ = \; &-\mathcal{U}^*(t;0)\,[h_H (t),\o_{N,t}-\wt\o_{N,t}]\,\mathcal{U} (t;0)\\
& +\mathcal{U}^*(t;0)\,([h_H(t),\o_{N,t}]-[-\e^2\,\Delta,\wt\o_{N,t}] - A_t)\,\mathcal{U}(t;0)\\
= \; &\mathcal{U}^*(t;0)\,([V*\rho_t,\wt\o_{N,t}]-A_t)\,\mathcal{U}(t;0)\\
= \; &\mathcal{U}^*(t;0)\,([V*(\rho_t-\wt\rho_t),\wt\o_{N,t}]+ B_t)\,\mathcal{U}(t;0)
\end{split}
\ee
where $B_t$ denotes the operator with the integral kernel 
\begin{equation}\label{eq:Bt} B_t (x;y) = \left[ (V*\wt{\rho}_t) (x) - (V*\wt{\rho}_t)(y) - \nabla (V*\wt{\rho}_t) \left(\frac{x+y}{2} \right) \cdot (x-y) \right] \wt{\omega}_{N,t} (x;y)
\end{equation}
Integration in time gives (since, at time $t=0$, $\omega_{N,0} = \wt{\omega}_{N,0} = \omega_N$) 
\be\label{eq:omega-omegatilde}
\begin{split}
\mathcal{U}^*(t;0)\,(\o_{N,t}-\wt\o_{N,t})\,\mathcal{U}(t;0) = \;& \frac{1}{i\e}\int_0^t \mathcal{U}^*(t;s)\,[V*(\rho_s-\wt\rho_s),\wt\o_{N,s}]\,\mathcal{U}(t;s)\,ds \\&+\frac{1}{i\e}\int_0^t \mathcal{U}^*(t;s)\,B_s\,\mathcal{U}(t;s)\,ds
\end{split}
\ee
Taking the trace norm, we obtain 
\be\label{eq:trace-norm}
\tr \, |\o_{N,t}-\wt\o_{N,t}|\leq \frac{1}{\e}\int_0^t \tr \, |[V*(\rho_s-\wt\rho_s),\wt\o_{N,s}]|\,ds + \frac{1}{\e}\int_0^t \tr \, |B_s|\,ds \,. 
\ee
We will estimate the two terms in the right-hand side of (\ref{eq:trace-norm}) separately, and conclude by applying Gronwall's lemma.

\medskip

\noindent{\bf Estimate of the first term in (\ref{eq:trace-norm}).} We start by considering the first term on the r.h.s. of (\ref{eq:trace-norm}). To this end, we observe that
\be\label{eq:trace-norm2} 
\begin{split} \tr \,|[V*(\rho_s-\wt\rho_s),\wt\o_{N,s}]| &\leq \int dz |\rho_s (z) - \wt{\rho}_s (z)| \, \tr  \left|[ V(.-z), \wt{\omega}_{N,s}] \right| \\ &\leq \| \rho_s - \wt{\rho}_s \|_1 \sup_z \tr \, |[V(z-.), \wt{\omega}_{N,s}]| . \end{split}
\ee
We start by estimating the last term in the right-hand side of (\ref{eq:trace-norm2}). We have
\be \label{eq:thm1-c1}
\begin{split}
\tr\,|[V(\cdot - z),\wt\o_{N,s}]| &= \tr\,|(1-\e^2\Delta)^{-1}(1+x^2)^{-1}(1+x^2)(1-\e^2\Delta)[V(\cdot - z),\wt\o_{N,s}]|\\
&\leq \|(1-\e^2\Delta)^{-1}(1+x^2)^{-1}\|_{\text{HS}} \, \|(1+x^2)(1-\e^2\Delta)[V(\cdot - z),\wt\o_{N,s}]\|_{\text{HS}}\;.
\end{split}
\ee
An explicit computation shows that 
\[ \| (1-\eps^2 \Delta)^{-1} (1+x^2)^{-1} \|_\text{HS} \leq C \sqrt{N} \]
As for the operator $D := (1+x^2) (1-\eps^2 \Delta) [V(z-.), \wt{\omega}_{N,s}]$, it has the integral kernel 
\[ \begin{split} D(x;y) &= (1+x^2) (1-\eps^2 \Delta_x) (V(x-z) - V(y-z)) \, \wt{\omega}_{N,s} (x;y)  \\ &= N (1+x^2) (1-\eps^2 \Delta_x) (V(x-z) - V(y-z)) \int dv \, \wt{W}_{N,s} \Big( \frac{x+y}{2} , v \Big) e^{i v \cdot \frac{x-y}{\eps}} \end{split} \]
where we used the definition of $\wt{\omega}_{N,s}$ as the Weyl quantization of the solution $W_s$ of the Vlasov equation, with initial data $W_0$. Taking into account the fact that the Laplacian $\Delta_x$ can act on the potential $V(x-z)$, on the function $W_s$ or on the phase $e^{iv \cdot (x-y)/\eps}$, we obtain that 
\be\label{eq:all-terms}
\begin{split}
D(x;y) =\; & N(1+x^2)\,(V(x-z)-V(y-z))\,\int \wt{W}_{N,s}\Big( \frac{x+y}{2} , v \Big)\,e^{i\,v\cdot(x-y)/\e}\,dv\\
& - N\e^2(1+x^2)\,(\Delta V) (x-z)\,\int \wt{W}_{N,s}\Big( \frac{x+y}{2} , v \Big)\,e^{i\,v\cdot(x-y)/\e}\,dv\\
& - \frac{N\e^2}{4} (1+x^2)\,(V(x-z)-V(y-z))\,\int (\Delta_1 \wt{W}_{N,s})\Big( \frac{x+y}{2} , v \Big)\,e^{i\,v\cdot(x-y)/\e}\,dv\\
& + N(1+x^2)\,(V(x-z)-V(y-z))\,\int \wt{W}_{N,s}\Big( \frac{x+y}{2} , v \Big)\,v^2\,e^{i\,v\cdot(x-y)/\e}\,dv\\
& - \frac{N\e^2}{2} (1+x^2)\, (\nabla V)(x-z)\cdot \int (\nabla_1 \wt{W}_{N,s}) \Big( \frac{x+y}{2} , v \Big)\,e^{i\,v\cdot(x-y)/\e}\,dv\\
& - iN\e(1+x^2)\,(\nabla V) (x-z)\cdot \int \wt{W}_{N,s}\Big( \frac{x+y}{2} , v \Big)\,v\,e^{i\,v\cdot(x-y)/\e}\,dv\\
& - \frac{iN\e}{2} (1+x^2)\,(V(x-z)-V(y-z))\,\int (\nabla_1 \wt{W}_{N,s})\Big( \frac{x+y}{2} , v \Big)\,v\,e^{i\,v\cdot(x-y)/\e}\,dv\,.\\
=: \; & \sum_{j=1}^7 D_j (x;y)
\end{split}
\ee

\medskip

We estimate now the Hilbert-Schmidt norm of the different contributions on the r.h.s. of (\ref{eq:all-terms}). To control the term $D_1$, we expand
\[ \begin{split} D_1 (x;y) &= N (1+x^2) (V(x-z) - V(y-z)) \int \wt{W}_{N,s} \Big( \frac{x+y}{2} , v \Big) e^{iv \cdot (x-y)/\e} dv \\ &= N (1+x^2) (\nabla V) (\xi) \cdot (x-y) \int \wt{W}_{N,s}  \Big( \frac{x+y}{2} , v \Big) e^{iv \cdot (x-y)/\e} dv \\ &= iN\eps (1+x^2) (\nabla V) (\xi) \cdot \int (\nabla_2 \wt{W}_{N,s})  \Big( \frac{x+y}{2} , v \Big) e^{iv \cdot (x-y)/\e} dv \end{split} \]
for an appropriate $\xi$ on the segment between $x-z$ and $y-z$. Using the bound
\[ 1+x^2 \leq 1 +2 \left(\frac{x+y}{2}\right)^2 + \frac{\e}{2}^2\left(\frac{x-y}{\e}\right)^2
\]
and the assumption $V \in W^{2,\infty} (\bR^3)$  we get:
\begin{equation}\label{eq:A1-bd} \begin{split} \| D_1 \|_\text{HS}^2 \leq &\; C N^2 \eps^2 \int dx dy \Big[ 1+2\Big(\frac{x+y}{2}\Big)^2 + \frac{\e^2}{2}\Big(\frac{x-y}{\e}\Big)^2 \Big]^2 \Big| \int (\nabla_2 \wt{W}_{N,s})  \Big( \frac{x+y}{2} , v \Big) e^{iv \cdot (x-y)/\e} dv \Big|^2 \\ = &\; C N \eps^2 \int dX dr \left[ 1+X^2 +\e^2 r^2 \right]^2 \Big| \int (\nabla_2 \wt{W}_{N,s}) (X, v) e^{iv \cdot r} dv \Big|^2 \\
\leq &\; CN \eps^2 \int dX dv (1+X^2)^2 |\nabla_2 \wt{W}_{N,s} (X,v)|^2 + CN\eps^6 \int dX dv |\nabla_2^3 \wt{W}_{N,s} (X,v)|^2 \\ \leq &\; CN \eps^2 \| \wt{W}_{N,s} \|_{H^1_{4}}  + C N \eps^6 \| \wt{W}_{N,s} \|_{H^3} \end{split} 
\end{equation}
Similarly, we control the Hilbert-Schmidt norm of the second term on the r.h.s. of (\ref{eq:all-terms}):
\[ \begin{split} 
\| D_2 \|_\text{HS}^2 &\leq C N \eps^4 \int dX dr \left[ 1+ X^2 + \eps^2 r^2 \right]^2 \Big| \int \wt{W}_{N,s} (X,v) e^{iv\cdot r} dv \Big|^2 \\ &\leq C N \eps^4 \| \wt{W}_{N,s} \|_{H^0_{4}}^2 + CN \eps^8 \| \wt{W}_{N,s} \|_{H^2}^2 \end{split} \]
Proceeding analogously to bound the Hilbert-Schmidt norm of the other terms on the r.h.s. of (\ref{eq:all-terms}), we conclude that
\[ \| D \|_{\text{HS}} \leq C \sqrt{N} \left[ \eps \| \wt{W}_{N,s} \|_{H^1_{4}} + \eps^2 \| \wt{W}_{N,s} \|_{H^2_{4}} + \eps^3 \| \wt{W}_{N,s} \|_{H^3_{4}} + \eps^4 \| \wt{W}_{N,s} \|_{H^4_4} \right] \]

Proposition \ref{prop:regularity} allows us to control the weighted Sobolev norms of the solution $\wt{W}_{N,s}$ of the Vlasov equation by their initial values. We obtain
\[\| D \|_{\text{HS}} \leq C e^{C|s|} \sqrt{N} \left[ \eps \| W_N \|_{H^1_{4}} + \eps^2 \| W_N \|_{H^2_{4}} + \eps^3 \| W_N \|_{H^3_{4}} + \eps^4 \| W_N \|_{H^4_4} \right] \]
for a constant $C>0$, depending on $\|W_{N}\|_{H^{2}_{4}}$. Thus, from (\ref{eq:thm1-c1}), we finally find 
\[ \tr\,|[V(\cdot - z),\wt\o_{N,s}]| \leq C e^{C|s|} N \eps \left[ \| W_N \|_{H^1_4} + \eps \| W_N \|_{H^2_4}  + \eps^2 \| W_N \|_{H^3_{4}} + \eps^3 \| W_N \|_{H^4_4} \right] \]
Therefore, from (\ref{eq:trace-norm2}):
\be\label{eq:IplusII}
\begin{split}
\tr \,|[V*(\rho_s-\wt\rho_s),\wt\o_{N,s}]| &\leq \| \rho_{s} - \tilde\rho_{s} \|_{1}\tr\,|[V(\cdot - z),\wt\o_{N,s}]| \\
&\leq C e^{C|s|} N \eps \| \rho_{s} - \tilde\rho_{s} \|_{1} \| W_{N} \|_{H^{1}_{4}}\\
&\quad + C e^{C|s|} N \eps^{2} \| \rho_{s} - \tilde\rho_{s} \|_{1} \| \big[ \| W_N \|_{H^2_4}  + \eps \| W_N \|_{H^3_{4}} + \eps^2 \| W_N \|_{H^4_4} \big]\\
&\equiv \text{I} + \text{II}
\end{split}
\ee

Consider first $\text{I}$. We have
\[ \| \rho_s - \wt{\rho}_s \|_1 = \sup_{J \in L^\infty (\bR^3) : \| J \|_\infty \leq 1} \left| \int J(z) (\rho_s (z) - \wt{\rho}_s (z)) dz \right| \leq N^{-1} \sup_{J : \| J \| \leq 1} \left| \tr \, J (\omega_{N,s} -\wt{\omega}_{N,s}) \right| \]
where on the r.h.s. the supremum is taken over all bounded operator with operator norm lesser or equal than one. We conclude that
\[ \| \rho_s - \wt{\rho}_s \|_1 \leq N^{-1} \, \tr \, |\omega_{N,s} - \wt{\omega}_{N,s}| \]
Therefore,
\be\label{eq:boundI}
\text{I} \leq C e^{C|s|} \eps\, \tr\,|\omega_{N,s} - \tilde\omega_{N,s}| \| W_{N} \|_{H^{1}_{4}}
\ee
To bound $\text{II}$, we write:
\[
\| \rho_{s} - \tilde\rho_{s} \|_{1} \leq \| \rho_{s} \|_{1} + \| \tilde\rho_{s} \|_{1} = N^{-1}\tr\, \omega_{N,s} + \| \tilde\rho_{s} \|_{1} \leq 1 + \| \widetilde W_{N,s} \|_{1}
\]
Using that the Vlasov dynamics preserves the $L^{p}$ norms, we get:
\[
\| \widetilde W_{N,s} \|_{1} = \| W_{N} \|_{1} = \| (1 + x^{2} + v^{2})^{-2} (1 + x^{2} + v^{2})^{2} W_{N}  \|_{1} \leq C\| W_{N} \|_{H^{0}_{4}}
\]
and thus:
\be\label{eq:boundII}
\text{II} \leq C e^{C|s|} N \eps^{2} \| W_{N} \|_{H^{0}_{4}} \big[ \| W_N \|_{H^2_4}  + \eps \| W_N \|_{H^3_{4}} + \eps^2 \| W_N \|_{H^4_4} \big]
\ee
From (\ref{eq:trace-norm2}), (\ref{eq:IplusII}), (\ref{eq:boundI}), (\ref{eq:boundII}), we obtain:
\begin{equation}\label{eq:term1A}
\begin{split}  \frac{1}{\eps} \int_0^t \tr \, \left| \left[ V* (\rho_s - \wt{\rho}_s), \wt{\omega}_{N,s} \right] \right| ds \leq \; &C  \int_0^t e^{C|s|} \, \tr \, |\omega_{N,s} - \wt{\omega}_{N,s}| \, ds \\ &+ C e^{C|t|}  N \eps \,\left[ \|W_N\|_{H^2_4} + \eps \| W_N \|_{H^3_4} + \eps^2 \| W_N \|_{H^4_4} \right] 
\end{split}
\end{equation}
where the constant $C >0$ depends on $\| W_N \|_{H^2_4}$, but not on the higher Sobolev norms of $W_N$. This concludes the estimate of the first term in the right-hand side of (\ref{eq:trace-norm}).

\medskip

\noindent{\bf Estimate for the second term in (\ref{eq:trace-norm})}. To conclude and apply Gronwall's lemma, we need to bound the second term in (\ref{eq:trace-norm}). We find
\begin{equation}\label{eq:trBs} \tr \, |B_s| \leq \| (1-\eps^2\Delta)^{-1} (1+x^2)^{-1} \|_\text{HS} \, \| (1+x^2) (1-\eps^2 \Delta) B_s \|_\text{HS} \leq C \sqrt{N} \| (1+x^2) (1-\eps^2\Delta) B_s \|_\text{HS} 
\end{equation}
Let $U_s := V* \wt{\rho}_s$. The kernel of the operator $\wt{B} := (1-\eps^2 \Delta) B_s$ is given by 
\be\label{eq:all-terms2}
\begin{split}
& \wt{B} (x;y) = N \Big[U_s(x) - U_s(y) - \nabla U_s\Big(\frac{x+y}{2}\Big)\cdot(x-y)\Big]\int \wt{W}_{N,s}\Big(\frac{x+y}{2},v\Big)e^{i\,v\cdot\frac{(x-y)}{\e}}dv   \\
&\quad - N\e^2\Big[\Delta U_s(x) - \frac{1}{4}\Delta \nabla U_s\Big(\frac{x+y}{2}\Big)\cdot(x-y) - \frac{1}{2} \Delta U_s \Big(\frac{x+y}{2}\Big)\Big] \int \wt{W}_{N,s}\Big(\frac{x+y}{2},v\Big)e^{i\,v\cdot\frac{(x-y)}{\e}}dv \\
&\quad - \frac{N\e^2}{4} \Big[U_s(x) - U_s(y) - \nabla U_s\Big(\frac{x+y}{2}\Big)\cdot(x-y)\Big] \int (\Delta_1 \wt{W}_{N,s}) \Big(\frac{x+y}{2},v\Big)e^{i\,v\cdot\frac{(x-y)}{\e}}dv \\
&\quad + N \Big[U_s(x) - U_s(y) - \nabla U_s\Big(\frac{x+y}{2}\Big)\cdot(x-y)\Big] \int \wt{W}_{N,s}\Big(\frac{x+y}{2},v\Big) v^2 e^{i\,v\cdot \frac{(x-y)}{\e}}dv \\
&\quad - \frac{N\e^2}{2} \Big[\nabla U_s(x) - \frac{1}{2} \nabla^2 U_s\Big(\frac{x+y}{2}\Big) (x-y) - \nabla U_s\Big(\frac{x+y}{2}\Big)\Big] \int (\nabla_1 \wt{W}_{N,s}) \Big(\frac{x+y}{2},v\Big) e^{i\,v\cdot \frac{(x-y)}{\e}}dv \\
&\quad - N\e \Big[\nabla U_s(x) - \frac{1}{2} \nabla^2 U_s\Big(\frac{x+y}{2}\Big)(x-y) - \nabla U_s\Big(\frac{x+y}{2}\Big)\Big] \int \wt{W}_{N,s}\Big(\frac{x+y}{2},v\Big) v e^{i\,v\cdot \frac{(x-y)}{\e}}dv \\
&\quad - N\e \Big[U_s(x) - U_s(y) - \nabla U_s\Big(\frac{x+y}{2}\Big)\cdot(x-y) \Big] \int (v\cdot \nabla_1 \wt{W}_{N,s}) \Big(\frac{x+y}{2},v\Big) e^{i\,v\cdot \frac{(x-y)}{\e}}dv\\ &\quad =: \sum_{j=1}^7 \wt{B}_j (x;y)  
\end{split}
\ee
In the contributions $\widetilde B_1, \widetilde B_4, \widetilde B_6, \widetilde B_7$, we need to extract additional factors of $\eps$; the goal is to show that $\| (1 + x^{2}) \wt{B} \|_\text{HS} \leq C \sqrt{N} \eps^2$. To this end, we write  
\[ \begin{split} U_s (x) -& U_s (y) - \nabla U_s \left( \frac{x+y}{2} \right) \cdot (x-y) \\ &= \int_0^1 d\lambda \left[ \nabla U_s \big(\lambda x + (1-\lambda) y\big) - \nabla U_s \big((x+y)/2\big) \right] \cdot (x-y)  \\ &= \sum_{i,j=1}^3\int_0^1 d\lambda \int_0^1 d\mu\, \partial_i \partial_j U_s \Big(\mu (\lambda x + (1-\lambda) y) + (1-\mu) (x+y)/2\Big) (x-y)_i (x-y)_j\Big(\lambda-\frac{1}{2}\Big) \end{split} \]
and we estimate, using the assumption $V \in W^{2,\infty} (\bR^3)$ and integrating by parts, 
\[ |\wt{B}_1 (x;y)| \leq C N \eps^2 \sum_{i,j=1}^3 \left| \int \partial_{v_i} \partial_{v_j} \wt{W}_{N,s} \Big( \frac{x+y}{2} , v \Big)   e^{i v\cdot (x-y)/\eps} \right| \]
Hence, proceeding similarly as we did in (\ref{eq:A1-bd}), we get:
\[ \begin{split} \| (1 + x^{2}) \wt{B}_1 \|^2_{\rm HS} &\leq C N^2 \eps^4 \int dx dy \, (1+x^2)^2 \sum_{i,j=1}^3 \left|  \int \partial_{v_i} \partial_{v_j} \wt{W}_{N,s}  \Big( \frac{x+y}{2} , v \Big)  e^{i v\cdot (x-y)/\eps} \right|^2 \\ &= C N \eps^4 \int dX dr \left[ 1+ X^2 + \eps^2 r^2 \right]^2 \left| \int \partial_{v_i} \partial_{v_j} \wt{W}_{N,s}  \left(X, v \right) e^{i v\cdot r} \right|^2 \\ &\leq C N\eps^4 \| \wt{W}_{N,s} \|_{H^2_{4}}^2 + CN \eps^8 \| \wt{W}_{N,s} \|_{H^4}^2 \end{split}\]
The Hilbert-Schmidt norm of the other terms in the right-hand side of (\ref{eq:all-terms2}) can be estimated in a similar way. To do this, it is useful to notice that:
\[
\| \nabla^3 U \|_\infty = \| \nabla^{2} V*\nabla\wt{\rho}_{s} \|_{\infty} \leq \| \nabla^2 V \|_\infty \| \nabla \wt{\rho}_s \|_1 \leq C e^{C|s|}
\]
where we used that $V \in W^{2,\infty} (\bR^3)$, and that $\| \nabla \wt{\rho}_{s} \|_{1} \leq C\|\widetilde W_{N,s}\|_{H^{1}_{4}}$. The final result is:
\[ \| (1 + x^{2}) \wt{B} \|_\text{HS} \leq C \sqrt{N} \left[ \eps^2 \| \wt{W}_{N,s} \|_{H^2_{4}} + \eps^3  \| \wt{W}_{N,s} \|_{H^3_{4}} + \eps^4 \| \wt{W}_{N,s} \|_{H^4_{4}} + \eps^5 \| \wt{W}_{N,s} \|_{H^5_4}  \right] \]
Therefore, by Proposition \ref{prop:regularity},
\[ \| (1 + x^{2}) \wt{B} \|_\text{HS} \leq C e^{C|s|} \sqrt{N} \eps^2 \left[ \|W_N\|_{H_4^2} + \eps  \| W_N \|_{H^3_4} + \eps^2 \| W_N \|_{H^4_4} + \eps^3 \| W_N \|_{H_4^5}  \right] \]
where the constant $C >0$ depends on $\| W_N \|_{H^2_4}$ but not on the higher Sobolev norms of $W_N$. This gives:
\be\label{eq:estimateB}
\tr \, |B_s| \leq C e^{C|s|} N \eps^2 \left[ \|W_N\|_{H_4^2}  + \eps \| W_N \|_{H^3_{4}} + \eps^2 \| W_N \|_{H^4_4} + \eps^3 \| W_N \|_{H_4^5} \right]
\ee

\medskip

\noindent{\bf Proof of Theorem \ref{thm:main}.} We are now in the position to conclude the proof. Inserting (\ref{eq:term1A}), (\ref{eq:estimateB}) into (\ref{eq:trace-norm}), we get:
\[ \begin{split} \tr \, | \omega_{N,t} - \wt{\omega}_{N,t} | \leq \; &C \int_0^t \tr \,e^{C|s|} \, |\omega_{N,s} - \wt{\omega}_{N,s}|  ds \\ & + C e^{C|t|} N \eps \left[ \|W_N\|_{H_4^2}  + \eps \sup_N \| W_N \|_{H^3_{4}} + \eps^2 \sup_N \| W_N \|_{H^4_4} + \eps^3 \sup_N \| W_N \|_{H_4^5} \right] \end{split} \]
Finally, Gronwall's lemma implies the desired bound
\[ 
\begin{split}
&\tr \, |\omega_{N,t} - \wt{\omega}_{N,t} | \\
&\quad\leq C N\eps \exp (C \exp (C|t|)) \left[ \sup_N \|W_N\|_{H_4^2}  + \eps \sup_N \| W_N \|_{H^3_{4}} + \eps^2 \sup_N \| W_N \|_{H^4_4} + \eps^3 \sup_N \| W_N \|_{H^5_4} \right] 
\end{split}
\]
with $C$ depending only on $\|W_N\|_{H_4^2}$.  This concludes the proof.
\qed 


\section{Hilbert-Schmidt norm convergence}
\label{sec:proof1b}

Here we prove Theorem \ref{cor:relaxed-main} and Theorem \ref{thm:new}. The proof of Theorem \ref{cor:relaxed-main} is based on an approximation argument, together with our previous result Theorem \ref{thm:main}. 

\medskip

\noindent{\bf Regularization of the initial data.} We start by approximating the initial  data $W_N$. For $k > 0$, we define 
\[ g_k (x,v)= \frac{k^3}{(2\pi)^3} e^{-\frac{k}{2} (x^2 + v^2)} \]
and
\[ W_N^k (x,v) = (W_N * g_k) (x,v) = \int dx' dv' g_k (x-x', v-v') W_N (x',v') \]
Then, we have $\| W_N^k \|_{H^5_4} < \infty$ for all $N \in \bN$. In fact, we find \begin{equation}\label{eq:bds-WNk}
\begin{split} 
\| W_N^k \|_{H^j_4} &\leq C\| W_N \|_{H^2_4} \quad \text{if $j \leq 2$ and } \\ \| W_N^k \|_{H^j_4} &\leq C k^{(j-2)/2} \| W_N \|_{H^2_4} \quad \text{for $j =3,4,5$}.
\end{split} \end{equation} Furthermore, we notice that 
\begin{equation}\label{eq:diff-Win}  
\| W_N - W_N^k \|_{H^s_a} \leq \frac{C}{\sqrt{k}} \| W_N \|_{H^{s+1}_a}
\end{equation}
for $s=0,1$ (with the convention $H^0 \equiv L^2$) and for $a\leq 4$. We denote by $\omega_N^k$ the Weyl quantization of $W_N^k$. We observe that 
\begin{equation}\label{eq:om-fermi} 
\begin{split} 
\omega_N^k (x;y) = \; &N \int dv \, W_N^k \Big( \frac{x+y}{2} , v \Big) e^{i v\cdot \frac{x-y}{\eps}} \\ = \; &\frac{Nk^3}{(2\pi)^3} \int dv dx' dv' \, e^{-\frac{k}{2} \left( \frac{x+y}{2} - x' \right)^2} e^{-\frac{k}{2} (v-v')^2} W_N (x',v') e^{iv \cdot \frac{x-y}{\eps}} \\
= \; & \frac{k^{3/2}}{(2\pi)^3} \int dw dx' \, e^{-\frac{k}{2} \left( \frac{x+y}{2} - x' \right)^2} e^{-w^2/2} \omega_N \left(x' + \frac{x-y}{2}, x' - \frac{x-y}{2} \right) e^{iw \cdot \frac{x-y}{\sqrt{k} \eps}} \\
= \; & \frac{1}{(2\pi)^3} \int dw dz \, e^{-z^2/2} e^{- w^2/2} \omega_N \left( x + \frac{z}{\sqrt{k}} , y + \frac{z}{\sqrt{k}} \right) e^{i w \cdot \frac{x-y}{\sqrt{k} \eps}} \\ = \; & \frac{1}{(2\pi)^3} \int dw dz \, e^{-z^2/2} e^{-w^2/2} \left[ e^{iw \cdot \frac{x}{\sqrt{k}\eps}} e^{\frac{z}{\sqrt{k}} \cdot \nabla} \omega_N e^{-\frac{z}{\sqrt{k}} \cdot \nabla} e^{-i w \cdot \frac{x}{\sqrt{k} \eps}} \right] (x;y) \end{split} \end{equation}
Hence $\omega_N^k$, as a convex combination of fermionic reduced densities, is again a fermionic reduced density (i.e. $0 \leq \omega_N^k \leq 1$ and $\tr\, \omega_N^k = N$). {F}rom (\ref{eq:diff-Win}), we find 
\begin{equation}\label{eq:diff-ini} \| \omega_N - \omega_N^k \|_\text{HS} = \sqrt{N} \| W_N - W_N^k \|_2 \leq \sqrt{\frac{N}{k}} \| W_N \|_{H^1} 
\end{equation}

We denote by $\omega_{N,t}$ and $\omega_{N,t}^k$ the solution of the Hartree equation with initial data $\omega_N$ and, respectively, $\omega_{N}^k$. On the other hand, $\wt{\omega}_{N,t}$ and $\wt{\omega}_{N,t}^k$ will denote the Wigner transform of the solutions $\wt{W}_{N,t}$ and $\wt{W}_{N,t}^k$ of the Vlasov equation with initial data $W_N$ and, respectively, $W_N^k$. Notice that, since the Vlasov equation preserves all the $L^{p}$ norms, $\|\wt{\omega}_{N,t} \|_\text{HS} = N^{1/2} \| \wt{W}_{N,t} \|_2 = N^{1/2} \| W_N \|_2$ and, similarly, $\| \wt{\omega}_{N,t}^k \|_\text{HS} = N^{1/2} \| W_N^k \|_2$, for all $t \in \bR$.

We need to compare $\omega_{N,t}$ with $\wt{\omega}_{N,t}$. To this end, we will first compare $\omega_{N,t}^k$ with $\wt{\omega}_{N,t}^k$. Later, we will have to compare $\omega_{N,t}$ with $\omega_{N,t}^k$ and, separately, $\wt{\omega}_{N,t}$ with $\wt{\omega}_{N,t}^k$. 

\medskip

\noindent{\bf Comparison of $\omega_{N,t}^{k}$ with $\wt{\omega}_{N,t}^{k}$.} To begin, we prove that there exists a constant $C>0$ such that 
\be\label{eq:HS-step1} 
\| \omega_{N,t}^k - \wt{\omega}_{N,t}^k \|_\text{HS} \leq C N^{1/2} \eps  \, \exp (C \exp (C|t|)) \big[1+ \sum_{j=1}^3(\eps \sqrt{k})^j \big]
\ee
The constant depends on $\sup_{N}\| W_{N} \|_{H^{2}_{4}}$, but not on the higher Sobolev norms. To show (\ref{eq:HS-step1}), we shall use our previous result, Theorem \ref{thm:main}. In fact, from (\ref{eq:main}), (\ref{eq:bds-WNk}) we find
\begin{equation}\label{eq:tr-bd} \begin{split} \| \omega_{N,t}^k - \wt{\omega}_{N,t}^k \|_\text{tr} \leq \; &C N\eps \, \exp (C \exp (C|t|)) \left( \| W_N^k \|_{H^2_4} + \sum_{\beta=1}^3 \eps^\beta\sup_N \| W_N^k \|_{H^{\beta+2}_4} \right) \\ \leq \; &CN \eps \, \exp (C \exp (C|t|)) \left(1+\sum_{j=1}^3(\eps \sqrt{k})^j \right) \end{split}
\end{equation}
for a constant $C > 0$ depending only on $\sup_N \| W_N \|_{H^2_4}$. We shall use this result to prove an estimate for the Hilbert-Schmidt norm of the difference of the two evolutions. Proceeding as in (\ref{eq:aux-dyn}) -- (\ref{eq:trace-norm}), we have:
\begin{equation}\label{eq:HS-bd} \| \omega_{N,t}^k - \wt{\omega}_{N,t}^k \|_\text{HS} \leq \frac{1}{\eps} \int_0^t ds \, \left\| [ V* (\rho_s^k - \wt{\rho}_s^k) , \wt{\omega}_{N,s}^k ] \right\|_\text{HS}  + \frac{1}{\eps} \int_0^t ds \, \| B^k_s \|_\text{HS} 
\end{equation}
where $B^k_s$ is the operator with the integral kernel
\[ B^k_s (x;y) = \left[ V*\wt{\rho}_s^k (x) - V* \wt{\rho}_s^k (y) - \nabla (V* \wt{\rho}_s^k) \left( \frac{x+y}{2} \right) \cdot (x-y) \right] \wt{\omega}^k_{N,t} (x;y) \]
We shall estimate the two terms in (\ref{eq:HS-bd}) separately. We start with the first. We have:
\[ \begin{split} \left\| [ V* (\rho_s^k - \wt{\rho}_s^k), \wt{\omega}_{N,s}^k ] \right\|_\text{HS} &\leq \int dz |\rho_s^k (z) - \wt{\rho}_s^k (z)| \cdot \left\| [ V(z-.), \wt{\omega}_{N,s}^k ] \right\|_\text{HS} \\ &\leq \| \rho_s^k - \wt{\rho}_s^k \|_1 \int dp \, |\wt{V} (p)| \left\| [ e^{ip \cdot x}, \wt{\omega}_{N,s}^k ] \right\|_\text{HS} \end{split} \]
Using $\| \rho_s^k - \wt{\rho}_s^k \|_1 \leq N^{-1} \| \omega^k_{N,t} - \wt{\omega}_{N,t}^k \|_\text{tr}$, the identity
\[ [e^{ip \cdot x}, \wt{\omega}_{N,s}^k ] = \int_0^1 d\lambda \, e^{i\lambda p \cdot x} \left[ i p \cdot x, \wt{\omega}_{N,s}^k \right] e^{i(1-\lambda) p \cdot x} \]
and the assumption (\ref{eq:pot-hyp}) on the potential, we conclude that
\begin{equation}\label{eq:comm-HS}
\left\| [ V* (\rho_s^k -\wt{\rho}_s^k), \wt{\omega}_{N,s}^k ] \right\|_\text{HS} \leq C N^{-1} \| \omega_{N,s}^k - \wt{\omega}_{N,s}^k \|_\text{tr} \| [x,\wt{\omega}_{N,s}^k] \|_\text{HS} 
\end{equation} 
We shall use the regularity of $W^{k}_{N,t}$ to extract a factor $\e$ from the commutator in (\ref{eq:comm-HS}). We have:
\[ [x,\wt{\omega}_{N,s}^k ] (x;y) = (x-y) \, \int dv\,  \widetilde W^k_{N,s} \Big( \frac{x+y}{2} , v \Big) e^{i v\cdot \frac{x-y}{\eps}}  = \eps \int dv \, \nabla_v \widetilde{W}_{N,s}^k \Big( \frac{x+y}{2} , v \Big) e^{iv \cdot \frac{x-y}{\eps}} \]
and thus, similarly to (\ref{eq:commgrad2}),
\[ \big| [x, \wt{\omega}_{N,s}^k ] \big\|_\text{HS} = \eps N^{1/2} \big\| \nabla_v \wt{W}_{N,s}^k \big\|_2 \leq C e^{C|s|} \eps N^{1/2} \| W_N \|_{H^1} \]
The second inequality follows from the propagation of regularity for solutions of the Vlasov equation, proven in Proposition \ref{prop:regularity}. Inserting the last bound and (\ref{eq:tr-bd}) in (\ref{eq:comm-HS}), we obtain
\begin{equation}\label{eq:sec-2} \left\| [ V * (\rho_s^k - \wt{\rho}_s^k), \wt{\omega}_{N,s}^k ] \right\|_\text{HS} \leq C N^{1/2} \eps^2 \, \exp (C \exp (C|t|)) \, \left(1+\sum_{j=1}^3(\eps \sqrt{k})^j \right)
\end{equation}
which concludes the estimate for the first term in (\ref{eq:HS-bd}). Let us now consider the second term in (\ref{eq:HS-bd}). We have:
\[ \begin{split}
\| B_s \|_\text{HS}^2 &= \int dx dy \left| (V*\wt{\rho}_s^k) (x) - (V*\wt{\rho}_s^k)(y) - \nabla (V* \wt{\rho}_s^k) \left( \frac{x+y}{2} \right) \cdot (x-y) \right|^2  |\wt{\omega}_{N,s}^k (x;y)|^2 \\ &= \int dx dy \, |\wt{\omega}_{N,s}^k (x;y)|^2 |x-y|^2 \, \left| \int_0^1 d\lambda \, \left[ \nabla (V * \wt{\rho}_s^k) (\lambda x + (1-\lambda)y) - \nabla (V*\wt{\rho}_s^k)((x+y)/2) \right] \right|^2 
\\ &\leq \int dx dy \, |\wt{\omega}_{N,s}^k (x;y)|^2 |x-y|^4 \\ &\hspace{1cm} \times \left| \int_0^1 d\lambda \int_0^1 d\mu \, (\lambda -1/2) \nabla^2 (V* \wt{\rho}_s^k) (\mu (\lambda x + (1-\lambda) y) + (1-\mu) (x+y)/2) \right|^2 
\\ &\leq C \int dx dy\, |x-y|^4 |\wt{\omega}_{N,s}^k (x;y)|^2  \end{split} \] 
using the assumption (\ref{eq:pot-hyp}).  Since
\[ (x-y)^2 \wt{\omega}_{N,s}^k (x;y) = -\eps^2 \int dv \, \Delta _v \wt{W}_{N,s}^k \Big( \frac{x+y}{2} , v \Big) e^{iv \cdot \frac{x-y}{\eps}} \]
we find, similarly to (\ref{eq:commgrad2}),
\[ \| B_s \|_\text{HS}^2 \leq C N \eps^4 \| \Delta_v \wt{W}_{N,s}^k \|_2^2 \leq C e^{C |s|} \eps^4 N \|W_N^k \|_{H^2}^2\leq C e^{C |s|} \eps^4 N \]
where we used again Proposition \ref{prop:regularity}. This concludes the estimate of the second term in (\ref{eq:HS-bd}). Therefore, plugging the estimates (\ref{eq:sec-2}), (\ref{eq:commgrad2}) into (\ref{eq:HS-bd}), we have:
\[ \| \omega_{N,t}^k - \wt{\omega}_{N,t}^k \|_\text{HS} \leq C N^{1/2} \eps \, \exp (C \exp (C|t|))\, (1+\sum_{j=1}^3(\eps \sqrt{k})^j) \]
as claimed.

\medskip

\noindent{\bf Comparison of $\omega_{N,t}^{k}$ with $\omega_{N,t}$.} The next step is to compare the Hartree dynamics of the regularized initial data with the Hartree dynamics of the original data.
Our goal is to show that: 
\be\label{eq:HS-step2} 
\| \omega_{N,t} - \omega_{N,t}^k \|_\text{HS} \leq C e^{C|t|} N^{1/2} \left( \eps + \frac{1}{\sqrt{k}} \right)
\ee
for a suitable constant $C>0$, dependent on $\sup_{N}\|W_{N}\|_{H^{2}_{4}}$ but not on the higher Sobolev norms. 

Let $\cU (t;s)$ be the unitary group generated by $h_H (t) = -\eps^2 \Delta + V*\rho_t$, with $\rho_t (x) = N^{-1} \omega_{N,t} (x;x)$. From 
\[ i\eps \partial_t \cU^* (t;0)\, \omega_{N,t}^k \, \cU (t;0) = -\cU^* (t;0) \left[ V* (\rho_t - \rho_t^k), \omega_{N,t}^k \right] \cU (t;0) \]
we have:
\[ \begin{split} \omega_{N,t} - \omega_{N,t}^k =\; &  \cU (t;0) \left( \omega_N - \cU^* (t;0) \omega_{N,t}^k \cU (t;0) \right) \cU^* (t;0) \\ = \; & \cU (t;0) (\omega_N - \omega_N^k) \cU^* (t;0) + \frac{1}{i\eps} \int_0^t ds \, \cU (t;s) \left[ V* (\rho_s - \rho_s^k) , \omega_{N,s}^k \right] \cU^* (t;s) \end{split} \]
Hence
\be\label{eq:eipxomega-omegak} 
\begin{split} \| \omega_{N,t} - \omega_{N,t}^k \|_\text{HS} \leq \; &\| \omega_N - \omega_N^k \|_\text{HS} + \frac{1}{N\eps} \int_0^t ds \int dp \, |\widehat{V} (p)| \, \left| \tr \, e^{-ip \cdot x} (\omega_{N,s} - \omega_{N,s}^k) \right| \, \| [e^{ip \cdot x} , \omega_{N,s}^k ] \|_\text{HS}  
\end{split}
\ee
We start by estimating the commutator in the right-hand side. We have 
\[  [e^{ip \cdot x} , \omega_{N,s}^k ]  = \int_0^1 e^{i\lambda p \cdot x} [ip \cdot x, \omega_{N,s}^k] e^{i(1-\lambda) p \cdot x} \]
By Proposition \ref{prop:prop-comm}, it follows that:
\[ \| [e^{i p \cdot x} , \omega_{N,s}^k ] \|_\text{HS} \leq |p| \| [x, \omega_{N,s}^k ] \|_\text{HS} \leq C|p|e^{C|s|} \left( \| [x, \omega_{N}^k] \|_\text{HS} + \| [\eps \nabla , \omega_{N}^k ] \|_\text{HS} \right) \]
Since 
\[ \begin{split} \| [x, \omega_N^k ] \|_\text{HS} &= \eps N^{1/2} \| \nabla_v W^k_0 \|_2 \leq \eps N^{1/2} \| W_0 \|_{H^1} 
\\ \| [\eps \nabla, \omega_N^k ] \|_\text{HS} &= \eps N^{1/2} \| \nabla_x W_0^k \|_2 \leq \eps N^{1/2} \| W_0 \|_{H^1} \end{split} \]
we conclude that 
\be\label{eq:commest} 
\| [e^{i p \cdot x} , \omega_{N,s}^k ] \|_\text{HS} \leq C N^{1/2} \eps |p| e^{C|s|}
\ee
Then, we are left with estimating the trace on the right-hand side of (\ref{eq:eipxomega-omegak}). To do this, we shall use the following lemma.

\begin{lemma} \label{lem:eipx}
Under the same assumptions of Theorem \ref{cor:relaxed-main}, there exists a constant $C>0$, only depending on $\sup_{N}\| W_{N} \|_{H^{2}_{4}}$ but not on the higher Sobolev norms, such that
\begin{equation}\label{eq:supptr0} 
\sup_{p \in \bR^3} \frac{1}{1+|p|} \left| \tr \, e^{ip \cdot x} (\omega_{N,t} - \omega_{N,t}^k) \right| \leq C e^{C|t|}N \left(\frac{1}{\sqrt{k}} + \eps\right)
\end{equation}
\end{lemma}

Plugging (\ref{eq:commest}), (\ref{eq:supptr0}) into (\ref{eq:eipxomega-omegak}), and using the bound (\ref{eq:diff-ini}) on the difference of the initial data, we get 
\[ \| \omega_{N,t} - \omega_{N,t}^k \|_\text{HS} \leq C e^{C|t|} N^{1/2} \left( \frac{1}{\sqrt{k}} + \eps \right) \]
which concludes the proof of (\ref{eq:HS-step2}). Thus, we are left with the proof of Lemma \ref{lem:eipx}.

\begin{proof}[Proof of Lemma \ref{lem:eipx}] Consider, for an arbitrary $p \in \bR^3$,
\[ \tr\, e^{ip\cdot x} (\omega_{N,t} - \omega_{N,t}^k) = \tr \, \cU^*(t;0) e^{ip \cdot x} \cU (t;0) (\omega_N - \cU^* (t;0) \, \omega_{N,t}^k \, \cU (t;0)) \]
where, as in (\ref{eq:aux-dyn}), $\cU (t;0)$ denotes the unitary group generated by $h_H (t) = -\eps^2 \Delta + V*\rho_t$, with $\rho_t (x) = N^{-1} \omega_{N,t} (x;x)$. {F}rom 
\[ i\eps \partial_t \cU^* (t;0)\, \omega_{N,t}^k \, \cU (t;0) = -\cU^* (t;0) \left[ V* (\rho_t - \rho_t^k), \omega_{N,t}^k \right] \cU (t;0) \]
we find
\begin{equation}\label{eq:stab-expo}
\begin{split} \tr\, e^{ip\cdot x} & (\omega_{N,t}- \omega_{N,t}^k) \\ = \; &\tr \, \cU^* (t;0) e^{ip \cdot x} \cU (t;0) (\omega_N - \omega_N^k) \\ &-\frac{1}{i\eps} \int_0^t \tr \, \cU^* (t;s) e^{ip \cdot x} \cU (t;s) \, \left[ V* (\rho_s - \rho_s^k) , \omega_{N,s}^k \right]\,ds \\
= \;& \tr \, \cU^* (t;0) e^{ip \cdot x} \cU (t;0) (\omega_N - \omega_N^k) \\ &+\frac{1}{i\eps} \int_0^t \int d\wt{p} \, \widehat{V} (\wt{p}) \left( \widehat{\rho}_s (\wt{p}) - \widehat{\rho}_s^k (\wt{p}) \right) \tr \, \cU^* (t;s) e^{ip \cdot x} \cU (t;s) [ e^{i \wt{p} \cdot x}, \omega_{N,s}^k ] 
\end{split} \end{equation}
Since
\[ \widehat{\rho}_s (\wt{p}) - \widehat{\rho}_s^k (\wt{p}) = \frac{1}{N}  \tr \, e^{i\wt{p} \cdot x} (\omega_{N,s}- \omega_{N,s}^k) \]
we conclude that
\[ \begin{split} \Big| \tr \, e^{ip \cdot x} (\omega_{N,t} - &\omega_{N,t}^k) \Big| \\ \leq \; & \left| \tr  \, \cU^*(t;0) e^{ip \cdot x} \cU (t;0) (\omega_N - \omega_N^k) \right| \\ &+ \frac{1}{N\eps} \int_0^t ds \int d\wt{p} \, |\widehat{V} (\wt{p})| \, |\tr \, e^{i\wt{p} \cdot x} (\omega_{N,s}- \omega_{N,s}^k)| \, \left| \tr \, [ \cU^* (t;s) e^{ip\cdot x} \cU (t;s), e^{i\wt{p} \cdot x} ] \omega_{N,s}^k \right|\end{split} \]
and therefore, using the assumption (\ref{eq:pot-hyp}), that 
\begin{equation}\label{eq:supptr} \begin{split}
\sup_{p \in \bR^3} \frac{1}{1+|p|}  \, \Big| \tr \,  e^{ip\cdot x} &(\omega_{N,t} - \omega_{N,t}^k) \Big| \\ \leq \; & \sup_{p\in \bR^3} \frac{1}{1+|p|}\left| \tr \, \cU^* (t;0) e^{ip \cdot x} \cU (t;0) (\omega_N - \omega_N^k) \right| \\ &+ \frac{C}{N\eps} \int_0^t ds \, \sup_{p,\wt{p} \in \bR^3} \frac{1}{(1+|p|)(1+|\wt{p}|)} \, \left| \tr \, \left[ \cU^* (t;s)e^{ip \cdot x} \cU (t;s), e^{i\wt{p} \cdot x} \right] \omega_{N,s}^k \right| 
\\ &\hspace{3cm} \times 
\, \sup_{\wt{p} \in \bR^3} \frac{1}{1+|\wt{p}|} \left| \tr \, e^{i \wt{p}\cdot x} (\omega_{N,s} - \omega_{N,s}^k) \right|  
\end{split} 
\end{equation}
To bound the second term on the r.h.s. of \eqref{eq:supptr} we shall use the following lemma, whose proof is deferred to Sect. \ref{sect:lemmas}.

\begin{lemma}\label{lm:cm1}
Assume that (\ref{eq:ass-pot}) holds true.
Let $\cU (t;s)$ be the unitary evolution generated by the Hartree Hamiltonian $h (t) = -\eps^2 \Delta + (V*\rho_t)$. There exists a constant $C > 0$ such that 
\[ \sup_{\omega,r} \frac{1}{|r|} \left| \tr \, \left[ e^{ir \cdot x}, \cU^* (t;s) e^{ix \cdot p + \eps \nabla \cdot q}  \cU (t;s) \right] \omega \right| \leq \eps (|p|+|q|) e^{C|t-s|} \]
for all $p,q \in \bR^3$. Here, the supremum is taken over $r \in \bR^3$ and over all trace class operators $\omega$ on $L^2 (\bR^3)$ with $\tr \, |\omega| \leq 1$. 
\end{lemma}

It follows from Lemma \ref{lm:cm1} and from $\tr\, |\o^{k}_{N,s}| = N$ that:
\begin{equation}\label{eq:Lbd} \left| \tr\,\left[ \cU^* (t;s) e^{ip \cdot x} \cU (t;s) , e^{i\wt{p} \cdot x} \right]  \omega^k_{N,s} \right| \leq C N \eps |p| |\wt{p}| \, e^{C|t-s|}   
\end{equation}
To bound the first term on the r.h.s. of (\ref{eq:supptr}), we proceed as follows. We choose a function $\chi_< \in C^\infty (\bR^3)$, with $\chi_< (x) = 1$ for $|x| \leq 1$ and $\chi_< (x) = 0$ for $|x| \geq 2$. We set $\chi_> = 1- \chi_<$. For an arbitrary $R \geq 1$, we decompose
\begin{equation} \label{eq:I+II+III} \begin{split} \tr \, \cU^* (t;0) &e^{ip \cdot x} \cU (t;0) (\omega_N - \omega_N^k) \\ = \; &\tr \, \cU^* (t;0) e^{ip \cdot x} \cU (t;0) \chi_< (-\eps^2 \Delta / R) (\omega_N - \omega_N^k) \\ & + \tr \, \cU^* (t;0) e^{ip \cdot x} \cU (t;0) \chi_> (-\eps^2 \Delta / R) (\omega_N - \omega_N^k) \chi_< (-\eps^2 \Delta / R) \\ &+ \tr \, \cU^* (t;0) e^{ip \cdot x} \cU (t;0) \chi_> (-\eps^2 \Delta /R) (\omega_N - \omega_N^k) \chi_> (-\eps^2 \Delta /R) \\ = \; &\text{I} + \text{II} + \text{III} \end{split} 
\end{equation}
To estimate the last term, we observe that 
\begin{equation}\label{eq:III} \begin{split} 
|\text{III}| \leq \; &\tr \, \chi^2_> (-\eps^2 \Delta /R) \omega_N  + \tr \, \chi^2_> (-\eps^2 \Delta /R) \omega^k_N \\ \leq \; &\frac{1}{R} \, \left[ \tr \, (-\eps^2 \Delta) \omega_N + \tr \, (-\eps^2 \Delta) \omega_N^k \right] \\
= \; &\frac{N}{R} \, \left[ \int dxdv \, v^2 W_N (x,v) + \int dxdv \, v^2 W_N^k (x,v) \right] \leq \frac{CN}{R} \end{split} 
\end{equation}
from the assumption $\sup_N \| W_N \|_{H^2_4} < \infty$, and using that $\chi_>(-\eps^2 \Delta /R)\leq (-\eps^2 \Delta /R)$. Next, let us consider the first term on the r.h.s. of (\ref{eq:I+II+III}). We write
\[ \text{I} = \tr \, \cU ^*(t;0) e^{ip \cdot x} \cU (t;0)\chi_< (-\eps^2 \Delta / R) (1+x^2)^{-1} (1+x^2) (\omega_N - \omega_N^k) \]
and we decompose 
\[ \begin{split} \big[ (1+x^2) (\omega_N - \omega_N^k) \big] (x;y) = \; &N (1+x^2) \int dv \left[ W_N \Big( \frac{x+y}{2} , v \Big) - W_N^k \Big( \frac{x+y}{2} , v \Big) \right]  e^{i v \cdot \frac{(x-y)}{\eps}} \\ = \; & D_1 (x;y) + D_2 (x;y) + D_3 (x;y) \end{split} \]
where
\[ D_1 (x;y) = N  \left[ 1+ \left( \frac{x+y}{2} \right)^2 \right] \int dv \left[ W_N \Big( \frac{x+y}{2} , v \Big) - W_N^k \Big( \frac{x+y}{2} , v \Big) \right] e^{iv \cdot \frac{(x-y)}{\eps}} \]
is the Weyl quantization of the function $(1+x^2) (W_N (x,v) - W_N^k (x,v))$ defined on phase-space, while
\[ \begin{split} D_2 (x;y) &= \frac{N \eps^2}{4} \left( \frac{x-y}{\eps} \right)^2 \int dv \left[ W_N \Big( \frac{x+y}{2} , v \Big) - W_N^k \Big( \frac{x+y}{2} , v \Big) \right] e^{iv \cdot \frac{(x-y)}{\eps}} \\
&= \frac{N \eps^2}{4} \int dv \left[ \Delta_v W_N \Big( \frac{x+y}{2} , v \Big) - \Delta_v W_N^k \Big( \frac{x+y}{2} , v \Big) \right] e^{iv \cdot \frac{(x-y)}{\eps}} \end{split} \]
is the Weyl quantization of $(\eps^2/4) (\Delta_v W_N (x,v) - \Delta_v W_N^k (x,v))$ and 
\[ \begin{split} D_3(x;y) &= N \eps \frac{x+y}{2}\cdot \frac{x-y}{\eps}  \int dv \left[ W_N \Big( \frac{x+y}{2} , v \Big) - W_N^k \Big( \frac{x+y}{2} , v \Big) \right] e^{iv \cdot \frac{(x-y)}{\eps}} \\
&= N \eps \frac{x+y}{2}\cdot \int dv \left[ \nabla_v W_N \Big( \frac{x+y}{2} , v \Big) - \nabla_v W_N^k \Big( \frac{x+y}{2} , v \Big) \right] e^{iv \cdot \frac{(x-y)}{\eps}} \end{split} \]
is the Weyl quantization of $\eps x \cdot ( \nabla_v W_N (x,v) - \nabla_v W_N^k (x,v))$. 
We bound the contributions of the three terms $D_1, D_2, D_3$ separately. We begin with
\[\begin{split}  \Big| \tr \, \cU^* (t;0) e^{ip \cdot x} &\cU (t;0)\chi_< (-\eps^2 \Delta / R) (1+x^2)^{-1} D_1 \Big| \\ &= \Big| \tr \, \cU^* (t;0) e^{ip \cdot x} \cU (t;0)\chi_< (-\eps^2 \Delta / R) (1+x^2)^{-1} (1-\eps^2 \Delta)^{-1} (1-\eps^2\Delta) D_1 \Big| \\ &\leq \| (1+x^2)^{-1} (1-\eps^2 \Delta)^{-1} \|_\text{HS} \| (1-\eps^2 \Delta) D_1 \|_\text{HS} \\ &\leq C \sqrt{N} \| (1-\eps^2 \Delta) D_1 \|_\text{HS}  \end{split} \]
where we used that $0\leq \chi_{<}(-\eps^{2}\Delta/R)\leq 1$. We have
\[ \begin{split} \big[ (1-\eps^2\Delta) &D_1 \big] (x;y) \\ &= \; N (1-\eps^2\Delta_x) \left[ 1+ \left( \frac{x+y}{2} \right)^2 \right] \int dv \left[ W_0 \Big( \frac{x+y}{2} , v \Big) - W_0^k \Big( \frac{x+y}{2} , v \Big) \right] e^{iv \cdot \frac{(x-y)}{\eps}} \end{split} \]
It is not difficult to see that:
\[ \begin{split} \| (1-\eps^2\Delta) D_1 \|_\text{HS} \leq \; &C \sqrt{N} \| (1+x^2)(1+v^2) (W_N - W_N^k) \|_2 \\ &+ C \sqrt{N} \eps \| W_N - W_N^k \|_{H^1_1} + C \sqrt{N} \eps^2 \| W_N - W_N^k \|_{H^2_2} \\ \leq \; &C \sqrt{N} \left( \frac{1}{\sqrt{k}} + \eps \right) \end{split} \]
Therefore
\begin{equation}\label{eq:D1} \Big| \tr \, \cU^* (t;0) e^{ip \cdot x} \cU (t;0)\chi_< (-\eps^2 \Delta / R) (1+x^2)^{-1} D_1 \Big| \leq C N \left( \frac{1}{\sqrt{k}} + \eps \right) 
\end{equation}
The contribution of $D_2$, on the other hand, can be controlled by 
\[\begin{split} \Big| \tr \, \cU^* (t;0) e^{ip \cdot x} &\cU (t;0)\chi_< (-\eps^2 \Delta / R) (1+x^2)^{-1} D_2 \Big| \\ 
\leq \; & \| \chi_< (-\eps^2 \Delta /R) (1+x^2)^{-1} \|_\text{HS} \| D_2 \|_\text{HS} \leq C \eps^2 \sqrt{N} \| W_0 \|_{H^2} \| \chi_< (-\eps^2 \Delta /R) (1+x^2)^{-1} \|_\text{HS}  \end{split} \]
where
\[ \begin{split} \| \chi_< (-\eps^2 \Delta /R)& (1+x^2)^{-1} \|_\text{HS}^2 \\ &= \tr \, (1+x^2)^{-1} (1-\eps^2 \Delta)^{-1} \chi^2_< (-\eps^2 \Delta / R) (1-\eps^2 \Delta)^{-1} (1+x^2)^{-1} \\
&= \tr \, (1+x^2)^{-1} (1-\eps^2 \Delta)^{-1} (1-\eps^2 \Delta)^2 \chi^2_< (-\eps^2 \Delta / R) (1-\eps^2 \Delta)^{-1} (1+x^2)^{-1} \\ &
\leq C R^2 \| (1+x^2)^{-1} (1-\eps^2 \Delta)^{-1} \|^2_\text{HS} \\ &\leq C R^2 N \end{split} \]
Hence, we conclude that 
\begin{equation}\label{eq:D2} \Big| \tr \, \cU^* (t;0) e^{ip \cdot x} \cU (t;0)\chi_< (-\eps^2 \Delta / R) (1+x^2)^{-1} D_2 \Big| \leq C N R \, \eps^2  
\end{equation}
We proceed similarly to bound the contribution of the term $D_3$. We find
\[ \begin{split} \Big| \tr \, \cU^* (t;0) e^{ip \cdot x} \cU (t;0)\chi_< (-\eps^2 \Delta / R) (1+x^2)^{-1} D_3 \Big| &\leq  \| \chi_< (-\eps^2 \Delta /R) (1+x^2)^{-1} \|_\text{HS} \| D_3 \|_\text{HS} \\ & \leq C N R \, \eps \|W_N - W_N^k\|_{H^1_1} \\
&\leq \frac{C N R \, \eps}{\sqrt{k}} \end{split} \]
where in the last step we used (\ref{eq:diff-Win}). The last equation, combined with (\ref{eq:D1}), (\ref{eq:D2}) implies that 
\[ |\text{I}| \leq C N \left( \frac{1}{\sqrt{k}} + \eps + R \eps^2 + \frac{R \eps}{\sqrt{k}} \right)  \]
Analogously, one can show that the same estimate holds for the term $\text{II}$ on the r.h.s. of (\ref{eq:I+II+III}) as well (in this case, we introduce the identity $(1+x^2) (1+x^2)^{-1}$ on the right of the difference $\omega_N - \omega_N^k$ and we use the cyclicity of the trace). With (\ref{eq:III}), we conclude that
\[  \left| \tr \, \cU^* (t;0) e^{ip \cdot x} \cU (t;0) (\omega_N - \omega_N^k) \right| \leq CN \left( \frac{1}{\sqrt{k}} + \eps + R \eps^2 + \frac{R \eps}{\sqrt{k}} + \frac{1}{R} \right)  \]
Choosing $R = \eps^{-1}$, we obtain
\[ \left| \tr \, \cU^* (t;0) e^{ip \cdot x} \cU (t;0) (\omega_N - \omega_N^k) \right| \leq CN \left( \frac{1}{\sqrt{k}} + \eps \right) \]

Inserting this bound and (\ref{eq:Lbd}) in 
(\ref{eq:supptr}) and applying Gronwall's lemma, we obtain
\begin{equation} \sup_{p \in \bR^3} \frac{1}{1+|p|} \left| \tr \, e^{ip \cdot x} (\omega_{N,t} - \omega_{N,t}^k) \right| \leq C e^{C|t|}N \left(\frac{1}{\sqrt{k}} + \eps\right)
\end{equation}
which concludes the proof of (\ref{eq:supptr0}).
\end{proof}

\medskip

\noindent{\bf Comparison of $\tilde\omega^{k}_{N,t}$ with $\tilde\omega_{N,t}$.} We now compare the Vlasov evolution of the regularized initial data with the Vlasov evolution of the original data. We claim that there exists a constant $C > 0$, depending on $\sup_{N}\| W_{N} \|_{H^{2}_{4}}$ but not on the higher Sobolev norms, such that: 
\be\label{eq:diffvlasov}
\| \tilde\omega_{N,t} - \tilde\omega_{N,t}^{k} \|_{\rm HS} = N^{1/2}\big\| \wt{W}_{N,t} - \wt{W}_{N,t}^k \big\|_2 \leq N^{1/2}\frac{Ce^{C|t|}}{\sqrt{k}}
\ee
To prove this, let \begin{equation}\label{eq:start-L1} \wt{\rho}_t (x) = \int dv \, \wt{W}_{N,t} (x,v) \qquad \text{and } \quad  
\wt{\rho}_t^k (x) = \int dv \, \wt{W}_{N,t}^k (x,v)\end{equation}
be the densities associated with $\wt{W}_{N,t}$ and $\wt{W}_{N,t}^k$. For $t \in \bR$, we denote by $(X_t (x,v), V_t (x,v))$ and by $(X_t^k (x,v), V_t^k (x,v))$ the flows satisfying the differential equations 
\begin{equation}\label{eq:Xtxv1}
\left\{ \begin{array}{ll} \dot{X}_t (x,v) &= 2 V_t (x,v) \\ \dot{V}_t (x,v) &= -\nabla (V*\wt{\rho}_t) (X_t (x,v)) \end{array} \right. 
\end{equation}
and
\begin{equation}\label{eq:Xtxv2} \left\{ \begin{array}{ll} \dot{X}^k_t (x,v) &= 2 V^k_t (x,v) \\ \dot{V}^k_t (x,v) &= -\nabla (V*\wt{\rho}^k_t) (X^k_t (x,v)) \end{array} \right. 
\end{equation}
with initial data given by, respectively, $X_{0}(x,v) = X_{0}^{k}(x,v) = x$, $V_{0}(x,v) = V_{0}^{k}(x,v) = v$. We compare the two flows $(X_t, V_t)$ and $(X_t^k, V_t^k)$. We have
\[ \begin{split} \frac{d}{dt} (X_t - X_t^k) (x,v) &= 2(V_t - V_t^k) (x,v) \\
\frac{d}{dt} (V_t - V_t^k) (x,v) &= -\nabla (V*\wt{\rho}_t) (X_t (x,v)) + \nabla (V* \wt{\rho}_t^k) (X_t^k (x,v)) \end{split} \]
and therefore
\[ \begin{split} \left| \frac{d}{dt} (X_t - X_t^k) (x,v) \right| &\leq 2\left| V_t (x,v) - V_t^k (x,v) \right| \\
\left| \frac{d}{dt} (V_t - V_t^k) (x,v) \right| &\leq C \| \wt{\rho}_t - \wt{\rho}_t^k \|_1 + C \left| X_t (x,v) - X_t^k (x,v) \right|
\end{split} \]
where we used the assumption (\ref{eq:pot-hyp}). Gronwall's lemma implies that
\begin{equation}\label{eq:XtXtk} 
\begin{split} 
\left| X_t (x,v) - X_t^k (x,v) \right| + \left| V_t (x,v) - V_t^k (x,v) \right| &\leq C e^{C|t|} \int_0^t ds \, \| \wt{\rho}_s - \wt{\rho}_s^k \|_1 \\
& \leq C e^{C|t|} \int_0^t ds \, \| W_s - W_s^k \|_1
\end{split} 
\end{equation}
We will also need to control the difference between derivatives of the flows $(X_t (x,v), V_t (x,v))$ and $(X^k_t (x,v), V_t^k (x,v))$. Integrating the flow equations (\ref{eq:Xtxv1}), (\ref{eq:Xtxv2}), we have
\begin{equation}\label{eq:nabXt} \begin{split} 
\nabla_x X_t (x,v) &= 1  + 2\int_0^t \nabla_x V_s (x,v) ds \\
\nabla_x V_t (x,v) &= -\int_0^t \nabla^2 (V*\wt{\rho}_s) (X_s (x,v)) \cdot \nabla_x X_s (x,v)\,ds \end{split}
\end{equation}
which implies that 
\[ \begin{split} |\nabla_x X_t (x,v)| &\leq 1 + 2\int_0^t ds\, |\nabla_x V_s (x,v)|  \\
|\nabla_x V_t (x,v)| &\leq C \int_0^t ds \, |\nabla_x X_s (x,v)| \end{split} \]
and hence, by Gronwall's lemma, that 
\begin{equation}\label{eq:nab-bd} |\nabla_x X_t (x,v)| + |\nabla_x V_t (x,v)| \leq e^{C|t|}
 \end{equation}
Analogously, we also find
\begin{equation}\label{eq:nab-bd2} |\nabla_v X_t (x,v)| + |\nabla_v V_t (x,v)| \leq e^{C|t|} 
\end{equation}
and 
\[
\begin{split}
|\nabla_x X_t^k (x,v)| + |\nabla_x V_t^k (x,v)| &\leq e^{C|t|} \\
|\nabla_v X_t^k (x,v)| + |\nabla_x V_t^k (x,v)| &\leq e^{C|t|}
\end{split}
\]
Moreover, from (\ref{eq:nabXt}), we obtain 
\[ \begin{split} \Big| \nabla_x X_t (x,v) - &\nabla_x X_t^k (x,v) \Big|  \leq \; 2\int_0^t ds \, \left| \nabla_x V_s (x,v) - \nabla_x V_s^k (x,v) \right|\end{split} \]
and thus
\[ \begin{split}
\Big| \nabla_x V_t (x,v) - &\nabla_x V_t^k (x,v) \Big| \\ \leq \; &\int_0^t ds \, \left| 
\nabla^2 (V*\wt{\rho}_s) (X_s (x,v)) \cdot \nabla_x X_s (x,v) - \nabla^2 (V*\wt{\rho}_s^k)(X_s^k (x,v)) \cdot \nabla_x X_s^k (x,v) \right| \\ \leq\; &
C \int_0^t ds \, \| \wt{\rho}_s - \wt{\rho}_s^k \|_1 + 
C\int_0^t ds \, |X_s (x,v) - X_s^k (x,v)| 
|\nabla_x X_s (x,v)| \\ &+ C \int_0^t ds \, |\nabla_x X_s (x,v) - \nabla_x X_s^k (x,v)| 
\end{split} \]
To get the second inequality, we used that 
\be\label{eq:L1toH14}
\| \nabla^{3} V*\tilde \rho_{s} \|_{\infty} \leq \| \nabla^{2}V \|_{\infty}\| \nabla \widetilde{W}_{N} \|_{1}\leq Ce^{C|s|}\| W_{N} \|_{H^{1}_{4}}
\ee
Using (\ref{eq:nab-bd}) and (\ref{eq:XtXtk}), and applying Gronwall's lemma, we conclude that 
\begin{equation}\label{eq:nabx-bd} \begin{split} 
\left| \nabla_x X_t (x,v) - \nabla_x X_t^k (x,v) \right| + &\left| \nabla_x V_t (x,v) - \nabla_x V_t^k (x,v) \right| \\
\leq\; &
Ce^{C|t|} \int_0^t ds \, \| \wt{\rho}_s - \wt{\rho}_s^k \|_1 + 
C e^{C|t|} \int_0^t ds \int_0^s dr \, \|\wt{\rho}_r - \wt{\rho}_r^k \|_1 
\end{split} 
\end{equation}
Similarly, we can also show that 
\begin{equation}\label{eq:nabv-bd}
\begin{split} 
\left| \nabla_v X_t (x,v) - \nabla_v X_t^k (x,v) \right| + &\left| \nabla_v V_t (x,v) - \nabla_v V_t^k (x,v) \right| \\ \leq\; &
C e^{C|t|} \int_0^t ds \, \| \wt{\rho}_s - \wt{\rho}_s^k \|_1 + 
C e^{C|t|} \int_0^t ds \, \int_0^s dr \, \|\wt{\rho}_r - \wt{\rho}_r^k \|_1 
\end{split} 
\end{equation}

Next, we control the $L^1$ norm of the difference $\wt{W}_{N,t} - \wt{W}_{N,t}^k$. To this end, we write
\[ \begin{split} \| \wt{W}_{N,t} - \wt{W}_{N,t}^k \|_1 =\; &\int dx dv \big| \wt{W}_{N,t} (x,v) - \wt{W}_{N,t}^k (x,v) \big| \\ =\; & \int dx dv \big| W_N (X_{-t} (x,v), V_{-t} (x,v)) - W_N^k (X^k_{-t} (x,v), V^k_{-t} (x,v)) \big| 
\\ \leq \; & \int dx dv \big| W_N (X_{-t} (x,v), V_{-t} (x,v)) - W^k_N (X_{-t} (x,v), V_{-t} (x,v)) \big| \\ &+ \int dx dv \big| W_N^k (X_{-t} (x,v) , V_{-t} (x,v)) - W_N^k (X_{-t}^k (x,v), V_{-t}^k (x,v)) \big|
\end{split} 
\]
Using that the Vlasov dynamics preserves the volume in phase-space, we get:
\[
\begin{split}
& \| \wt{W}_{N,t} - \wt{W}_{N,t}^k \|_1 \leq  \| W_N - W_N^k \|_1 \\
&\quad + \int dx dv \left| \int_0^1 d\lambda \, \frac{d}{d\lambda} W_N^k \left( \lambda (X_{-t} (x,v), V_{-t}(x,v)) + (1-\lambda) (X_{-t}^k (x,v), V_{-t}^k (x,v)) \right) \right| \\ &\quad \leq \| W_N - W_N^k \|_1 + \int d\lambda dx dv \Big[ \big| (\nabla_x W_N^k) \left( \tilde{x} (x,v,\lambda), \tilde{v} (x,v,\lambda)\right) \big|  \big|X_{-t} (x,v) - X_{-t}^k (x,v) \big| \\
&\qquad + \big| (\nabla_v W_N^k) \left( \tilde{x} (x,v,\lambda), \tilde{v} (x,v,\lambda)\right)\big|  \big|V_{-t} (x,v) - V_{-t}^k (x,v)\big|\Big]
\end{split} \]
where we introduced the notation 
\[
\tilde{x} (x,v,\lambda) := \lambda X_{-t} (x,v) + (1-\lambda) X^k_{-t} (x,v)\;,\quad \tilde{v} (x,v,\lambda) := \lambda V_{-t} (x,v) + (1-\lambda) V_{-t}^k (x,v)
\]
From (\ref{eq:XtXtk}), we obtain 
\begin{equation}\label{eq:W-Wk1} 
\begin{split} 
\| \wt{W}_{N,t} - &\wt{W}_{N,t}^k \|_1  \leq \| W_N - W_N^k \|_1  + C\int_0^t ds\, e^{C|s|} \| \wt{W}_{N,s} - \wt{W}_{N,s}^k \|_1 \\ &\times \Big[ \int d\lambda dx dv\, \big| (\nabla_x W^k_N) \left( \tilde{x} (x,v,\lambda), \tilde{v} (x,v,\lambda)\right) \big| + \big| (\nabla_v W^k_N) \left( \tilde{x} (x,v,\lambda), \tilde{v} (x,v,\lambda)\right) \big| \Big] 
\end{split} 
\end{equation}
We observe that
\begin{equation}\label{eq:Jintro} \int dx dv \, \big| (\nabla_x W_N^k) \left( \tilde{x} (x,v,\lambda), \tilde{v} (x,v,\lambda) \right) \big| = \int \big |(\nabla_x W_N^k) (\tilde{x}, \tilde{v})\big| \frac{1}{|J(\tilde{x}, \tilde{v})|} \, d\tilde{x} d\tilde{v} 
\end{equation}
with the Jacobian
\[ J = \det \left[ \lambda \left( \begin{array}{ll} \nabla_x X_{-t} & \nabla_x V_{-t} \\ \nabla_v X_{-t} & \nabla_v V_{-t} \end{array} \right) + (1-\lambda) \left( \begin{array}{ll} \nabla_x X^k_{-t} & \nabla_x V^k_{-t} \\ \nabla_v X^k_{-t} & \nabla_v V^k_{-t} \end{array} \right) \right] \]

To estimate the determinant $J(\tilde{x}, \tilde{v})$ in (\ref{eq:Jintro}), we proceed as follows. For a fixed constant $C >0$ (that later will be chosen large enough), let us define $t^{*}>0$ such that:
\begin{equation}\label{eq:t*} \frac{C^3 \, e^{2C t^*}}{\sqrt{k}} = 1/2 
\end{equation}
We claim that, for all $|t| < t^*$, 
\begin{equation}\label{eq:clai} \| \wt{W}_{N,t} - \wt{W}_{N,t}^k \|_1 \leq C \frac{e^{C|t|}}{\sqrt{k}} 
\end{equation}
We prove (\ref{eq:clai}) for $t > 0$ (the case of $t < 0$ can be handled similarly, of course). We set 
\be\label{eq:deft0}
t_0 = \inf \Big\{ t > 0 : \| \wt{W}_{N,t} - \wt{W}_{N,t}^k \|_1 >  \frac{Ce^{C|t|}}{\sqrt{k}} \Big\} 
\ee
and we proceed by contradiction, assuming that $t_0 < t^*$. At time $t=0$, we have:
\begin{equation}\label{eq:W0k} \begin{split} 
\| W_N - W_N^k \|_1 \leq \; &\int dx dv dx' dv' g_k (x-x', v-v') |W_N (x,v) - W_N (x',v')| \\ = \; &\frac{1}{(2\pi)^3} \int dx dv dr ds \, e^{-(r^2 + s^2)/2} \left| W_N \left( x+ \frac{r}{\sqrt{k}} , v + \frac{s}{\sqrt{k}} \right) - W_N (x,v) \right| \\ \leq \; &\frac{1}{(2\pi)^{3}}\int dx dv dr ds \int_0^1 d\lambda \, e^{-(r^2 + s^2)/2} \\ &\quad \times  \left[ \frac{|r|}{\sqrt{k}} \left| \nabla_x W_N \left( x + \lambda \frac{r}{\sqrt{k}} , v + \lambda \frac{s}{\sqrt{k}} \right) \right| +  \frac{|s|}{\sqrt{k}} \left| \nabla_v W_N \left( x + \lambda \frac{r}{\sqrt{k}} , v + \lambda \frac{s}{\sqrt{k}} \right) \right| \right] \\ \leq \; & \frac{8\pi}{(2\pi)^{3}}\frac{1}{\sqrt{k}} (\| \nabla_x W_N \|_1 + \| \nabla_v W_N \|_1) \leq \frac{C}{\sqrt{k}} \| W_N \|_{H^1_4} \leq \frac{\widetilde C}{\sqrt{k}}
\end{split} 
\end{equation}
where in the last line we estimated the $L^{1}$-norms by proceeding as in (\ref{eq:L1toH14}). Since, moreover, $t\to \wt{W}_{N,t}$ and $t\to \wt{W}_{N,t}^k$ are continuous in the $L^1$-topology, by choosing $C = 2\tilde C$ in Eq. (\ref{eq:deft0}), we conclude that $t_0 > 0$. The continuity property is a standard fact (see e.g. \cite{Dobr79}).

By definition, for $0 \leq t \leq t_0$, we have (\ref{eq:clai}) and therefore, from (\ref{eq:nabx-bd}) and (\ref{eq:nabv-bd}),  
\[ \big| \nabla_x X_{-t} (x,v) - \nabla_x X^k_{-t} (x,v) \big| + \big| \nabla_x V_{-t} (x,v) - \nabla_x V_{-t}^k (x,v) \big| \leq C^2 \frac{e^{2C|t|}}{\sqrt{k}} \]
and 
\[ \big| \nabla_v X_{-t} (x,v) - \nabla_v X^k_{-t} (x,v) \big| + \big| \nabla_v V_{-t} (x,v) - \nabla_v V_{-t}^k (x,v) \big| \leq C^2 \frac{e^{2C|t|}}{\sqrt{k}} \]
Writing 
\[ J(\wt{x}, \wt{v}) = \det \left[\lambda \left(\begin{array}{ll} \nabla_x X_{-t} & \nabla_x V_{-t} \\ \nabla_v X_{-t} & \nabla_v V_{-t} \end{array} \right) + (1-\lambda) \left(\begin{array}{ll} \nabla_x X^k_{-t} - \nabla_x X_t & \nabla_x V^k_{-t} - \nabla_x V_{-t} \\ \nabla_v X^k_{-t} - \nabla_v X_{-t} & \nabla_v V^k_{-t} - \nabla_v V_{-t} \end{array} \right) \right] \]
and using that 
\[ \left| \det \left( \begin{array}{ll} \nabla_x X_{-t}& \nabla_x V_{-t} \\ \nabla_v X_{-t} & \nabla_v V_{-t} \end{array} \right) \right| = 1 \] 
we conclude that
\[ \left| |J(\wt{x} , \wt{v})| - 1 \right| \leq C^3 \frac{e^{3C|t|}}{\sqrt{k}} \]
if the constant $C > 0$ is large enough. {F}rom (\ref{eq:t*}), and from the assumption $t_0 < t^*$, we conclude that
\[ |J(\wt{x}, \wt{v})| > 1/2 \]
for all $0 \leq t \leq t_0$. Eq. (\ref{eq:Jintro}) implies:
\[ \int dx dv \, \big| (\nabla_x W_N^k) \left( \tilde{x} (x,v,\lambda) , \tilde{v} (x,v,\lambda)\right) \big| \leq 2 \int d\tilde{x} d\tilde{v}\, \big|\nabla_x W_N^k (\tilde{x}, \tilde{v}) \big| \leq C \| W_N^k \|_{H^1_{4}} \leq C \|W_N \|_{H^1_{4}} \]
for all $0 < t \leq t_0$. Similarly, we obtain 
\[ \int dx dv \, \big|  (\nabla_v W_N^k) \left(\tilde{x} (x,v,\lambda) , \tilde{v} (x,v,\lambda) \right) \big| \leq C \| W_N \|_{H^1_{4}} \]
Plugging the last two bounds in the r.h.s. of (\ref{eq:W-Wk1}), we find that 
\[ \| \wt{W}_{N,t} - \wt{W}_{N,t}^k \|_1 \leq \| W_N - W_N^k \|_1 + C \int_0^t ds\, e^{C|s|} \| \wt{W}_{N,s} - \wt{W}_{N,s}^k \|_1 \]
for all $0 \leq t \leq t_0$. Eq. (\ref{eq:W0k}) and Gronwall's lemma imply that, if the constant $C > 0$ is sufficiently large, 
\[ \| \wt{W}_{N,t} - \wt{W}_{N,t}^k \|_1 \leq C \frac{e^{C|t|}}{\sqrt{k}} \]
for all $0 \leq t \leq t_0$, in contradiction with the definition of $t_0$. This shows that $t_0 > t^*$. Repeating the same argument for $t<0$, we obtain that
\begin{equation}\label{eq:concl-L1} \| \wt{W}_{N,t} - \wt{W}_{N,t}^k \|_1 \leq C \frac{e^{C|t|}}{\sqrt{k}} \end{equation}
for all $|t| < t^*$. {F}rom (\ref{eq:XtXtk}), we also find that  
\begin{equation}\label{eq:XtXkt-2} |X_t (x,v) - X^k_t (x,v)| + |V_t (x,v) - V_t^k (x,v)| \leq C^2 \frac{e^{2C|t|}}{\sqrt{k}} 
\end{equation}
for all $|t| < t^*$. Moreover, Eqs. (\ref{eq:nabx-bd}) and (\ref{eq:nabv-bd}) imply that
\begin{equation}\label{eq:J2} |J(\wt{x}, \wt{v})| \geq 1/2 
\end{equation}
for all $|t| < t^*$ and for all $\wt{x}, \wt{v} \in \bR^3$. 

Finally, we control the difference $\wt{W}_{N,t} - \wt{W}_{N,t}^k$ in the $L^2$-norm. To this end, we observe that
\[ 
\begin{split} 
\| \wt{W}_{N,t} - \wt{W}_{N,t}^k \|_2^2 = \; & \int dx dv \, \big| W_N (X_{-t} (x,v), V_{-t} (x,v)) - W_N^k (X_{-t}^k (x,v), V_{-t}^k (x,v)) \big|^2 \\
\leq \; &2 \int dx dv \, \big| W_N (X_{-t} (x,v), V_{-t} (x,v)) - W_N^k (X_{-t} (x,v) V_{-t} (x,v)) \big|^2 \\
&+2 \int dx dv \, \big| W_N^k (X_{-t} (x,v), V_{-t} (x,v)) - W_N^k (X^k_{-t} (x,v), V^k_{-t} (x,v)) \big|^2
\end{split}
\]
Using that the Vlasov dynamics preserves the phase-space volume, we get, for all $|t|<t^{*}$:
\[
\begin{split}
\| \wt{W}_{N,t} - \wt{W}_{N,t}^k \|_2^2 \leq \; &2 \| W_N - W_N^k \|_2^2 \\ &+ 2 \int_0^1 d\lambda \int dx dv \, \Big\{ \big| (\nabla_x W_N^k) \left(\tilde{x} (x,v,\lambda) , \tilde{v} (x,v,\lambda) \right) \big|^2 \big|X_{-t} (x,v) - X_{-t}^k (x,v)\big|^2  \\ &\hspace{3cm} + \big| (\nabla_v W_N^k) \left(\tilde{x} (x,v,\lambda) , \tilde{v} (x,v,\lambda)\right) \big|^2 \big|V_{-t} (x,v) - V_{-t}^k (x,v)\big|^2 \Big\}\\
\leq \; & 2 \| W_N - W_N^k \|_2^2 \\ &+2\frac{C^4 e^{4C|t|}}{k} \int_0^1 d\lambda \int d\tilde x d\tilde v \, \big[ \big| (\nabla_x W_N^k) (\tilde{x}, \tilde{v}) \big|^2 + \big| (\nabla_v W_N^k) (\tilde{x}, \tilde{v}) \big|^2 \big] \frac{1}{|J (\tilde{x}, \tilde{v})|}  \\
\leq \; &10\frac{C^4 e^{4C|t|}}{k} \| W_N \|_{H^1}^2
\end{split} 
\]
To get the first inequality we used the estimate (\ref{eq:XtXkt-2}), while to get the last one we used (\ref{eq:J2}). By definition of $t^*$, we conclude that, after an appropriate change of the constant $C > 0$, 
\[ \|\wt{W}_{N,t} - \wt{W}_{N,t}^k \|_2 \leq \frac{C e^{C|t|}}{\sqrt{k}} \]
for all $t \in \bR$ (recall that the bounds $\| \wt{W}_{N,t} \|_2 ,\, \| \wt{W}_{N,t}^k \|_2 \leq C$ are trivial, since the Vlasov equation preserves the $L^{p}$ norms). This concludes the proof of (\ref{eq:diffvlasov}).

\medskip

\noindent{\bf Proof of Theorem \ref{cor:relaxed-main}.} We have, using (\ref{eq:HS-step1}), (\ref{eq:HS-step2}), (\ref{eq:diffvlasov}):
\be
\begin{split}
\| \omega_{N,t} - \wt{\omega}_{N,t} \|_{\text{HS}} &\leq \|  \omega_{N,t} - \omega_{N,t}^{k} \|_{\text{HS}} + \| \omega_{N,t}^{k} - \wt{\omega}_{N,t}^{k} \|_{\text{HS}} + \| \wt{\omega}_{N,t}^{k} - \wt{\omega}_{N,t} \|_{\text{HS}} \\
&\leq CN^{1/2} \left(\eps + \frac{1}{\sqrt{k}} \right) \exp (C \exp (C|t|)) \, \left(1+ \sum_{j=1}^3(\eps \sqrt{k})^j\right)
\end{split}
\ee
for a constant $C>0$ that depends on $\sup_{N} \| W_{N} \|_{H^{2}_{4}}$ but not on the higher Sobolev norms. Choosing $k = \eps^{-2}$, we conclude that 
\[ \|  \omega_{N,t} - \wt{\omega}_{N,t} \|_\text{HS} \leq C N^{1/2} \eps \, \exp (C \exp (C|t|)) \]
as claimed. This concludes the proof of Theorem \ref{thm:HS}. \qed

\medskip

\noindent{\bf Proof of Theorem \ref{thm:new}.} Let $\wt{W}_{N,t}$ be the solution of the Vlasov equation with initial data $W_N$. We estimate
\begin{equation}\label{eq:proofnew} \| W_{N,t} - W_t \|_2 \leq \| W_{N,t} - \wt{W}_{N,t} \|_2 + \| \wt{W}_{N,t} - W_t \|_2 \end{equation}
The first term can be bounded by Theorem \ref{cor:relaxed-main}. In particular, (\ref{eq:WWNt}) implies that 
\begin{equation}\label{eq:proofnew-2} \| W_{N,t} - \wt{W}_{N,t} \|_2 \leq C \eps \exp (C \exp (C|t|)) 
\end{equation} 
As for the second term on the r.h.s. of (\ref{eq:proofnew}), we have to compare two solutions of the Vlasov equation, with slightly different initial data. But this is exactly what we did in Step 3 of the proof of Theorem \ref{cor:relaxed-main}. The only ingredients that we used there were a bound for the $L^1$ and for $L^2$ norm of the difference of the initial data. Now, by assumption we have $\| W_{0} - W_{N} \|_1 \leq \kappa_{N,1}$, $\| W_0 - W_N \|_2 \leq \kappa_{N,2}$ and $\| W_0 \|_{H^2_4} \leq C$. 
Therefore, the arguments used in Step 3 of Section \ref{sec:proof1b} imply that 
\[ \| \wt{W}_{N,t} - W_t \|_2 \leq C (\kappa_{N,1} + \kappa_{N,2}) e^{C|t|} \]
Together with (\ref{eq:proofnew-2}), we conclude that
\[ \| W_{N,t} - W_t \|_2 \leq C \eps \exp (C \exp (C|t|)) +  C (\kappa_{N,1} + \kappa_{N,2}) \exp (C |t|) \, . \]

\qed


\section{Convergence for the expectation of semiclassical observables} 
\label{sect:proof2}
\setcounter{equation}{0}

Here we prove Theorem \ref{thm:main2} and Theorem \ref{thm:new2}. To show Theorem \ref{thm:main2}, we make first the additional assumption that the Wigner transforms $W_N$ of the fermionic operators $\omega_{N}$ are so that $\sup_N\| W_N \|_{H^4_4} < \infty$; later, we will relax this assumption with an approximation argument.

\medskip

\noindent{\bf Case $\sup_{N}\|W_{N}\|_{H^{4}_{4}}<\infty$.} We use the expression \eqref{eq:omega-omegatilde} for the difference $\omega_{N,t}- \wt{\omega}_{N,t}$ to write 
\be
\begin{split}\label{eq:trace}
\tr\, e^{ip \cdot x + q \cdot \eps \nabla} \, (\o_{N,t}-\wt\o_{N,t})= \; & \frac{1}{\e}\int_0^t \tr\, e^{ip \cdot x + q \cdot \eps \nabla} \,\mathcal{U}(t;s) [ V*(\rho_s - \wt\rho_s), \wt\o_{N,s} ]\,\mathcal{U}^*(t;s)  ds\\
&+ \frac{1}{\e}\int_0^t \tr\, e^{ip \cdot x + q \cdot \eps \nabla} \,\mathcal{U}(t;s) B_s \, \mathcal{U}^*(t;s)  ds\,
\end{split}
\ee
with $B_s$ as defined in (\ref{eq:Bt}). We start by considering the first term on the r.h.s. of \eqref{eq:trace}. We have 
\be\label{eq:secondterm}
\begin{split}
\tr\, e^{ip\cdot x + q \cdot \eps \nabla} \,&\mathcal{U}(t;s) [ V*(\rho_s - \wt\rho_s), \wt\o_{N,s} ]\,\mathcal{U}^* (t;s)   \\ = \; &\int dz \, (\rho_s (z) - \wt{\rho}_s (z)) \tr \, e^{ip\cdot x + q \cdot \eps \nabla} \,\mathcal{U}(t;s) [ V(x-z), \wt\o_{N,s} ]\,\mathcal{U}^*(t;s) \\
= \; &\int dk\,\widehat{V} (k) \int dz \, e^{-ik \cdot z} (\rho_s (z) - \wt{\rho}_s (z)) \tr \, e^{ip\cdot x + q \cdot \eps \nabla} \,\mathcal{U}(t;s) [ e^{ik \cdot x}, \wt\o_{N,s} ]\,\mathcal{U}^* (t;s) \\
= \; &\frac{1}{N} \int dk\,\widehat{V} (k)  \, \tr \, e^{-ik \cdot z} (\o_{N,s}  - \wt{\o}_{N,s}) \, \tr \, e^{ip\cdot x + q \cdot \eps \nabla} \,\mathcal{U}(t;s) [ e^{ik \cdot x}, \wt\o_{N,s} ]\,\mathcal{U}^* (t;s) 
\end{split} 
\ee
Hence
\begin{equation}\label{eq:su2} \begin{split} 
\Big| \tr \, e^{ip\cdot x + q \cdot \eps \nabla} \,&\mathcal{U}(t;s) [ V*(\rho_s - \wt\rho_s), \wt\o_{N,s} ]\,\mathcal{U} (t;s)  ds\, \Big| \\ \leq \; &\frac{1}{N}  \int dk \, |\widehat{V} (k)| \left| \tr \, e^{ik \cdot x} (\o_{N,s} - \wt{\o}_{N,s})\right| \, \left| \tr \, e^{ip\cdot x + q \cdot \eps \nabla} \,\mathcal{U}(t;s) [ e^{ik \cdot x}, \wt\o_{N,s} ]\,\mathcal{U}^* (t;s) \right|  \\
\leq \;&  \frac{C \, \tr \, |\wt{\omega}_{N,s}|}{N} \sup_{k \in \bR^3} \frac{1}{(1+|k|)^2} \left| \tr \, e^{ik \cdot x} (\o_{N,s} - \wt{\o}_{N,s})\right| \\ &\hspace{3cm} \times  \sup_{\o,k} \frac{1}{|k|} \left| \tr \,[e^{ik \cdot x}, \mathcal{U}^* (t;s)e^{ip\cdot x + q \cdot \eps \nabla} \,\mathcal{U}(t;s)] \, \o \,\right|
\end{split} 
\end{equation}
where we used the assumption (\ref{eq:ass-pot}) and where the supremum is taken over all $k \in \bR^3$ and all $\o$ with $\tr \, |\o| \leq 1$.
{F}rom Lemma \ref{lm:cm1}, we obtain 
\begin{equation}\label{eq:gronw1}  
\begin{split}
\Big| \tr \, e^{ip\cdot x + q \cdot \eps \nabla} \mathcal{U}(t;s) & [ V*(\rho_s - \wt\rho_s), \wt\o_{N,s} ]\,\mathcal{U} (t;s)  ds\, \Big|\\ & \leq \frac{C \, \tr \, |\wt{\omega}_{N,s}|}{N} \, \eps (|p|+|q|) e^{C|t-s|} \sup_{k} \frac{1}{(1+|k|)^2} \left| \tr \, e^{ik \cdot x} (\omega_{N,s} - \wt{\omega}_{N,s} ) \right| 
\end{split} 
\end{equation}

Consider now the second term on the r.h.s. of (\ref{eq:trace}). 
By the cyclicity of the trace, we find  
\be\label{eq:trace-third}
\tr\, e^{ip \cdot x + q \cdot \eps \nabla} \,\mathcal{U}(t;s) \,B_s\,\mathcal{U}^*(t;s)= \tr\, \mathcal{U}^*(t;s)\,e^{ip \cdot x + q \cdot \eps \nabla} \,\mathcal{U}(t;s)\,B_s \,.
\ee
We recall that the kernel of the operator $B_s$ is 
$$
B_s(x;y)= \left[ (V*\wt{\rho}_s)(x) - (V*\wt{\rho}_s)(y) - \nabla(V*\wt{\rho}_s)\left(\frac{x+y}{2}\right)\cdot(x-y) \right] \wt{\omega}_{N,s} (x;y) \,,
$$
Expanding the parenthesis with the potentials in Fourier integrals, we obtain
\[ \Big[  (V*\wt{\rho}_s)(x) - (V*\wt{\rho}_s)(y) - \nabla(V*\wt{\rho}_s)\left(\frac{x+y}{2}\right)\cdot(x-y) \Big] = \int dk \, \widehat{U} (k) \left( e^{i k \cdot x} - e^{ik \cdot y} - e^{i k \cdot \frac{(x+y)}{2}} ik \cdot (x-y) \right) \]
with $U = V* \wt{\rho}_s$. We write
\[ e^{ik \cdot x} - e^{i k \cdot y} = \int_0^1 d\lambda \frac{d}{d\lambda} e^{ik \cdot (\lambda x + (1-\lambda) y)} = \int_0^1 d\lambda e^{ik \cdot (\lambda x + (1-\lambda) y)} ik \cdot (x-y) \]
and hence
\[  \begin{split} e^{ik \cdot x} - e^{i k \cdot y} - e^{i k \cdot \frac{x+y}{2}} ik \cdot (x-y) &= \int_0^1 d\lambda \left[ e^{ik \cdot (\lambda x + (1-\lambda) y)} - e^{ik \cdot \frac{x+y}{2}} \right] ik \cdot (x-y) \\ &= \int_0^1 d\lambda \int_0^1 d\mu \, e^{ik \cdot [\mu (\lambda x + (1-\lambda) y) + (1-\mu)(x+y)/2]} (\lambda -1/2) [k \cdot (x-y)]^2 \end{split} \]
This implies that
\[ B_s = \sum_{i,j =1}^3 \int_0^1 d\lambda \, (\lambda-1/2) \int_0^1 d\mu \int dk \,\widehat{U} (k) k_i k_j \left[ x_i, \left[ x_j, e^{i(\mu\lambda +(1-\mu)/2) k \cdot x} \wt{\omega}_{N,s} e^{i(\mu(1-\lambda)+(1-\mu)/2) k \cdot x} \right] \right]  \]
Therefore, we can bound the absolute value of the second term on the r.h.s. of (\ref{eq:trace}) by
\begin{equation}\label{eq:su3} \begin{split} 
\Big| \tr\, e^{ip \cdot x + q \cdot \eps \nabla} \,&\mathcal{U}(t;s) \,B_s\,\mathcal{U}^*(t;s) \Big| \\ \leq \; &\sum_{i,j=1}^3 \int_0^1 d\lambda |\lambda-1/2| \int_0^1 d\mu \int dk |\widehat{U} (k)| |k|^2 \\ &\hspace{1cm} \times \left| \tr \, \cU^* (t;s)  e^{ip \cdot x + q \cdot \eps \nabla} \mathcal{U}(t;s) [x_i, [x_j, e^{i(\mu\lambda + (1-\mu)/2) k \cdot x}  \wt{\omega}_{N,s}e^{i(\mu(1-\lambda)+(1-\mu)/2) k \cdot x} ]] \right| \\
= \; &\sum_{i,j=1}^3 \int_0^1 d\lambda |\lambda-1/2| \int_0^1 d\mu \int dk |\widehat{U} (k)| |k|^2 \\ &\hspace{1cm} \times \left| \tr \, [x_i, [x_j, \cU^* (t;s)  e^{ip \cdot x + q \cdot \eps \nabla} \mathcal{U}(t;s) ]] \, e^{i(\mu\lambda + (1-\mu)/2) k \cdot x}  \wt{\omega}_{N,s}e^{i(\mu(1-\lambda)+(1-\mu)/2) k \cdot x} \right| \\
\leq \; & C \tr |\wt{\omega}_{N,s}| \, \int dk\, |\widehat{U} (k)| |k|^2 \sup_{\o,i,j} \left| \tr \, [x_i, [x_j, \cU^* (t;s) e^{ip \cdot x + q \cdot \eps \nabla} \cU (t;s)]] \o \right| \,. 
\end{split} 
\end{equation}
The supremum on the r.h.s. is taken over all indices $i,j \in \{1,2,3 \}$ and all trace class operators $\omega$ with $\tr \, |\omega| \leq 1$. This term is controlled thanks to the next lemma, whose proof is deferred to the end of the section.

\begin{lemma}\label{lm:cm2}
Under the same assumptions of Theorem \ref{thm:main2}, there exists $C > 0$ such that 
 \be\label{eq:cm2}
 \sup_{i,j,\o} \left| \tr \, [x_i \,[x_j\,,\,\mathcal{U}^*(t;s)\,e^{ip\cdot x + q \cdot \eps \nabla} \,\mathcal{U}(t;s)\,]]\,\o \right|  \leq C \eps^2 (|p| + |q|)^2 e^{C|t-s|}
 \ee
where the supremum is taken over all $i,j \in \{1,2,3 \}$ and over all trace class operators on $L^2 (\bR^3)$ with $\tr \, |\omega| \leq 1$. 
\end{lemma}

\medskip

Using $|\widehat{U} (k)| \leq |\widehat{V} (k)|$, the assumption (\ref{eq:ass-pot}) and (\ref{eq:cm2}), we conclude that
\begin{equation}\label{eq:gronw2} \begin{split} 
\Big| \tr\, e^{ip \cdot x + q \cdot \eps \nabla} \,&\mathcal{U}(t;s) \,B_s\,\mathcal{U}^*(t;s) \Big| \leq C  \tr \, |\wt{\omega}_{N,s}| \, (|p|+|q|)^2 \eps^2 e^{C|t-s|} 
\end{split}
\end{equation}

Inserting (\ref{eq:gronw1}) and (\ref{eq:gronw2}) on the r.h.s. of (\ref{eq:trace}), we obtain
\begin{equation}\label{eq:grpq} \begin{split} 
\sup_{p,q \in \bR^3} \frac{1}{(|p|+|q|+1)^2} &\left| \tr\, e^{ip \cdot x + q \cdot \eps \nabla} \, (\o_{N,t}-\wt\o_{N,t}) \right| \\ \leq \; &C \int_0^t ds \, \frac{\tr \, |\wt{\omega}_{N,s}|}{N} \, e^{C|t-s|} \sup_{k} \frac{1}{(1+|k|)^2} \left| \tr \, e^{ik \cdot x} (\omega_{N,s} - \wt{\omega}_{N,s}) \right| \\
&+ C \int_0^t ds \, \tr \, |\wt{\omega}_{N,s}|  \eps e^{C|t-s|} \\
\leq \; &C \int_0^t ds \,\frac{\tr \, |\wt{\omega}_{N,s}|}{N} \, e^{C|t-s|} \sup_{p,q} \frac{1}{(1+|p|+|q|)^2} \left| \tr \, e^{ip \cdot x+q \cdot \eps \nabla} (\omega_{N,s} - \wt{\omega}_{N,s}) \right| \\
&+ C \int_0^t ds \, \tr \, |\wt{\omega}_{N,s}| \,  \eps \, e^{C|t-s|}
\end{split} \end{equation}

Now, we estimate the trace norm of $\wt{\omega}_{N,s}$ (here we need the additional regularity of the Wigner transforms of the initial data assumed at the beginning of the proof). We have
\begin{equation}\label{eq:bd-trn} \begin{split} \tr \, |\wt{\omega}_{N,s}| & = \tr \, \left| (1-\eps^2\Delta)^{-1} (1+x^2)^{-1} (1+x^2) (1-\eps^2 \Delta) \wt{\omega}_{N,s} \right| \\
& \leq \, \| (1-\eps^2\Delta)^{-1} (1+x^2)^{-1} \|_\text{HS} \|  (1+x^2) (1-\eps^2 \Delta) \wt{\omega}_{N,s} \|_\text{HS} \\ &\leq C \sqrt{N} \| (1+x^2) (1-\eps^2 \Delta) \wt{\omega}_{N,s} \|_\text{HS} \end{split} \end{equation}
The operator $K = (1+x^2) (1-\eps^2 \Delta) \, \wt{\omega}_{N,s}$ has the integral kernel 
\[ \begin{split} K(x;y) =\; & N (1+x^2) (1-\eps^2 \Delta_x)  \int dv \, \wt{W}_{N,s} \Big( \frac{x+y}{2} , v \Big) e^{-i v \cdot \frac{x-y}{\eps}} \\ 
= \; &N (1+x^2) \int dv \, \wt{W}_{N,s} \Big( \frac{x+y}{2} , v \Big) e^{-i v \cdot \frac{x-y}{\eps}} \\& + N (1+x^2)\int dv \, v^2 \wt{W}_{N,s}  \Big( \frac{x+y}{2} , v \Big) e^{-i v \cdot \frac{x-y}{\eps}} \\
&-\eps^2 N (1+x^2)\int dv \, (\Delta_v \wt{W}_{N,s}) \Big( \frac{x+y}{2} , v \Big) e^{-i v \cdot \frac{x-y}{\eps}} \\ &+ i \eps N (1+x^2) \int dv  \, v \cdot \nabla_v \wt{W}_{N,s} \Big( \frac{x+y}{2} , v \Big) e^{-i v \cdot \frac{x-y}{\eps}} 
\end{split} \]
Writing \[ (1+x^2) = 1+ \left(\frac{x+y}{2}\right)^2 + \left( \frac{x-y}{\eps} \right)^2 \]
we conclude that
\[  \| K \|_\text{HS} \leq C \sqrt{N} \sum_{j=0}^4 \eps^j \| \wt{W}_{N,s} \|_{H^j_4} \]
The propagation of regularity of the Vlasov equation from Proposition \ref{prop:regularity} gives us
\[ \| K \|_\text{HS} \leq C \sqrt{N} e^{C|s|} \sum_{j=0}^4 \eps^j \| W_N \|_{H^j_4} \]
and thus, with (\ref{eq:bd-trn}),  
\[ \tr \, |\wt{\omega}_{N,s}| \leq C N e^{C|s|}\sum_{j=0}^4 \eps^j \| W_N \|_{H^j_4} \]

Inserting in (\ref{eq:grpq}) and applying Gronwall's inequality, we find 
\begin{equation}\label{eq:fin-reg} 
\begin{split}  \sup_{p,q \in \bR^3} \frac{1}{(1+|p|+|q|)^2} & \, \left| \tr \, e^{ip\cdot x + q \cdot \eps \nabla} (\omega_{N,t} - \wt{\omega}_{N,t}) \right| \\ &\leq C \Big[ \sum_{j=0}^4 \eps^j \| W_N \|_{H^j_4} \Big]  N \eps \exp \Big( C \Big[ \sum_{j=0}^4 \eps^j \| W_N \|_{H^j_4} \Big] \exp (C |t|) \Big) \end{split} 
\end{equation}

This completes the proof of Theorem \ref{thm:main2}, under the additional assumption that $\| W_N \|_{H^4_4}$ is bounded. 

\medskip

\noindent{\bf Proof of Theorem \ref{thm:main2}.} We have to relax the condition $\sup\|W_{N}\|_{H^{4}_{4}}<\infty$. To this end, we proceed as follows. We set
\[ W_N^k (x,v) = (W_N * g_k) (x,v) = \int dx' dv' g_k (x-x', v-v') W_N (x',v') \]
with
\[ g_k (x,v) = \frac{k^3}{(2\pi)^3} e^{-\frac{k}{2} (x^2 + v^2)} \]
and we denote by $\omega_N^k$ the Weyl quantization of $W_N^k$. We recall from (\ref{eq:om-fermi}), that 
\begin{equation}\label{eq:om-fermi2} \omega_N^k (x;y) = \frac{1}{(2\pi)^3} \int dw dz \, e^{-z^2/2} e^{-w^2/2} \left[ e^{iw \cdot \frac{x}{\sqrt{k}\eps}} e^{\frac{z}{\sqrt{k}} \cdot \nabla} \omega_N e^{-\frac{z}{\sqrt{k}} \cdot \nabla} e^{-i w \cdot \frac{x}{\sqrt{k} \eps}} \right] (x;y)  \end{equation}
is a fermionic reduced density with $0 \leq \omega_N^k \leq 1$ and $\tr\, \omega_N^k =N$. In fact, (\ref{eq:om-fermi2}), together with the assumption (\ref{eq:sc-in}), also implies that 
\begin{equation}\label{eq:tr-nr} \tr \big| \omega_N - \omega_N^k \big| \leq C \frac{N}{\sqrt{k}} \end{equation}

To see this, we write:
\[
\tr\big| \omega_{N} - \omega_{N}^{k} \big| \leq \frac{1}{(2\pi)^{3}}\int dwdz\, e^{-z^{2}/2}e^{-w^{2}/2} \,\tr \Big| e^{iw\cdot \frac{x}{\sqrt{k}\eps}} e^{\frac{z}{\sqrt{k}}\cdot\nabla} \omega_{N} e^{-\frac{z}{\sqrt{k}}\cdot \nabla}e^{-iw\cdot \frac{x}{\sqrt{k}\eps}} - \omega_{N} \Big|\;,
\]
where
\[
\begin{split}
\tr \Big| e^{iw\cdot \frac{x}{\sqrt{k}\eps}} e^{\frac{z}{\sqrt{k}}\cdot\nabla} \omega_{N} e^{-\frac{z}{\sqrt{k}}\cdot \nabla}e^{-iw\cdot \frac{x}{\sqrt{k}\eps}} - \omega_{N} \Big| &\leq \tr\big| \big[\omega_{N}, e^{-\frac{z}{\sqrt{k}}\cdot \nabla}\big] \big| + \tr\big| e^{iw\cdot \frac{x}{\sqrt{k}\e}}\omega_{N}e^{-iw\cdot \frac{x}{\sqrt{k}\e}} - \omega_{N} \big| \\
&\leq \tr\big| \big[\omega_{N}, e^{-\frac{z}{\sqrt{k}}\cdot \nabla}\big] \big| + \tr\big| \big[\omega_{N}, e^{-iw\cdot \frac{x}{\sqrt{k}\e}}\big] \big| \\
&\leq \frac{|z|}{\sqrt{k}} \tr\big| [\omega_{N}, \nabla] \big| + \frac{|w|}{\sqrt{k}\eps}\tr\big| [\omega_{N}, x] \big|\;;
\end{split}
\]
this estimate together with the assumptions (\ref{eq:sc-in}) implies that:
\[
\tr\big| \omega_{N} - \omega_{N}^{k} \big| \leq \frac{C N}{\sqrt{k}}\frac{1}{(2\pi)^{3}} \int dwdz\, e^{-z^{2}/2}e^{-w^{2}/2} \big[ |z| + |w| \big]\;,
\]
which proves Eq. (\ref{eq:tr-nr}). 

\medskip

We have $\| W_N^k \|_{H_4^j} \leq C k^{j/2}$, for all $j=1,\dots , 4$. Choosing $k = \eps^{-2}$, (\ref{eq:fin-reg}) implies that
\[ \sup_{p,q \in \bR^3} \frac{1}{(1+|p|+|q|)^2} \, \left| \tr \, e^{ip\cdot x + q \cdot \eps \nabla} (\omega^k_{N,t} - \wt{\omega}_{N,t}) \right| \leq C N \eps \exp ( C \exp (C |t|)) \]

On the other hand, proceeding as we did between (\ref{eq:stab-expo}) and (\ref{eq:supptr}) (replacing the observable $e^{ip\cdot x}$ with $e^{ip\cdot x + q \cdot \eps \nabla}$), we obtain
\[ \sup_{p,q} \frac{1}{1+|p|+|q|} \left| \tr \, e^{ip\cdot x + q \cdot \eps \nabla} (\omega_{N,t} - \omega_{N,t}^k)\right| \leq C  \exp (C|t|) \,  \tr |\omega_N - \omega_N^k| \]
With (\ref{eq:tr-nr}), we conclude that (again with the choice $k = \eps^{-2}$) 
\[ \sup_{p,q} \frac{1}{1+|p|+|q|} \left| \tr \, e^{ip\cdot x + q \cdot \eps \nabla} (\omega_{N,t} - \omega_{N,t}^k)\right| \leq C N \eps \exp (C|t|) \] 

Finally, we observe that 
\[ \tr \, e^{ip\cdot x + q \cdot \eps \nabla} (\wt{\omega}_{N,t} - \wt{\omega}_{N,t}^k) = N \left( \widehat{W}_{N,t} (p,q) - \widehat{W}_{N,t}^k (p,q) \right) \]
and therefore we estimate 
\[ \left| \tr \, e^{ip\cdot x + q \cdot \eps \nabla} (\wt{\omega}_{N,t} - \wt{\omega}_{N,t}^k) \right| \leq C N \| \wt{W}_{N,t} - \wt{W}_{N,t}^k \|_1 \]
The $L^1$-stability of the Vlasov equation with respect to perturbation of the initial data has been already established in the proof of Theorem \ref{cor:relaxed-main}. Following the arguments between (\ref{eq:start-L1}) and (\ref{eq:concl-L1}) (using the assumption on the $W^{1,1}$ Sobolev norm of the sequence $W_N$), we obtain 
\[ \| \wt{W}_{N,t} - \wt{W}_{N,t}^k \|_1 \leq C e^{C|t|} \| W_N - W_N^k \|_1 \]
Using again the uniform bound $\| W_N \|_{W^{1,1}} \leq C$, and the choice $k = \eps^{-2}$, we find
\[ \| \wt{W}_{N,t} - \wt{W}_{N,t}^k \|_1 \leq C \eps e^{C|t|} \]
We conclude that 
\[ \sup_{p,q} \frac{1}{(1+|p|+|q|)^2} \left| \tr \, e^{ip\cdot x+ q \cdot \eps \nabla} (\omega_{N,t} - \wt{\omega}_{N,t}) \right| \leq C N \eps \exp (C \exp (C|t|)) \]
for any sequence of initial densities $\omega_N$ satisfying (\ref{eq:sc-in}) and whose Wigner transforms $W_N$ satisfy $\| W_N \|_{W^{1,1}} \leq C$ uniformly in $N$. 
\qed 

\bigskip

\noindent {\bf Proof of Theorem \ref{thm:new2}.} We write 
\[ \left| \widehat{W}_{N,t} (p,q) - \widehat{W}_t (p,q)\right|  \leq \left|  \widehat{W}_{N,t} (p,q) - \widehat{\wt{W}}_{N,t} (p,q)\right| + \left| \widehat{\wt{W}}_{N,t} (p,q) - \widehat{W}_t (p,q)\right| \]
where $\wt{W}_{N,t}$ denotes the solution of the Vlasov equation with initial data $W_N$. {F}rom Theorem \ref{thm:main2}, we know that 
\[  
\left|  \widehat{W}_{N,t} (p,q) - \widehat{\wt{W}}_{N,t} (p,q)\right| \leq C \eps (1+|p| + |q|)^2 e^{C|t|} \]
To conclude the proof of the theorem, we need to compare the solutions $\wt{W}_{N,t}$ and $W_t$ of the Vlasov equation, using the fact that the two initial data are close in $L^1$. As in the approximation argument used in the proof of Theorem \ref{thm:main2}, we make use of the $L^1$-stability of the solution of the Vlasov equation, established in Step 3 of the proof of Theorem \ref{cor:relaxed-main}. Following the arguments between (\ref{eq:start-L1}) and (\ref{eq:concl-L1}), we obtain 
\[ \| \wt{W}_{N,t} - W_t \|_1 \leq C e^{C|t|} \| W_N - W_0 \|_1  \]
where the constant $C > 0$ depends only on $\| W_0 \|_{W^{1,1}}$. This implies that 
\[ \| \widehat{\wt{W}}_{N,t} - \widehat{W}_t \|_\infty \leq \| \wt{W}_{N,t} - W_t \|_1 \leq C \kappa_N e^{C|t|} \]
Hence,
\[ \left| \widehat{W}_{N,t} (p,q) - \widehat{W}_t (p,q)\right| \leq C (1+|p|+|q|)^2 (\eps + \kappa_N) e^{C|t|} \]
which concludes the proof Theorem \ref{thm:new2}. \qed

\section{Proof of auxiliary lemmas}\label{sect:lemmas}

In this section we show Lemma \ref{lm:cm1} and Lemma \ref{lm:cm2}.
\begin{proof}[Proof of Lemma \ref{lm:cm1}]
We define the unitary evolution $\wt{\mathcal{U}}(t;s)$ satisfying 
\begin{equation}
\begin{split}
i\e\partial_t\,\wt{\mathcal{U}}(t;s) &= e^{ir\cdot x}\,h(t)\,e^{-ir\cdot x}\,\wt{\mathcal{U}}(t;s)\\
&=(h(t)+2i\e^2 r\cdot\nabla + r^2\e^2)\,\wt{\mathcal{U}}(t;s)
\end{split}
\end{equation}
We observe that
\begin{equation}\label{eq:sup1} \begin{split} \sup_{\omega} \left| \tr \, \left[ e^{ir \cdot x}, \cU^* (t;s) e^{ix \cdot p + \eps \nabla \cdot q}  \cU (t;s) \right] \omega \right| &= \sup_{\omega} \left| \tr \, \left[ e^{ir \cdot x}, \cU^* (t;s) e^{ix \cdot p + \eps \nabla \cdot q}  \cU (t;s) \right] \cU (s;0) \omega \wt{\cU}^* (s;0) \right|\\ &= \sup_{\omega} \left| \tr \, \wt{\cU}^* (s;0) \left[ e^{ir \cdot x}, \cU^* (t;s) e^{ix \cdot p + \eps \nabla \cdot q}  \cU (t;s) \right] \cU (s;0) \omega \right| \end{split}
\end{equation}
where the supremum is taken over all trace class operators $\omega$ on $L^2 (\bR^3)$ with $\tr \, |\omega| \leq 1$ and where we used the fact that $\tr \, |\cU (s;0) \omega \wt{\cU}^* (s;0)| \leq \tr \, |\omega|$. For a fixed $\omega$ and for fixed $t \in \bR$, we compute now the time-derivative of
\[  \begin{split} i \eps \partial_s \tr \, &\wt{\cU}^* (s;0) \left[ e^{ir \cdot x}, \cU^* (t;s) e^{ix \cdot p + \eps \nabla \cdot q}  \cU (t;s) \right] \cU (s;0) \omega  \\ =\; & - \tr\,\wt{\mathcal{U}}^* (s;0) \, [ h(s),[e^{ir\cdot x} , \mathcal{U}^* (t;s)\, e^{ix \cdot p + \eps \nabla \cdot q} \,\mathcal{U}(t;s) ] ]\,\mathcal{U} (s;0)\o \\
& - 2\e^2\,\tr\,\wt{\mathcal{U}}^* (s;0)\,ir\cdot\nabla\,[e^{ir\cdot x} , \mathcal{U}^* (t;s)\, e^{ix \cdot p + \eps \nabla \cdot q} \,\mathcal{U}(t;s) ]\,\mathcal{U} (s;0)\o\\
& - \e^2\,r^2\,\tr\,\wt{\mathcal{U}}^* (s;0)\,[e^{ir\cdot x} , \mathcal{U}^* (t;s)\, e^{ix \cdot p + \eps \nabla \cdot q} \,\mathcal{U}(t;s) ] \,\mathcal{U} (s;0)\o \\
& + \tr\,\wt{\mathcal{U}}^* (s;0)\,[e^{ir\cdot x} ,\,[h(s) , \mathcal{U}^* (t;s)\, e^{ix \cdot p + \eps \nabla \cdot q} \,\mathcal{U}(t;s) ] ]\,
\mathcal{U} (s;0)\o \end{split}\]
Using the properties of commutators, we find 
\be\label{eq:gronw-step}
\begin{split}
i \eps \partial_s \tr \, &\wt{\cU}^* (s;0) \left[ e^{ir \cdot x}, \cU^* (t;s) e^{ix \cdot p + \eps \nabla \cdot q}  \cU (t;s) \right] \cU (s;0) \omega  \\ = \; & - 2\e^2\,\tr\,\wt{\mathcal{U}}^* (s;0)\,ir\cdot\nabla\,[e^{ir\cdot x} , \mathcal{U}^* (t;s)\,e^{ix \cdot p + \eps \nabla \cdot q} \,\mathcal{U}(t;s) ]\,\mathcal{U} (s;0)\o \\
& - \e^2\,r^2\,\tr\,\wt{\mathcal{U}}^* (s;0)\,[e^{ir\cdot x} , \mathcal{U}^*(t;s)\,e^{ix \cdot p + \eps \nabla \cdot q} \,\mathcal{U}(t;s) ]\,\mathcal{U} (s;0) \o\\
&+ \tr\, \wt{\mathcal{U}}^* (s;0)\,[\mathcal{U}^* (t;s)\,e^{ix \cdot p + \eps \nabla \cdot q} \,\mathcal{U}(t;s) , [ h(s) , e^{ir\cdot x} ] ]\,\mathcal{U} (s;0) \o
\end{split}
\ee
We have 
$$[h(s)\,,\,e^{ir\cdot x}]=(-2i\e^2r\cdot\nabla - \e^2r^2)\,e^{ir\cdot x}\,.$$
Inserting this expression in \eqref{eq:gronw-step}, we get
\be\label{eq:gronw-step2}
\begin{split}
i \eps \partial_s \tr \, \wt{\cU}^* (s;0) &\left[ e^{ir \cdot x}, \cU^* (t;s) e^{ix \cdot p + \eps \nabla \cdot q}  \cU (t;s) \right] \cU (s;0) \omega  \\ = \; & 2\e \tr \, \wt{\mathcal{U}}^* (s;0)\,[\mathcal{U}^* (t;s)\,e^{ix \cdot p + \eps \nabla \cdot q}  \,\mathcal{U}(t;s) , ir\cdot\eps\nabla ]\,e^{ir\cdot x}\,\mathcal{U} (s;0) \o
\end{split}
\ee
Integrating this equation from time $s$ to time $t$, we find
\[ \begin{split} 
\tr \, \wt{\cU}^* (s;0) &\left[ e^{ir \cdot x}, \cU^* (t;s) e^{ix \cdot p + \eps \nabla \cdot q}  \cU (t;s) \right] \cU (s;0) \omega \\ = \;&\tr \, \wt{\cU}^* (t;0) \left[ e^{ir \cdot x}, e^{ix \cdot p + \eps \nabla \cdot q} \right] \cU (t;0) \omega  \\ &+ 2i \int_s^t d\tau \, \tr \, \wt{\mathcal{U}}^* (\tau ;0)\,[\mathcal{U}^* (t;\tau) \,e^{ix \cdot p + \eps \nabla \cdot q}  \,\mathcal{U}(t;\tau) , ir\cdot\eps\nabla ]\,e^{ir\cdot x}\,\mathcal{U} (\tau;0) \o \end{split} \]
which implies that 
\[ \begin{split} 
\Big| \tr \, \wt{\cU}^* (s;0) & \left[ e^{ir \cdot x}, \cU^* (t;s) e^{ix \cdot p + \eps \nabla \cdot q}  \cU (t;s) \right] \cU (s;0) \omega  \Big| \\ &\leq \left| \tr \, [e^{ir \cdot x}, e^{ix\cdot p + \eps \nabla \cdot q}] \cU (t;0) \omega \wt{\cU}^* (t;0) \right| \\ &+ 2 \int_s^t d\tau \, \left| \tr \, \left[ \cU^* (t;\tau) e^{ix \cdot p + \eps \nabla \cdot q} \cU (t;\tau), i r \cdot \eps \nabla \right] e^{ir \cdot x} \cU (\tau;0) \omega\,\wt{\cU}^* (\tau;0) \right| \end{split} \]
Since
 \[ [e^{ir \cdot x}, e^{ix\cdot p + \eps \nabla \cdot q}] =  (e^{-i\eps r\cdot q/2}-e^{i\eps r\cdot q/2}) e^{ix\cdot (p+r) + \eps \nabla \cdot q} \]
we conclude that, for any trace class operator $\omega$ on $L^2 (\bR^3)$, with $\tr \, |\omega| \leq 1$, we have
\[ \begin{split} 
\frac{1}{|r|} \Big| \tr \, &\left[ e^{ir \cdot x}, \cU^* (t;s) e^{ix \cdot p + \eps \nabla \cdot q}  \cU (t;s) \right] \cU (s;0)\, \omega \, \wt{\cU}^* (s;0) \Big| \\ &\leq \eps |q| + 2 \int_s^t d\tau \, \sup_\omega \left| \tr \, \left[ \cU^* (t;\tau) e^{ix \cdot p + \eps \nabla \cdot q} \cU (t;\tau), i \frac{r}{|r|} \cdot \eps \nabla \right] \omega \right| \end{split} \]
where, on the r.h.s., the supremum is taken over all trace class $\omega$ with $\tr \, |\omega| \leq 1$. {F}rom (\ref{eq:sup1}), we obtain 
\begin{equation}\label{eq:gron1} \begin{split} \sup_{\o ,r} \, \Big| \tr \, &\left[ e^{ir \cdot x}, \cU^* (t;s) e^{ix \cdot p + \eps \nabla \cdot q}  \cU (t;s) \right] \omega \Big|
\\ &\leq \eps |q| + 2 \int_s^t d\tau \, \sup_{\o,r} \left| \tr \, \left[  i \frac{r}{|r|} \cdot \eps \nabla, \cU^* (t;\tau) e^{ix \cdot p + \eps \nabla \cdot q} \cU (t;\tau) \right] \omega \right| 
\end{split} 
\end{equation}
Next, we bound the supremum on the r.h.s. of the last equation. To this end, we observe that 
\begin{equation}\label{eq:supk2} \begin{split} \sup_\omega &\left|  \tr \, \left[ i \frac{r}{|r|} \cdot \eps \nabla, \cU^* (t;s) e^{ix \cdot p + \eps \nabla \cdot q} \cU (t;s) \right] \omega \right| \\ & \hspace{2cm} = \sup_\omega \left|  \tr \, \left[ \cU^* (t;s) e^{ix \cdot p + \eps \nabla \cdot q} \cU (t;s), i \frac{r}{|r|} \cdot \eps \nabla \right] \cU (s;0) \omega \cU^* (s;0) \right| \\ & \hspace{2cm} = \sup_\omega \left|  \tr \, \cU^* (s;0) \, \left[ \cU^* (t;s) e^{ix \cdot p + \eps \nabla \cdot q} \cU (t;s), i \frac{r}{|r|} \cdot \eps \nabla \right] \cU (s;0) \omega \right|
\end{split} 
\end{equation}
We compute
\[ \begin{split} 
i\eps \partial_s  \tr \, \cU^* (s;0) \, &\left[ \cU^* (t;s) e^{ix \cdot p + \eps \nabla \cdot q} \cU (t;s), i r \cdot \eps \nabla \right] \cU (s;0) \omega \\ = \; & -\tr\,\mathcal{U}^* (s;0)\,[h(s)\,,\,[\mathcal{U}^* (t;s)\,e^{ix \cdot p + \eps \nabla \cdot q} \,\mathcal{U}(t;s)\, ,\,i\e r\cdot\nabla]]\,\mathcal{U} (s;0) \o \\
&+\tr\,\mathcal{U}^* (s;0)\,[[h(s)\,,\,\mathcal{U}^*(t;s)\,e^{ix \cdot p + \eps \nabla \cdot q}\,\mathcal{U}(t;s)]\,,\,i\e r\cdot\nabla]\,\mathcal{U} (s;0) \o
\end{split} \]
The Jacobi identity implies that
\be
\begin{split}
i\e\partial_s 
\tr \, \cU^* (s;0) \, &\left[ \cU^* (t;s) e^{ix \cdot p + \eps \nabla \cdot q} \cU (t;s), i r \cdot \eps \nabla \right] \cU (s;0) \omega \\ = \; &-\tr\, \mathcal{U}^* (s;0)\,[\mathcal{U}^*(t;s)\, e^{ix \cdot p + \eps \nabla \cdot q} \,\mathcal{U}(t;s)\,,\,[h(s)\,,\,i\e r\cdot\nabla]]\,\mathcal{U} (s;0) \o\,.
\end{split}
\ee
We have
\[ \begin{split} 
[h(s), i r \cdot \eps \nabla ] &= i\eps r \cdot \nabla (V*\rho_s) (x) = i\eps r \cdot \int dk k \widehat{V} (k) \widehat{\rho}_s (k) e^{ik \cdot x} \end{split} \]
Hence
\[ \begin{split} 
i\e\partial_s
\tr \, \cU^* (s;0) \, &\left[ \cU^* (t;s) e^{ix \cdot p + \eps \nabla \cdot q} \cU (t;s), i r \cdot \eps \nabla \right] \cU (s;0) \omega \\ = \; &-i \eps \cdot \int dk r \cdot k \widehat{V} (k) \widehat{\rho}_t (k)\, \tr\, \mathcal{U}^* (s;0)\,[\mathcal{U}^*(t;s)\, e^{ix \cdot p + \eps \nabla \cdot q} \,\mathcal{U}(t;s)\,,\,e^{ik \cdot x}]\,\mathcal{U} (s;0) \o\,.
\end{split} \]
Integrating from time $s$ to time $t$, we find
\[ \begin{split} 
\tr \, \cU^* (s;0) \, &\left[ \cU^* (t;s) e^{ix \cdot p + \eps \nabla \cdot q} \cU (t;s), i r \cdot \eps \nabla \right] \cU (s;0) \omega \\ =\; & \tr \, \cU^* (t;0) \left[ e^{ix \cdot p + \eps \nabla \cdot q}, i r \cdot \eps \nabla \right] \cU (t;0) \omega \\ &+ i \int_s^t d\tau \int dk r \cdot k \widehat{V} (k) \widehat{\rho}_\tau (k)\, \tr\, \mathcal{U}^* (\tau;0)\,[\mathcal{U}^*(t;\tau)\, e^{ix \cdot p + \eps \nabla \cdot q} \,\mathcal{U}(t;\tau)\,,\,e^{ik \cdot x}]\,\mathcal{U} (\tau;0) \o\, \end{split} \]
Since
\[ \left[e^{ix \cdot p + \eps \nabla \cdot q}, i r \cdot \eps \nabla \right] = \eps r \cdot p e^{ix \cdot p + \eps \nabla \cdot q} \] 
we conclude that, for any trace class operator $\omega$ with $\tr \, |\omega| \leq 1$, 
\[  \begin{split} 
\Big| \tr \, \cU^* (s;0) \, &\left[ \cU^* (t;s) e^{ix \cdot p + \eps \nabla \cdot q} \cU (t;s), i \frac{r}{|r|} \cdot \eps \nabla \right] \cU (s;0) \omega \Big| \\ =\; &\eps |p| + \int_s^t d\tau \, \sup_{\omega,k} \frac{1}{|k|} \left| \tr\, \mathcal{U}^* (\tau;0)\,[\mathcal{U}^*(t;\tau)\, e^{ix \cdot p + \eps \nabla \cdot q} \,\mathcal{U}(t;\tau)\,,\,e^{ik \cdot x}]\,\mathcal{U}(\tau;0) \o \right| \, \int dk |k|^2 |\widehat{V} (k)| \end{split} \]
{F}rom (\ref{eq:supk2}), we find 
\[\begin{split}  \sup_{\o, r} \Big| \tr \, &\left[  i \frac{r}{|r|} \cdot \eps \nabla , \cU^* (t;s) e^{ix \cdot p + \eps \nabla \cdot q} \cU (t;s)\right] \omega \Big| \\ &\leq \eps |p| + C \int_s^t d\tau \, \sup_{\o,r} \frac{1}{|r|} \left| \tr\, \mathcal{U}^* (\tau;0)\,[\mathcal{U}^*(t;\tau)\, e^{ix \cdot p + \eps \nabla \cdot q} \,\mathcal{U}(t;\tau)\,,\,e^{ik \cdot x}]\,\mathcal{U} (\tau;0) \o \right| \end{split} \]
Combining this bound with (\ref{eq:gron1}) and applying Gronwall, we obtain
\[ \begin{split} \sup_{\o ,r} \, \Big| \tr \, \left[ e^{ir \cdot x}, \cU^* (t;s) e^{ix \cdot p + \eps \nabla \cdot q}  \cU (t;s) \right] \omega \Big| &+ \sup_{\o, r} \Big| \tr \, \left[  i \frac{r}{|r|} \cdot \eps \nabla , \cU^* (t;s) e^{ix \cdot p + \eps \nabla \cdot q} \cU (t;s)\right] \omega \Big| \\ &\hspace{6cm} \leq \eps N (|p|+|q|) e^{C|t-s|}\end{split}  \]
\end{proof}

\medskip

\begin{proof}[Proof of Lemma \ref{lm:cm2}] We observe, first of all, that
\begin{equation}\label{eq:obs} \begin{split} \sup_\o \Big| \tr \, [x_i \,[x_j \,,\,&\mathcal{U}^*(t;s)\,e^{ip\cdot x + q \cdot \eps \nabla} \,\mathcal{U}(t;s)\,] ]\,\o \Big| \\ & = \sup_\o \left| \tr \, [x_i \,[x_j \,,\,\mathcal{U}^*(t;s)\,e^{ip\cdot x + q \cdot \eps \nabla} \,\mathcal{U}(t;s)\,] ] \, \cU (s;0)\o \cU^* (s;0) \right| \\ &= \sup_\o \left| \tr \, \cU^* (s;0) \, [x_i \,[x_j \,,\,\mathcal{U}^*(t;s)\,e^{ip\cdot x + q \cdot \eps \nabla} \,\mathcal{U}(t;s)\,] ] \, \cU (s;0)\o  \right| \end{split} \end{equation}
We consider now the derivative
\[ \begin{split} i\eps \partial_s \tr \, &\cU^* (s;0) \, [x_i \,[x_j \,,\,\mathcal{U}^*(t;s)\,e^{ip\cdot x + q \cdot \eps \nabla} \,\mathcal{U}(t;s)\,] ] \, \cU (s;0)\o \\ =\; & - \tr \,\cU^* (s;0) \,[ h(s), [x_i \,[x_j \,,\,\mathcal{U}^*(t;s)\,e^{ip\cdot x + q \cdot \eps \nabla} \,\mathcal{U}(t;s)\,] ]] \, \cU (s;0)\o \\ &+ \tr \, \cU^* (s;0) \, [x_i \,[x_j \,,\,[h(s), \mathcal{U}^*(t;s)\,e^{ip\cdot x + q \cdot \eps \nabla} \,\mathcal{U}(t;s)\,] ]] \, \cU (s;0)\o  \end{split} \]
The Jacobi identity implies that
\[ \begin{split} i\eps \partial_s \tr \, &\cU^* (s;0) \, [x_i \,[x_j \,,\,\mathcal{U}^*(t;s)\,e^{ip\cdot x + q \cdot \eps \nabla} \,\mathcal{U}(t;s)\,] ] \, \cU (s;0)\o \\ =\; & \tr \,\cU^* (s;0) \, [x_i \, [ [x_j , h(s)] , \mathcal{U}^*(t;s)\,e^{ip\cdot x + q \cdot \eps \nabla} \,\mathcal{U}(t;s)\,] ] 
 \, \cU (s;0)\o \\ &+ \tr \, \cU^* (s;0) \, [ [x_i, h(s)], \,[x_j ,\mathcal{U}^*(t;s)\,e^{ip\cdot x + q \cdot \eps \nabla} \,\mathcal{U}(t;s)\, ] ]  \, \cU (s;0)\o  \end{split} \]
Since $[x_j, h(s)]= \e^2\nabla_{x_j}$ (and since $[\nabla_{x_j} , x_i] = \delta_{ij}$ is a number), we conclude that
\begin{equation}\label{eq:gro1} \begin{split}
i\eps \partial_s \tr \, &\cU^* (s;0) \, [x_i , \,[x_j \,,\,\mathcal{U}^*(t;s)\,e^{ip\cdot x + q \cdot \eps \nabla} \,\mathcal{U}(t;s)\,] ] \, \cU (s;0)\o \\ =\; & \eps \tr \,\cU^* (s;0) \, [\eps \nabla_{x_j} , \, [ x_i , \mathcal{U}^*(t;s)\,e^{ip\cdot x + q \cdot \eps \nabla} \,\mathcal{U}(t;s)\,] ] 
 \, \cU (s;0)\o \\ &+ \eps \tr \, \cU^* (s;0) \, [ \eps \nabla_{x_i} , \,[x_j ,\mathcal{U}^*(t;s)\,e^{ip\cdot x + q \cdot \eps \nabla} \,\mathcal{U}(t;s)\, ] ]  \, \cU (s;0)\o  \end{split} 
\end{equation}
Integrating over time, we find
\[ \begin{split} \tr \, \cU^* (s;0) [x_i , [ x_j ,\, &\cU^* (t;s) e^{ip \cdot x + q \cdot \eps \nabla} \cU (t;s) ]] \cU(s;0) \o \\ =  \; &\tr \, \cU^* (t;0) [x_i , [ x_j , e^{ip \cdot x + q \cdot \eps \nabla} ]] \cU(t;0) \o \\ &+ i \int_s^t d\tau \tr \, \cU^* (\tau;0) [\eps \nabla_{x_j} , [ x_i , \cU^* (t;\tau) e^{ip \cdot x + q \cdot \eps \nabla} \cU (t;\tau) ]] \cU(\tau;0) \o \\ &+  i \int_s^t d\tau \tr \, \cU^* (\tau;0) [\eps \nabla_{x_i} , [ x_j , \cU^* (t;\tau) e^{ip \cdot x + q \cdot \eps \nabla} \cU (t;\tau) ]] \cU(\tau;0) \o
 \end{split} \]
Since
\[ [x_i , [ x_j, e^{i p\cdot x+ q \cdot \eps \nabla}]] = \eps^2 q_i q_j e^{i p\cdot x + q \cdot \eps \nabla} \] 
we find
\[ \begin{split} &\left|  \tr \, \cU^* (s;0) [x_i , [ x_j , \cU^* (t;s) e^{ip \cdot x + q \cdot \eps \nabla} \cU (t;s) ]] \cU(s;0) \o \right| \\ & \hspace{2cm} \leq \eps^2 |q|^2 + \int_s^t d\tau \, \sup_{\o, i,j} \left| \tr \,  [\eps \nabla_{x_j} , [ x_i , \cU^* (t;\tau) e^{ip \cdot x + q \cdot \eps \nabla} \cU (t;\tau) ]] \o \right| \end{split} \]
for all trace class $\o$ with $\tr \, |\o| \leq 1$. {F}rom (\ref{eq:obs}), we obtain 
\[ \begin{split} \sup_{\o ,i,j} &\left|  \tr \, [x_i , [ x_j , \cU^* (t;s) e^{ip \cdot x + q \cdot \eps \nabla} \cU (t;s) ]] \o \right| \\ &\hspace{2cm} \leq \eps^2 |q|^2 + \int_s^t d\tau \, \sup_{\o, i,j} \left| \tr \,  [\eps \nabla_{x_j} , [ x_i , \cU^* (t;\tau) e^{ip \cdot x + q \cdot \eps \nabla} \cU (t;\tau) ]] \o \right| \end{split} \]
where the suprema are taken over all trace class $\omega$ on $L^2 (\bR^3)$ with $\tr \, |\omega| \leq 1$.  
 
Next, we look for an estimate for 
\[ \begin{split} \sup_{\o,i,j} & \left| \tr \,  [\eps \nabla_{x_j} , [ x_i , \cU^* (t;s) e^{ip \cdot x + q \cdot \eps \nabla} \cU (t;s) ]] \o \right| \\ &\hspace{2cm} = \sup_{\o,i,j}  \left| \tr \, \cU^* (s;0) [\eps \nabla_{x_j} , [ x_i , \cU^* (t;s) e^{ip \cdot x + q \cdot \eps \nabla} \cU (t;s) ]] \cU(s;0) \o \right| \end{split} \]
To this end, we compute the derivative 
\[\begin{split} i\eps \partial_s \tr \, \cU^* (s;0) \, [ \eps \nabla_{x_i} , \,[x_j ,&\mathcal{U}^*(t;s)\,e^{ip\cdot x + q \cdot \eps \nabla} \,\mathcal{U}(t;s)\, ] ]  \, \cU (s;0)\o \\ = \; &- \tr \, \cU^* (s;0) \,[h(s), [ \eps \nabla_{x_i} , \,[x_j , \mathcal{U}^*(t;s)\,e^{ip\cdot x + q \cdot \eps \nabla} \,\mathcal{U}(t;s)\, ] ] ] \, \cU (s;0)\o  \\ &+\tr \, \cU^* (s;0) \, [ \eps \nabla_{x_i} , \,[x_j , [h(s), \mathcal{U}^*(t;s)\,e^{ip\cdot x + q \cdot \eps \nabla} \,\mathcal{U}(t;s)\, ] ] ] \, \cU (s;0)\o \\
= \; & \tr \, \cU^* (s;0) \,[ [\eps \nabla_{x_i}, h(s)] , \,[x_j , \mathcal{U}^*(t;s)\,e^{ip\cdot x + q \cdot \eps \nabla} \,\mathcal{U}(t;s)\, ] ]  \, \cU (s;0)\o  \\ &+\tr \, \cU^* (s;0) \, [ \eps \nabla_{x_i} , \,[ [ x_j , h(s)] , \mathcal{U}^*(t;s)\,e^{ip\cdot x + q \cdot \eps \nabla} \,\mathcal{U}(t;s)\, ] ] \, \cU (s;0)\o \end{split} \]
With \begin{equation}\label{eq:commnab} [\eps \nabla_{x_i} , h(s)] = \eps \nabla_{x_i} (V*\rho_s) (x) = \eps \int dk k_i \widehat{V} (k) \widehat{\rho}_s (k) e^{ik \cdot x} 
\end{equation} we obtain 
\[ \begin{split} 
i\eps \partial_s \tr \, \cU^* (s;0) \, [ \eps \nabla_{x_i} , \,[x_j ,&\mathcal{U}^*(t;s)\,e^{ip\cdot x + q \cdot \eps \nabla} \,\mathcal{U}(t;s)\, ] ]  \, \cU (s;0)\o  \\ 
= \; &\eps \int dk k_i \widehat{V} (k) \widehat{\rho}_s (k) \tr \, \cU^* (s;0) \,[ e^{ik \cdot x} , \,[x_j , \mathcal{U}^*(t;s)\,e^{ip\cdot x + q \cdot \eps \nabla} \,\mathcal{U}(t;s)\, ] ]  \, \cU (s;0)\o  \\ &+\eps \tr \, \cU^* (s;0) \, [ \eps \nabla_{x_i} , \,[ \eps \nabla_{x_i} , \mathcal{U}^*(t;s)\,e^{ip\cdot x + q \cdot \eps \nabla} \,\mathcal{U}(t;s)\, ] ] \, \cU (s;0)\o \end{split} \]
Using the identity
\begin{equation}\label{eq:eikxcom} [e^{ik \cdot x} , A] = e^{ik \cdot x} A - A e^{i k \cdot x} = \int_0^1 d\lambda \, \frac{d}{d\lambda} \, e^{i\lambda k \cdot x} A e^{i(1-\lambda) k \cdot x}  = \int_0^1 d\lambda e^{i\lambda k \cdot x} [ i k \cdot x, A] e^{i (1-\lambda) k \cdot x} 
\end{equation}
we conclude that 
\[ \begin{split} 
i\eps \partial_s &\tr \, \cU^* (s;0) \, [ \eps \nabla_{x_i} , \,[x_j ,\mathcal{U}^*(t;s)\,e^{ip\cdot x + q \cdot \eps \nabla} \,\mathcal{U}(t;s)\, ] ]  \, \cU (s;0)\o  \\ 
= \; & \eps \int_0^1 d\lambda \int dk k_i k_\ell \widehat{V} (k) \widehat{\rho}_s (k) \tr \, \cU^* (s;0) \, e^{i\lambda k \cdot x} [ x_\ell , \,[x_j , \mathcal{U}^*(t;s)\,e^{ip\cdot x + q \cdot \eps \nabla} \,\mathcal{U}(t;s)\, ] ]  \,e^{i(1-\lambda) k \cdot x} \cU (s;0)\o  \\ &+\eps \tr \, \cU^* (s;0) \, [ \eps \nabla_{x_i} , \,[ \eps \nabla_{x_i} , \mathcal{U}^*(t;s)\,e^{ip\cdot x + q \cdot \eps \nabla} \,\mathcal{U}(t;s)\, ] ] \, \cU (s;0)\o \end{split} \]
and hence, after integrating over time,
\[ \begin{split} 
\Big| \tr \, \cU^* (s;0) \, &[ \eps \nabla_{x_i} , \,[x_j ,\mathcal{U}^*(t;s)\,e^{ip\cdot x + q \cdot \eps \nabla} \,\mathcal{U}(t;s)\, ] ]  \, \cU (s;0)\o \Big| \\ 
\leq \; & \Big| \tr \, \cU^* (t;0) \, [ \eps \nabla_{x_i} , \,[x_j ,\,e^{ip\cdot x + q \cdot \eps \nabla} ] ]  \, \cU (t;0)\o \Big|
\\ &+ \int_s^t d\tau \, \int_0^1 d\lambda \int dk |k|^2 |\widehat{V} (k)|\\ &\hspace{2cm} \times \left| \tr \, \cU^* (\tau;0) \, e^{i\lambda k \cdot x} [ x_\ell , \,[x_j , \mathcal{U}^*(t;\tau)\,e^{ip\cdot x + q \cdot \eps \nabla} \,\mathcal{U}(t;\tau)\, ] ]  \,e^{i(1-\lambda) k \cdot x} \cU (\tau;0)\o \right| 
\\ &+\int_s^t d\tau \, \Big|\tr \, \cU^* (\tau;0) \, [ \eps \nabla_{x_i} , \,[ \eps \nabla_{x_j} , \mathcal{U}^*(t;\tau)\,e^{ip\cdot x + q \cdot \eps \nabla} \,\mathcal{U}(t;\tau)\, ] ] \, \cU (\tau;0) \o \Big| \end{split} \]
Since \[ [\eps \nabla_{x_i}, [x_j, e^{ip\cdot x+ q \eps \cdot \nabla}]] = -i\eps^2 p_i q_j e^{ip\cdot x+ q \eps \cdot \nabla} \]
this implies that
\begin{equation}\label{eq:gro2} \begin{split}  \sup_{\o, i,j} \Big| \tr \,  &[ \eps \nabla_{x_i} , \,[x_j ,\mathcal{U}^*(t;s)\,e^{ip\cdot x + q \cdot \eps \nabla} \,\mathcal{U}(t;s)\, ] ] \o \Big| \\ 
\leq \; &\eps^2 |p||q| + C \int_s^t d\tau \,\sup_{\o,i,j} \left| \tr \, [ x_i , \,[x_j , \mathcal{U}^*(t;\tau)\,e^{ip\cdot x + q \cdot \eps \nabla} \,\mathcal{U}(t;\tau)\, ] ] \o \right| 
\\ &+\int_s^t d\tau \, \sup_{\o,i,j} \Big|\tr \, [ \eps \nabla_{x_i} , \,[ \eps \nabla_{x_j} , \mathcal{U}^*(t;\tau)\,e^{ip\cdot x + q \cdot \eps \nabla} \,\mathcal{U}(t;\tau)\, ] ] \, \o \Big| \end{split} 
\end{equation}

Finally, we need an estimate for 
\[ \begin{split} 
\sup_{\o,i,j} \Big|\tr \, &[ \eps \nabla_{x_i} , \,[ \eps \nabla_{x_j} , \mathcal{U}^*(t;s)\,e^{ip\cdot x + q \cdot \eps \nabla} \,\mathcal{U}(t;s)\, ] ] \, \o \Big| \\ &= \sup_{\o,i,j} \Big|\tr \, \cU^* (s,0) [ \eps \nabla_{x_i} , \,[ \eps \nabla_{x_j} , \mathcal{U}^*(t;s)\,e^{ip\cdot x + q \cdot \eps \nabla} \,\mathcal{U}(t;s)\, ] ] \cU (s;0) \, \o \Big| \end{split} \]
Hence, we compute the derivative
\[\begin{split}  i\eps \partial_s \tr \, \cU^* (s;0) &[ \eps \nabla_{x_i} , \,[ \eps \nabla_{x_j} , \mathcal{U}^*(t;s)\,e^{ip\cdot x + q \cdot \eps \nabla} \,\mathcal{U}(t;s)\, ] ] \cU (s;0) \, \o \\ = \; &- \tr \, \cU^* (s;0) [ h(s), [ \eps \nabla_{x_i} , \,[ \eps \nabla_{x_j} , \mathcal{U}^*(t;s)\,e^{ip\cdot x + q \cdot \eps \nabla} \,\mathcal{U}(t;s)\, ] ] ] \cU (s;0) \, \o \\ &+ \tr \, \cU^* (s;0) [ \eps \nabla_{x_i} , \,[ \eps \nabla_{x_j} , [h(s), \mathcal{U}^*(t;s)\,e^{ip\cdot x + q \cdot \eps \nabla} \,\mathcal{U}(t;s)\, ] ] ] \cU (s;0) \, \o \\
= \; &\tr \, \cU^* (s;0) [ [ \eps \nabla_{x_i}, h(s)], [\eps \nabla_{x_j}, \mathcal{U}^*(t;s)\,e^{ip\cdot x + q \cdot \eps \nabla} \,\mathcal{U}(t;s)]] \cU(s;0) \omega \\ &+ \tr \, \cU^* (s;0) [ \eps \nabla_{x_i}, [[\eps \nabla_{x_j}, h(s)], \mathcal{U}^*(t;s)\,e^{ip\cdot x + q \cdot \eps \nabla} \,\mathcal{U}(t;s)]] \cU(s;0) \omega \end{split} \]
{F}rom (\ref{eq:commnab}), we find 
\begin{equation}\label{eq:2termsgro3} \begin{split} 
 i\eps \partial_s \tr \, \cU^* (s;0) &[ \eps \nabla_{x_i} , \,[ \eps \nabla_{x_j} , \mathcal{U}^*(t;s)\,e^{ip\cdot x + q \cdot \eps \nabla} \,\mathcal{U}(t;s)\, ] ] \cU (s;0) \, \o \\
 = \; &\eps \int dk k_i \widehat{V} (k) \widehat{\rho}_s (k) \, \tr \, \cU^* (s;0) [ e^{ik \cdot x}, [\eps \nabla_{x_j}, \mathcal{U}^*(t;s)\,e^{ip\cdot x + q \cdot \eps \nabla} \,\mathcal{U}(t;s)]] \cU(s;0) \omega \\ &+ \eps \int dk k_j \widehat{V} (k) \widehat{\rho}_s (k) \, \tr \, \cU^* (s;0) [ \eps \nabla_{x_i}, [e^{ik \cdot x} , \mathcal{U}^*(t;s)\,e^{ip\cdot x + q \cdot \eps \nabla} \,\mathcal{U}(t;s)]] \cU(s;0) \omega \end{split}
\end{equation}
In the first term on the r.h.s. of the last equation we use (\ref{eq:eikxcom}). In the second term, on the other hand, we notice that
\[ \begin{split} 
\tr \, \cU^* (s;0) [ \eps \nabla_{x_i}, &[e^{ik \cdot x} , \mathcal{U}^*(t;s)\,e^{ip\cdot x + q \cdot \eps \nabla} \,\mathcal{U}(t;s)]] \cU(s;0) \omega \\ = \; &\tr \, \cU^* (s;0) [ e^{ik \cdot x} , [\eps \nabla_{x_i},  \mathcal{U}^*(t;s)\,e^{ip\cdot x + q \cdot \eps \nabla} \,\mathcal{U}(t;s)]] \cU(s;0) \omega 
\\ &+ \tr \, \cU^* (s;0) [ [\eps \nabla_{x_i}, e^{ik \cdot x} ] , \mathcal{U}^*(t;s)\,e^{ip\cdot x + q \cdot \eps \nabla} \,\mathcal{U}(t;s)] \cU(s;0) \omega
\end{split} \]
Again, the first term on the r.h.s. of the last equation can be handled with (\ref{eq:eikxcom}). As for the second term, we use that $[\eps \nabla_{x_i}, e^{ik \cdot x} ] = i \eps k_i e^{ik \cdot x}$. Integrating (\ref{eq:2termsgro3}) over time, we find
\[  \begin{split} \Big| \tr \, \cU^* (s;0) &[ \eps \nabla_{x_i} , \,[ \eps \nabla_{x_j} , \mathcal{U}^*(t;s)\,e^{ip\cdot x + q \cdot \eps \nabla} \,\mathcal{U}(t;s)\, ] ] \cU (s;0) \, \o \Big| \\ \leq \;& \left| \tr \, \cU^* (t;0) [ \eps \nabla_{x_i} , \,[ \eps \nabla_{x_j} , \,e^{ip\cdot x + q \cdot \eps \nabla} \, ] ] \cU (t;0) \, \o \right| \\ &+ C \int_s^t d\tau \sup_{\o,i,j} \left| \tr \, [ \eps \nabla_{x_i}, [x_j , \mathcal{U}^*(t;\tau)\,e^{ip\cdot x + q \cdot \eps \nabla} \,\mathcal{U}(t;\tau)]]  \omega \right| \\ &+ C \eps \left[ \int dk \, |\widehat{V} (k)| |\widehat{\rho}_s (k)| |k|^3 dk \right] \left[ \int_s^t d\tau \, \sup_{\o, k} \frac{1}{|k|} \left| \tr \, [ e^{ik\cdot x},  \mathcal{U}^*(t;\tau)\,e^{ip\cdot x + q \cdot \eps \nabla} \,\mathcal{U}(t;\tau)] \omega \right| \right]  \end{split} \]
To bound the integral involving the potential in the last term on the r.h.s. of the last equation, we use (\ref{eq:ass-pot}) with $\| \widehat{\rho}_s \|_\infty \leq 1$. 
{F}rom 
\[ [ \eps \nabla_{x_i} , \,[ \eps \nabla_{x_j} , \,e^{ip\cdot x + q \cdot \eps \nabla} \, ] = -\eps^2 p_i p_j e^{ip\cdot x + q \cdot \eps \nabla} \]
and from Lemma \ref{lm:cm1}, we obtain
\begin{equation}\label{eq:gro3}  \begin{split} \Big| \tr \, \cU^* (s;0) &[ \eps \nabla_{x_i} , \,[ \eps \nabla_{x_j} , \mathcal{U}^*(t;s)\,e^{ip\cdot x + q \cdot \eps \nabla} \,\mathcal{U}(t;s)\, ] ] \cU (s;0) \, \o \Big|\\ \leq \;&C \eps^2 + C \int_s^t d\tau \sup_{\o,i,j} \left| \tr \, \, [ \eps \nabla_{x_i}, [x_j , \mathcal{U}^*(t;\tau)\,e^{ip\cdot x + q \cdot \eps \nabla} \,\mathcal{U}(t;\tau)]]  \omega \right| 
\end{split} 
\end{equation}
Combining (\ref{eq:gro1}) and (\ref{eq:gro2}) with the last equation and applying Gronwall lemma, we deduce that
\[ \begin{split}  \sup_{i,j,\o} \Big| \tr \, \, [ x_i , \,[ x_j , \mathcal{U}^*(t;s)\,e^{ip\cdot x + q \cdot \eps \nabla} \,\mathcal{U}(t;s)\, ] ] \, \o \Big| & \leq C\eps^2  (|p|+ |q|)^2 e^{C|t-s|} \\
\sup_{i,j,\o} \Big| \tr \, \, [ \eps \nabla_{x_i} , \,[ x_j , \mathcal{U}^*(t;s)\,e^{ip\cdot x + q \cdot \eps \nabla} \,\mathcal{U}(t;s)\, ] ]\, \o \Big| & \leq C\eps^2  (|p|+ |q|)^2 e^{C|t-s|}\\
\sup_{i,j,\o} \Big| \tr \, [ \eps \nabla_{x_i} , \,[ \eps \nabla_{x_j} , \mathcal{U}^*(t;s)\,e^{ip\cdot x + q \cdot \eps \nabla} \,\mathcal{U}(t;s)\, ] ] \, \o \Big| &\leq C\eps^2  (|p|+ |q|)^2 e^{C|t-s|}\end{split} \]
\end{proof}


\appendix

\section{Well-posedness of the Vlasov equation for signed measures}
\label{app:dobr}

The goal of this appendix is to show that the arguments of \cite{Dobr79} can be extended to prove the well-posedness of the Vlasov equation (\ref{eq:vlasov}) for initial data given by signed measures. 

Following the notation of \cite{Dobr79}, let $\cM$ denote the space of all finite signed measures on the Borel $\sigma$-algebra $\cB (\bR^6)$. For an open interval $\Delta \subset \bR$, we denote by $\cM_\Delta$ the set of all families $M = \{ \mu_t \}_{t \in \Delta}$, with $\mu_t \in \cM$ for all $t \in \Delta$ such that, for all bounded intervals $\Delta' \subset \Delta$, there exists $C_{\Delta'} > 0$ with $\sup_{t \in \Delta'} \| \mu_t \| < C_{\Delta'}$, and such that the function 
\[ (\nabla V * \mu_t) \, (x) = \int \nabla V (x-x') \, d\mu_t (x',v') \]
is continuous in $t \in \Delta$, for all $x \in \bR^3$ ($V$ denotes the interaction potential entering the Vlasov equation (\ref{eq:vlasov})). For $C > 0$, we also denote by $\cM_\Delta (\kappa)$ the set of families $M = \{ \mu_t \}_{t \in \Delta}$ with $\| \mu_t \| = \sup_{B \in \cB (\bR^6)} |\mu (B)| = \kappa$ for all $t \in \Delta$.

Defining $A,B : \bR^3 \times \bR^3 \to \bR^3 \times \bR^3$ by $A(x,v) = (2v,0)$ and $B(x,v) = (0, \nabla V (x))$ and, for every $\mu \in \cM$, 
\[ B_\mu (x,v) = \int B(x-x' , v-v') d\mu (x',v') \]
we say that a family $M = \{ \mu_t \}_{t \in \Delta} \in \cM_\Delta$ is a weak solution of the Vlasov equation on the interval $\Delta$ if, for every test function $h \in \cD (\bR^6)$ in the Schwarz space, 
\[ \mu_t (h) = \int h(x,v) d\mu_t (x,v) \]
is differentiable in $t \in \Delta$ and 
\[ \frac{d}{dt} \mu_t (h) = \mu_t ((A+B_{\mu_t}) \nabla h) \]
It is easy to check that, if the weak solution $\mu_t$ has a density $W_t (x,v)$, differentiable in $t$, then 
$W_t$ is a solution of the standard Vlasov equation (\ref{eq:vlasov}). 

\begin{proposition}\label{prop:dobr}
Let $V \in C^2_b (\bR^3)$. For any finite signed measure $\mu^0 \in \cM$, and every open interval $\Delta \subset \bR$ with $0 \in \Delta$, there exists a unique weak solution $M = \{ \mu_t \}_{t \in \Delta}$ of the Vlasov equation on $\Delta$ with $\mu_{t= 0} = \mu^0$.
\end{proposition}

\begin{proof}
We follow the strategy of \cite{Dobr79}, adapting it to the case of signed $\mu^0$. We will use the variable $z = (x,v) \in \bR^6$. For $M = \{ \mu_t \}_{t\in \Delta} \in \cM_\Delta$, we define \begin{equation}\label{eq:GM} G_M (t,z) = A(z) + B_{\mu_t} (z) 
\end{equation} 
and we consider the solution of Newton's equation
\begin{equation}\label{eq:newt} 
\frac{d}{dt} z(t) = G_M (t, z(t)) 
\end{equation}
We denote by $z_M (t,u)$ the solution of (\ref{eq:newt}), with initial data $z_M (0,u) = u$. 

For a fixed $\mu^0 \in \cM$, we define the map $T : \cM_\Delta \to \cM_\Delta$ by
\[ (TM)_t (E) = \mu^0 \left( \left\{ u \in \bR^6 : z_M (t,u) \in E \right\} \right) \]
for all $E \in \cB (\bR^6)$. As in \cite{Dobr79}, it is easy to check that $M \in \cM_\Delta$ is a weak solution of the Vlasov equation with initial data $\mu^0$ if and only if $M$ is a fixed point of $T$, i.e. if $TM=M$. 

Hence, to prove Proposition \ref{prop:dobr}, it is enough to show that $T$ is a contraction on $\cM_\Delta$. In fact, since clearly $TM \in \cM_\Delta (\| \mu^0 \|)$, for all $M \in \cM_\Delta$, it is enough to show that the restriction of $T$ to $\cM_\Delta (\| \mu^0 \|)$ is a contraction, with respect to an appropriate metric, that we are now going to define.   

For two signed measures $\mu, \mu' \in \cM$, we define
\begin{equation}\label{eq:def-d} \wt{d} (\mu, \mu') = d_{KR} (\mu_+ , \mu'_+) + d_{KR} (\mu_-, \mu'_-) \end{equation}
where $\mu = \mu_+ - \mu_-$ is the Jordan decomposition of $\mu$ in its positive and negative parts and where $d_{KR}$ is the Kantorovich-Rubinshtein metric, defined by
\[ d_{KR} (\nu, \nu') = \inf_{m \in N(\nu,\nu')} \int \rho (z_1, z_2) dm (z_1 ,z_2) \]
where $\rho (z_1, z_2) = \min (|z_1 - z_2|,1)$ and $N(\nu,\nu')$ is the space of all positive measures $m$ on $\cB (\bR^{12})$ such that $m (E \times \bR^6) = \nu (E)$ and $m (\bR^6 \times E) = \nu' (E)$ for all $E \in \cB (\bR^6)$. Furthermore, for $M = \{ \mu_t \}_{t \in  \Delta}, M' = \{ \mu'_t \}_{t \in \Delta} \in \cM_\Delta$, we define
\[ d (M,M') = \int_{\Delta} \wt{d} (\mu_t , \mu'_t) \]
It is easy to check that (\ref{eq:def-d}) defines a metric on $\cM (\Delta)$. 

We claim that, for $|\Delta|$ small enough, 
\begin{equation}\label{eq:contr} 
d (TM, TM') \leq \frac{1}{2} d (M,M') 
\end{equation}
for all $M,M' \in \cM_\Delta ( \| \mu^0 \| )$. To prove (\ref{eq:contr}) we observe that, for all $\mu,\mu' \in \cM$,
\[ \begin{split}
B_\mu (z) - B_{\mu'} (z) = \; & \int B(z-w) d\mu (w) - \int B(z-w) d\mu' (w) \\ 
\leq \; & \int ( B(z-w_1) - B (z - w_2)) dm_+ (w_1, w_2) \\ &- \int  ( B(z-w_1) - B (z - w_2)) dm_- (w_1, w_2) 
\end{split} \]
for any $m_+ \in N (\mu_+ , \mu'_+), m_- \in N (\mu_-, \mu'_-)$. Recalling that $B (x,v) = (0, \nabla V (x))$ and the assumption $V \in C^2_b (\bR^3)$, we find
\[ \begin{split} 
|B_\mu (z) - B_{\mu'} (z)| \leq \; &  \int |B(z-w_1) - B (z - w_2)| dm_+ (w_1, w_2) \\ &+ \int  |B(z-w_1) - B (z - w_2)| dm_- (w_1, w_2) \\ \leq \; & C \int \rho (w_1, w_2) dm_+ (w_1, w_2) + \int \rho (w_1, w_2) dm_- (w_1, w_2)
\end{split} \]
Since the inequality holds for every $m_+ \in N (\mu_+, \mu'_+)$ and $m_- \in N (\mu_-, \mu'_-)$, we conclude that
\begin{equation}\label{eq:B-B'} |B_\mu (z) - B_{\mu'} (z)| \leq C \wt{d} (\mu,\mu') 
\end{equation}
for all $\mu, \mu' \in \cM$ and all $z \in \bR^6$. 

Furthermore, recalling the definition (\ref{eq:GM}), we observe that there is a constant $C$, depending on $\| \mu^0 \|$, such that 
\begin{equation}\label{eq:GM-bd} 
| G_M (z) - G_M (z') | \leq C |z-z'| 
\end{equation}
for all $z,z' \in \bR^6$ and for all $M \in \cM_\Delta (\| \mu^0 \|)$. 

For $M,M' \in \cM_\Delta (\| \mu^0 \|)$ and $u \in \bR^6$, we define the quantity
\[ \alpha (M,M', u) = \sup_{t\in \Delta} |z_M (t,u) - z_{M'} (t,u)| \]
With (\ref{eq:GM-bd}), we obtain 
\[ \begin{split} 
|z_M (t,u) - z_{M'} (t,u)| \leq \; &\int_0^t |G_M (s, z_M (s,u) - G_{M'} (s, z_{M'} (s,u))| ds \\ \leq \; & \int_\Delta |G_M (s, z_M (s,u) - G_{M} (s, z_{M'} (s,u))| ds \\ &+ \int_\Delta G_M (s, z_{M'} (s,u) - G_{M'} (s, z_{M'} (s,u))| ds \end{split} \]
for all $t \in \Delta$. Combining (\ref{eq:B-B'}) and (\ref{eq:GM-bd}), we find
\[  |z_M (t,u) - z_{M'} (t,u)| \leq C \int_\Delta |z_M (s,u) - z_{M'} (s,u)| ds + C d (M,M') \]
Taking the supremum over $t$, we conclude that, for sufficiently small $|\Delta|$, 
\[ \alpha (M,M',u) \leq \frac{C}{1-C|\Delta|} d (M,M') \]

We are now ready to bound $d (TM, TM')$. For $M,M' \in \cM_\Delta (\| \mu^0 \|)$, we notice that
\[ (TM_t)_\pm (E) = \mu^0_\pm \left( \{ u \in \bR^6 : z_M (t,u) \in E \} \right)  \]
for every $E \in \cB (\bR^6)$. Now, let 
\[ m_{t,\pm} (F) = \mu_\pm^0 \left( \{ u \in \bR^6 : ( z_M (t,u) , z_{M'} (t,u)) \in F \} \right) \]
for every $F \in \cB (\bR^{12})$. Then we have $m_{t,\pm} \in N ( (TM_t)_\pm , (TM'_t)_\pm )$ and 
\[ \begin{split} \int \rho (z_1, z_2) dm_{t, \pm} (z_1, z_2) &= \int \rho (z_M (t,u) , z_{M'} (t,u)) d \mu^0_\pm (u) \\ &\leq \int \alpha (M,M',u) d\mu_\pm^0 (u) \\ &\leq \frac{C \| \mu_\pm^0 \|}{1-C|\Delta|} d(M,M') \end{split} \]
This implies that 
\[ d_{KR} ((TM_t)_\pm , (TM'_t)_\pm) \leq \frac{C \| \mu_\pm^0 \|}{1-C|\Delta|} d(M,M') \]
and therefore that 
\[ d (TM, TM') \leq \frac{C |\Delta|\| \mu^0 \|}{1-C|\Delta|} d (M,M') \]
Hence, for $|\Delta|$ sufficiently small, we obtain  (\ref{eq:contr}) for all $M,M' \in \cM_\Delta (\| \mu^0 \|)$. This proves that $T$ defines a contraction on $\cM_\Delta (\| \mu^0 \|)$ and implies the existence and the uniqueness of a weak solution of the Vlasov equation, for $|\Delta|$ sufficiently small. The argument can then be iterated to obtain existence and uniqueness for all times. 
\end{proof}

\section{Regularity estimates for solutions of the Vlasov equation}
\label{app:prop}
\setcounter{equation}{0}

In the next proposition we estimate the weighted Sobolev norms $\| W_t \|_{H^k_2}$ of the solution at time $t$ of the Vlasov equation in terms of their value at $t=0$. 

\begin{proposition}\label{prop:regularity}
Assume that
\begin{equation}\label{eq:ass-pot2} \int dp \, |\widehat{V} (p)| (1+|p|^2) < \infty 
\end{equation}
Let $W_t$ be the solution of the Vlasov equation (\ref{eq:vlasov}) with initial data $W_0$. For $k=1,2,3,4,5$, there exists a constant $C > 0$, which depends on $\|W_{0}\|_{H^{2}_{4}}$ but not on the higher Sobolev norms, such that 
\be\label{eq:Hk4regularity} 
\| W_t \|_{H^k_4} \leq C e^{C|t|} \| W_0 \|_{H^k_4}
\ee
\end{proposition}
\begin{proof}
We use a standard argument. We denote by $\Phi_{t}(x,v) := (X_t (x,v), V_t (x,v))$ the solution of Newton's equations
\[ \begin{split} \dot{X}_t (x,v) &= 2V_t (x,v) \\
\dot{V}_t (x,v) &= - \nabla \left( V* \wt{\rho}_t \right) (X_t (x,v)) \end{split} \]
with initial data $X_0 (x,v) = x$ and $V_0 (x,v) = v$.
Here $\wt{\rho}_t (x) = \int dv\, W_t (x,v)$. We can rewrite Newton's equation in integral form
\begin{equation}\label{eq:XVint} \begin{split} X_t (x,v) &= x + 2\int_0^t ds \, V_s (x,v) \\ 
V_t (x,v) &= v - \int_0^t \nabla (V * \rho_s) (X_s (x,v)) \end{split} 
\end{equation}
In the following, it will be convenient to introduce the following shorthand notation:
\be\label{eq:Phi}
\begin{split}
\| X^{(j)}_{t} \|_{\infty} &:= \max_{|\alpha| = j}\big\| \nabla^{\alpha} X_{t}(x,v) \big\|_{\infty} \\
\| V^{(j)}_{t} \|_{\infty} &:= \max_{|\alpha| = j}\big\| \nabla^{\alpha} V_{t}(x,v) \big\|_{\infty} \\
\| \Phi^{(j)}_{t} \|_{\infty} &:= \| X^{(j)}_{t} \|_{\infty} + \| V^{(j)}_{t} \|_{\infty}\;.
\end{split}
\ee
In general, to control $\|W_{t}\|_{H^{k}_{s}}$, it is sufficient to control $\| \Phi^{(j)}_{t} \|_{\infty}$ for $j\leq k$. In fact, it is not difficult to see that:
\be\label{eq:compositeder}
\begin{split}
\| W_{t} \|_{H^{k}_{4}}^{2} &= \sum_{|\alpha| = k}\int dxdv\, (1 + x^{2} + v^{2})^{4} \big| \nabla^{\alpha} W_{0}\big(X_{-t}(x,v),\, V_{-t}(x,v) \big) \big|^{2}\\
&\leq C\sum_{|\beta|\leq k}\sum_{\substack{\alpha_{1},\ldots,\, \alpha_{|\beta|} \\ \alpha_{1}+\ldots + \alpha_{|\beta|} = k,\, |\alpha_{i}|\geq 1}}\int dxdv\, (1 + x^{2} + v^{2})^{4} \big| \big(\nabla^{\beta} W_{0} \big) \big(X_{-t}(x,v),\, V_{-t}(x,v) \big) \big|^{2} \\
&\quad\times \big| \nabla^{\alpha_{1}} \big( X_{-t}(x,v),\, V_{-t}(x,v) \big) \big|^{2}\cdots \big| \nabla^{\alpha_{|\beta|}} \big( X_{-t}(x,v),\, V_{-t}(x,v) \big) \big|^{2}\\
&\leq C\sum_{n=1}^{k}\sum_{\substack{m_{1},\ldots,\, m_{n} \\ m_{1}+\ldots + m_{n} = k,\, m_{i} \geq 1}} \| W_{0} \|_{H^n_4}^{2} \| \Phi^{(m_{1})}_{-t} \|_{\infty}^{2} \cdots \| \Phi_{-t}^{(m_{n})} \|_{\infty}^{2}\;;
\end{split}
\ee
to get the last step we performed a change of variables and we used that, by Gronwall's lemma together with (\ref{eq:XVint}) and $\| \nabla V \|_\infty <\infty$:
\be\label{eq:propweight} 
1+ X^2_t (x,v)+ V^2_t (x,v) \leq C e^{C|t|} (1+x^2+v^2)
\ee
We start by estimating $\| W_t \|_{H^1_4}$. To this end, we need to control $\|\Phi^{(1)}_{t}\|_{\infty}$. For any multi-index $\alpha$ with $|\alpha|=1$, we obtain from (\ref{eq:XVint}) that 
\[ \begin{split}  \| \nabla^\alpha X_t \|_\infty &\leq 1 + 2\int_0^t ds \, \| \nabla^\alpha V_s \|_\infty \\ 
\| \nabla^\alpha V_t \|_\infty &\leq 1 + \int_0^t ds\, \| \nabla^2 (V*\wt{\rho}_s) \circ X_s \cdot \nabla^\alpha X_s \|_\infty \\ &\leq 1 + C \int_0^t ds\,  \| \nabla^\alpha X_s \|_\infty \end{split} \]
where we used that $\| \nabla^{2} (V * \wt{\rho}_{s}) \|_{\infty} \leq \| \nabla^{2} V \|_{\infty} \| \wt{\rho}_{s} \|_{1}$, and $\| \wt{\rho}_{s} \|_{1}\leq \| W_{0} \|_{1} \leq C\| W_{0} \|_{H^{0}_{4}}$ (see (\ref{eq:L1toH14})). Gronwall's lemma, together with the assumption $\|\nabla^{2}V\|_{\infty}<\infty$, implies that
\begin{equation}\label{eq:alpha=1} 
\| \Phi^{(1)}_{t} \|_{\infty} \leq C e ^{C|t|} 
\end{equation}
where the constant $C$ depends on $\|W_{0}\|_{H^{0}_{4}}$, but not on the higher Sobolev norms. Thanks to (\ref{eq:compositeder}), the bound (\ref{eq:alpha=1}) immediately implies 
\be\label{eq:Wt1-bd}
\| W_t \|_{H^1_4}^2 \leq C e^{C|t|} \| W_0 \|_{H^1_4}^2
\ee
where the constant $C$ depends on $\|W_{0}\|_{H^{0}_{4}}$, but not on the higher Sobolev norms. This concludes the proof of (\ref{eq:Hk4regularity}) with $k=1$. Next, let $k=2$. As before, we start by considering the derivatives $\nabla^\alpha X_t, \nabla^\alpha V_t$, now for $|\alpha|=2$. We have:
\be
\begin{split}
\| \nabla^{\alpha} (\nabla (V*\wt{\rho}_s) \circ X_s) \|_\infty &\leq \| \nabla^3 (V*\wt{\rho}_s) \|_\infty \| X^{(1)}_s \|_\infty^{2} + \| \nabla^2 (V*\wt{\rho}_s) \|_\infty \| X^{(2)}_{s} \|_\infty \\
&\leq Ce^{C|s|}\|W_{0}\|_{H^{1}_{4}} + \| \nabla^2 (V*\wt{\rho}_s) \|_\infty \| X^{(2)}_s \|_\infty
\end{split}
\ee
where in the last step we used $\| \nabla^3 (V*\wt{\rho}_s) \|_\infty\leq \|\nabla^{2}V\|_{\infty}\|\nabla \wt{\rho}_{s}\|_{1}\leq Ce^{C|s|}\|W_{0}\|_{H^{1}_{4}}$, and we estimated $\| \nabla\wt{\rho}_{s} \|_{1}\leq C\| W_{s} \|_{H^{1}_4}\leq Ce^{C|s|}\| W_{0} \|_{H^{1}_{4}}$. This, together with the estimate (\ref{eq:alpha=1}), implies:
\[ 
\begin{split} 
\| X^{(2)}_t \|_{\infty}  &\leq 2\int_0^t ds \, \| V^{(2)}_s \|_{\infty} \\ 
\| V^{(2)}_t \|_{\infty} & \leq C\int_0^t ds \| X^{(2)}_s \|_\infty + C e^{C|t|}\| W_{0} \|_{H^{1}_{4}} 
\end{split} 
\]
thus, by Gronwall's lemma: 
\be\label{eq:step2} 
\| \Phi^{(2)}_{t} \|_{\infty} \leq C e ^{C|t|} \| W_{0} \|_{H^{1}_{4}}
\ee
Therefore, proceeding as in (\ref{eq:Wt1-bd}), we get 
\be\label{eq:H24}
\| W_t \|_{H^2_4} \leq C e^{C|t|}\|W_{0}\|_{H^{2}_{4}}
\ee
where the constant $C>0$ is allowed to depend on $\|W_{0}\|_{H^{1}_{4}}$, but not on the higher Sobolev norms. This concludes the proof of (\ref{eq:Hk4regularity}) for $k=2$. Consider now $k=3,4,5$. We will use that, for $|\alpha| = k$:
\be\label{eq:L1toHk-1}
\begin{split}
\| \nabla^{\alpha}\nabla (V*\wt{\rho}_s) \|_{\infty} \leq &\; C\| \nabla^{2}  V \|_{\infty} \sum_{|\beta| = k-1} \| \nabla^{\beta} \tilde\rho_{s} \|_{1} \\
\leq&\; C \| \nabla^{2} V \|_{\infty} \|W_{s}\|_{H^{k-1}_{4}}\\
\end{split}
\ee
and
\be\label{eq:dercomposta}
\begin{split}
\| \nabla^{\alpha}(\nabla (V*\wt{\rho}_s)\circ X_{s}) \|_{\infty} \leq C\sum_{|\beta|\leq |\alpha|}\sum_{\substack{\alpha_{1},\ldots,\, \alpha_{|\beta|} \\ \alpha_{1} + \ldots + \alpha_{|\beta|} = k}}\big\| \nabla^{\beta}\nabla (V*\wt\rho_{s}) \big\|_{\infty} \| X^{(|\alpha_{1}|)}_{s} \|_{\infty}\cdots \| X^{(|\alpha_{|\beta|}|)}_{s} \|_{\infty}
\end{split}
\ee
for a $k$-dependent constant $C>0$. Let $k= 3$. We have, for $|\alpha| = 3$:
\begin{equation}\label{eq:nabalpha3} 
\begin{split} 
&\| \nabla^{\alpha}( \nabla (V*\wt{\rho}_s) \circ X_s) \|_\infty \\
&\leq C\Big[\| \nabla^4 (V*\wt{\rho}_s) \|_\infty  \| X^{(1)}_{s} \|_{\infty}^{3} + \| \nabla^3 (V*\wt{\rho}_s) \|_\infty \| X^{(2)}_{s} \|_{\infty} \| X^{(1)}_{s} \|_{\infty} + \| \nabla^2 (V*\wt{\rho}_s) \|_\infty \| X^{(3)}_{s} \|_{\infty} \Big] \\
&\leq C e^{C|s|}\| W_{0} \|_{H^{2}_{4}} + C \| X^{(3)}_s \|_\infty 
\end{split} 
\end{equation}
where the constant $C>0$ is allowed to depend on $\|W_{0}\|_{H^{1}_{4}}$, but not on the higher Sobolev norms. The last step follows from (\ref{eq:L1toHk-1}) and from the previous estimates on $\|W_{s}\|_{H^{j}_{4}}$, $\| X^{(j)}_{s} \|_{\infty}$, $j=1,2$. Plugging this bound in (\ref{eq:XVint}), we find
\be\label{eq:step3} 
\| \Phi^{(3)}_{t} \|_{\infty} \leq C e ^{C|t|} \| W_{0} \|_{H^{2}_{4}}
\ee
where the constant $C>0$ is allowed to depend on $\|W_{0}\|_{H^{1}_{4}}$, but not on the higher Sobolev norms. Thus, proceeding as in (\ref{eq:Wt1-bd}):
\be\label{eq:H34}
\| W_t \|_{H^3_4} \leq C e^{C|t|} \| W_0 \|_{H^3_4}
\ee
where the constant $C>0$ is allowed to depend on $\|W_{0}\|_{H^{2}_{4}}$, but not on the higher Sobolev norms. This concludes the proof of (\ref{eq:Hk4regularity}) for $k=3$. Let $k= 4$. Similarly to (\ref{eq:nabalpha3}), using (\ref{eq:L1toHk-1}) together with the estimates for $\| W_{s} \|_{H^{j}_{4}}$, $j=1,2,3$, we find, for $|\alpha| = 4$:
\[ 
\begin{split} 
&\| \nabla^{\alpha} (\nabla (V*\wt{\rho}_s) \circ X_s )\|_\infty \\
&\leq C\Big[\| \nabla^5 (V*\wt{\rho}_s) \|_\infty  \| X^{(1)}_s \|^4_\infty  + \| \nabla^4 (V*\wt{\rho}_s) \|_\infty \| X^{(2)}_s \|_\infty \| X^{(1)}_s \|_\infty^{2} \\
&\quad + \| \nabla^3 (V*\wt{\rho}_s) \|_\infty \Big( \| X^{(3)}_s \|_\infty \| X^{(1)}_s \|_\infty + \| X^{(2)}_s \|^2_\infty \Big) + \| \nabla^2 (V*\wt{\rho}_s) \|_\infty \| X^{(4)}_s \|_\infty\Big] \\ 
&\leq  C e^{C|s|}\| W_{0} \|_{H^{3}_{4}} + C \| X^{(4)}_s \|_\infty  
\end{split} 
\]
where the constant $C>0$ is allowed to depend on $\|W_{0}\|_{H^{2}_{4}}$, but not on the higher Sobolev norms. This implies
\be\label{eq:step4} 
\| \Phi^{(4)}_{t} \|_\infty \leq C e ^{C|t|}\| W_{0} \|_{H^{3}_{4}}
\ee
where the constant $C>0$ is allowed to depend on $\|W_{0}\|_{H^{2}_{4}}$, but not on the higher Sobolev norms. Then, we claim that:
\be\label{eq:endstep4} 
\| W_t \|_{H^4_4} \leq C e^{C|t|} \| W_0 \|_{H^4_4}
\ee
where the constant $C>0$ is allowed to depend on $\|W_{0}\|_{H^{2}_{4}}$, but not on the higher Sobolev norms. In fact, from (\ref{eq:compositeder}) we get and from the previous estimates on $\| \Phi^{(j)}_{t} \|_{\infty}$, $j\leq 4$, we have:
\be\label{eq:compositeder2}
\| W_{t} \|_{H^{4}_{4}}^{2} \leq Ce^{C|t|}\Big[ \| W_{0} \|_{H^{1}_{4}}^{2} \| \Phi^{(4)}_{0} \|_{\infty}^{2} + \sum_{k=2}^{4}\|W_{0}\|^{2}_{H^{k}_{4}} \| W_{0} \|_{H^{2}_{4}}^{2k} \Big]\\
\ee
This, together with (\ref{eq:step4}), implies (\ref{eq:endstep4}) and concludes the proof of (\ref{eq:Hk4regularity}) for $k=4$. The case $k=5$ can be studied in a similar way. Let $|\alpha| = 5$. Using once more (\ref{eq:L1toHk-1}), (\ref{eq:dercomposta}), and proceeding as for the previous cases, we get:
\begin{equation}\label{eq:nabalpha5} \begin{split} \| \nabla^{\alpha}( \nabla (V*\wt{\rho}_s) \circ X_s) \|_\infty \leq C e^{C|t|} \| W_0 \|_{H^4_4} + C \| X^{(5)}_s \|_\infty \end{split} 
\end{equation}
where the constant $C>0$ is allowed to depend on $\|W_{0}\|_{H^{2}_{4}}$, but not on the higher Sobolev norms. 
By Gronwall's lemma, we get:
\be\label{eq:step5} 
\| \Phi^{(5)}_{t}\|_\infty \leq C e ^{C|t|} \| W_0 \|_{H^4_4}
\ee
where the constant $C>0$ is allowed to depend on $\|W_{0}\|_{H^{2}_{4}}$, but not on the higher Sobolev norms. Then, we claim that:
\be\label{eq:endstep5}
\| W_t \|_{H^5_4} \leq C e^{C|t|} \| W_0 \|_{H^5_4}
\ee
where the constant $C>0$ is allowed to depend on $\|W_{0}\|_{H^{2}_{4}}$, but not on the higher Sobolev norms. To see this, we use again (\ref{eq:compositeder}). We get:
\be
\begin{split}
\| W_{t} \|_{H^{5}_{4}}^{2} &\leq Ce^{C|t|} \Big[ \|W_{0}\|_{H^{1}_{4}}^{2}\| \Phi^{(5)}_{0} \|_{H^{4}_{4}}^{2} + \| W_{0} \|_{H^{2}_{4}}^{2}\| \Big( \| \Phi^{(4)}_{0} \|^{2}_{\infty}\| \Phi^{(1)}_{0} \|^{2} + \| \Phi^{(3)}_{0} \|^{2}_{\infty}\| \Phi^{(2)}_{0} \|^{2}_{\infty}\Big) \\
&\quad + \sum_{k=3}^{5} \|W_{0}\|_{H^{k}_{4}}^{2}\| W_{0}\|_{H^{2}_{4}}^{2k}\Big]
\end{split}
\ee
which, together with (\ref{eq:step3}), (\ref{eq:step4}), (\ref{eq:step5}), implies (\ref{eq:endstep5}). This concludes the proof of (\ref{eq:Hk4regularity}) for $k=5$, and of Proposition \ref{prop:regularity}.
\end{proof} 
 

\section{Propagation of commutator bounds along the Hartree dynamics}
\label{sect:comm}
\setcounter{equation}{0}

Bounds for norms of commutators of the form $[x,\omega_{N,t}]$ and $[\eps \nabla, \omega_{N,t}]$ play an important role in our analysis. In this section, 
we show how they can be propagated along the Hartree evolution. Similar bounds have been proven in \cite{BPS}. 
\begin{proposition}\label{prop:prop-comm}
Assume 
\begin{equation}\label{eq:pot-ass3} \int |\widehat{V} (p)| (1+|p|^2) dp < \infty 
\end{equation}
Let $\omega_{N,t}$ be the solution of the nonlinear Hartree equation 
\[ i\eps \partial_t \omega_{N,t} = \left[ -\eps^2 \Delta + (V* \rho_t) , \omega_{N,t} \right] \]
with initial data $\omega_{N,t=0}= \omega_N$. Then there exists a constant $C > 0$ such that 
\[ \begin{split} \| [x, \omega_{N,t} ] \|_\text{HS}  
&\leq C e^{C|t|} \left[ \| [x, \omega_{N} ] \|_\text{HS} + \| [\eps \nabla, \omega_{N} ] \|_\text{HS}  \right] 
\\ \| [\eps \nabla, \omega_{N,t} ] \|_\text{HS}  &\leq C e^{C|t|} \left[ \| [x, \omega_{N} ] \|_\text{HS} + \| [\eps \nabla, \omega_{N} ] \|_\text{HS} \right] \end{split} \]
Moreover, 
\[ \begin{split} \| [x, \omega_{N,t} ] \|_\text{tr}  
&\leq C e^{C|t|} \left[ \| [x, \omega_{N} ] \|_\text{tr} + \| [\eps \nabla, \omega_{N} ] \|_\text{tr} \right]  \\ \| [\eps \nabla, \omega_{N,t} ] \|_\text{tr}  &\leq C e^{C|t|} \left[ \| [x, \omega_{N} ] \|_\text{tr} + \| [\eps \nabla, \omega_{N} ] \|_\text{tr} \right]
\end{split} \]
\end{proposition}

\begin{proof}
Let $h_H (t) = -\eps^2 \Delta + (V*\rho_t) (x)$ and $\cU (t;s)$ be the unitary evolution generated by $h_H (t)$, as defined in (\ref{eq:aux-dyn}). We compute
\[ \begin{split} 
i\eps \partial_t \cU^* (t;0) [ x, \omega_{N,t}] \cU (t;0) = \; & - \cU^* (t;0) [h_H (t), [x,\omega_{N,t}]] \cU (t;0) + \cU^* (t;0) [x, [h_H(t), \omega_{N,t}]] \cU (t;0) \\ = \; &\cU^* (t;0) [[h(t),x], \omega_{N,t}] \cU (t;0) \\ = \; & \eps \, \cU^* (t;0) [ \eps \nabla , \omega_{N,t}] \cU (t;0)  \end{split} \]
Integrating over time, we find
\[ [x,\omega_{N,t}] = \cU (t;0) [x, \omega_N] \cU^* (t;0) + i \int_0^t ds \, \cU (t;s) [\eps \nabla , \omega_{N,s} ] \cU^* (t;s) \]
and thus
\begin{equation}\label{eq:comm-bdHS} \| [x, \omega_{N,t} ] \|_\text{HS} \leq \| [x,\omega_N] \|_\text{HS} + \int_0^t ds \, \| [\eps \nabla , \omega_{N,s} ] \|_\text{HS} 
\end{equation}
On the other hand,
\[ \begin{split} 
i\eps \partial_t \cU^* (t;0) [\eps \nabla , \omega_{N,t}]  \cU (t;0) = \; &- \cU^* (t;0) [ h_H (t), [\eps \nabla, \omega_{N,t}]] \cU (t;0) + \cU^* (t;0) [\eps \nabla, [ h_H (t) , \omega_{N,t}]] \cU (t;0) 
\\ = \; &\cU^* (t;0) [ [\eps \nabla, h_H (t)], \omega_{N,t}] \cU (t;0) \\ = \; &\eps \, \cU^* (t;0) [ \nabla (V*\rho_t) , \omega_{N,t} ] \cU (t;0) \\ = \; & \eps \int dp \, p\, \widehat{V} (p) \, \widehat{\rho}_t (p) \, \cU^* (t;0) [e^{ip \cdot x}, \omega_{N,t} ] \, \cU (t;0) \end{split} \] 
Using the identity \[ [e^{ip \cdot x}, \omega_{N,t} ] = \int_0^1 d\lambda \,  e^{i\lambda p \cdot x} [ip \cdot x,  \omega_{N,t} ] e^{i(1-\lambda) p \cdot x} \]
we obtain, with (\ref{eq:pot-ass3}), 
\begin{equation} \begin{split} 
\|  [\eps \nabla , \omega_{N,t}] \|_\text{\text{HS}} \leq \; &\| [\eps \nabla, \omega_N ] \|_\text{HS} +   \int dp |\widehat{V} (p)| |p|^2 |\widehat{\rho}_t (p)| \int_0^t ds \| [x,\omega_{N,s}] \|_\text{HS} \\ \leq \; &
\| [\eps \nabla, \omega_N ] \|_\text{HS} + C \int_0^t ds \, \| [x,\omega_{N,s}] \|_\text{HS}
\end{split} 
\end{equation}
Combining the last equation with (\ref{eq:comm-bdHS}) and applying Gronwall's lemma, we find
\[ \left[ \| [x, \omega_{N,t} ] \|_\text{HS} + \|  [\eps \nabla , \omega_{N,t}] \|_\text{\text{HS}} \right] \leq C e^{C|t|} \left[ \| [x, \omega_{N} ] \|_\text{HS} + \|  [\eps \nabla , \omega_{N}] \|_\text{\text{HS}} \right] \]
as claimed. In the same way, one can also prove the estimates for the trace norms of the commutators.
\end{proof}

{\it Acknowledgments.} The work of M. Porta, C. Saffirio and B. Schlein has been partially supported by the ERC grant MAQD 240518, the NCCR SwissMAP and the SNF Project ``Effective equations from quantum dynamics''. N. Benedikter has been partially supported by the ERC grant CoMBoS-239694 and by the ERC Advanced grant 321029. The authors declare that there is no conflict of interest. The authors would also like to thank the Erwin Schr\"{o}dinger Institute in Vienna for the hospitality; part of this paper has been written during the conference ``Scaling limits and effective theories in classical and quantum mechanics''.

\bigskip

\adresse

 \end{document}